%% file: main.tex
\documentclass[a4paper,11pt]{article}
\usepackage{a4wide}
\usepackage[svgnames,table]{xcolor}
\usepackage{hyperref}
\hypersetup{colorlinks=true,citecolor=RoyalBlue,linkcolor=RoyalBlue,urlcolor=RoyalBlue}
\usepackage{amsfonts,amsmath,amsthm,amssymb}
\usepackage{mathrsfs}  
\usepackage{cancel}
\usepackage{mathtools}
\usepackage{bm}
\usepackage{enumitem}
\usepackage{thm-restate} 
\usepackage{pgfplots}
\usepackage{verbatim} 
\usepackage{nicematrix,tikz}
\usepackage{tikz-cd}
\usepackage[capitalise]{cleveref} 

\newtheorem{theorem}{Theorem}[section]
\newtheorem{proposition}[theorem]{Proposition}
\newtheorem{lemma}[theorem]{Lemma}

\usepackage{todonotes}

\theoremstyle{definition}
\newtheorem{definition}[theorem]{Definition}
\newtheorem{example}[theorem]{Example}

\input{macros}

\fussy

\begin{document}

\author{Libor Barto\\
Charles University
\and
Silvia Butti\\
University of Oxford
\and
Alexandr Kazda\\
Charles University
\and
Caterina Viola\\
University of Catania
\and
Stanislav \v{Z}ivn\'y\\
University of Oxford
}

\title{The Rise of Plurimorphisms:\\ Algebraic Approach to Approximation\thanks{An extended abstract of this work appeared in the Proceedings of LICS 2024~\cite{Barto24:lics}. Libor Barto and Caterina Viola were funded by the European Union (ERC, CoCoSym, 771005). Libor Barto was also funded by (ERC, POCOCOP, 101071674). Caterina Viola was also funded by ICSC – Centro Nazionale di Ricerca in High-Performance Computing, Big Data and Quantum Computing, co-founded by European Union - NextGenerationEU. Views and opinions expressed are however those of the authors only and do not necessarily reflect those of the European Union or the European Research Council Executive Agency. Neither the European Union nor the granting authority can be held responsible for them. 
This work was supported by UKRI EP/X024431/1. For the purpose of Open Access, the authors have applied a CC BY public copyright licence to any Author Accepted Manuscript version arising from this submission. All data is provided in full in the results section of this paper.}}
\date{}
\maketitle

\begin{abstract}
\input{abstract}
\end{abstract}

\input{paper.tex}

{
\small
\bibliographystyle{plainurl}
\bibliography{bib}
}

\appendix

\input{appendix.tex}

\end{document}

%% file: macros.tex
\newcommand{\arxiv}[1]{{#1}}
\newcommand{\conference}[1]{\ignore{#1}}

\newcommand{\yes}{\textsf{yes} }
\newcommand{\no}{\textsf{no} }

\newcommand{\id}{\textnormal{id}}
\DeclareRobustCommand{\rchi}{{\mathpalette\irchi\relax}} 
\newcommand{\irchi}[2]{\raisebox{\depth}{$#1\chi$}} 

\newcommand{\caln}{\mathcal{N}}

\newcommand{\fold}[1]{\overline{#1}}

\newcommand{\sort}{\textnormal{sort}}

\newcommand{\sig}{\sigma}
\newcommand{\typ}{\tau}

\newcommand{\pcsp}{\textnormal{PCSP}}

\newcommand{\ex}[1]{\mathbb{E}_{#1}}
\newcommand{\Ex}[1]{\mathop{\mathbb{E}}_{#1}}

\newcommand{\bfa}{\mathbf{A}}
\newcommand{\bfb}{\mathbf{B}}

\newcommand{\ab}{(\bfa,\bfb)}
\newcommand{\abp}{(\bfa',\bfb')}
\newcommand{\aaa}{(\bfa,\bfa)}
\newcommand{\aap}{(\bfa',\bfa')}
\newcommand{\abt}{(\bfa,\bfb, \IO)}
\newcommand{\abtp}{(\bfa',\bfb', \IO')}
\newcommand{\abcs}{(\bfa,\bfb,c,s)}
\newcommand{\abcsp}{(\bfa',\bfb',c',s')}
\newcommand{\abzz}{(\bfa,\bfb,0,0)}
\newcommand{\abzzp}{(\bfa',\bfb',0,0)}
\newcommand{\ba}{\mathbf{a}}

\newcommand{\ar}{\textnormal{ar}}

\newcommand{\supp}{\mathrm{Supp}}

\newcommand{\mnn}[1]{\mathscr{#1}}

\newcommand{\feas}{\textnormal{feas}}

\newcommand{\qq}{\mathbb{Q}}
\newcommand{\rr}{\mathbb{R}}

\newcommand{\lin}{\textnormal{3LIN2}}
\newcommand{\klin}[1]{#1\textnormal{LIN2}}
\newcommand{\glc}{\textnormal{GLC}}
\newcommand{\lc}{\textnormal{LC}}

\newcommand{\PCSP}{\textnormal{PCSP}}

\newcommand{\tuple}[1]{\mathbf{#1}}
\newcommand{\cols}{\mathrm{cols}}
\newcommand{\rows}{\mathrm{rows}}
\newcommand{\col}[2]{\mathrm{col}_{#1}(#2)}
\newcommand{\row}[2]{\mathrm{row}_{#1}(#2)}

\newcommand{\Mat}[2]{\mathrm{Mat}(#1,#2)}
\newcommand{\proj}[2]{\mathrm{proj}^{#1}_{#2}}
\newcommand{\pol}{\mathrm{Pol}}
\newcommand{\FinSet}{\mathrm{FinSet}}
\newcommand{\ii}{\mathrm{in}}
\newcommand{\oo}{\mathrm{out}}
\newcommand{\OmegaI}{\Omega^{\mathrm{in}}}
\newcommand{\OmegaO}{\Omega^{\mathrm{out}}}
\newcommand{\OmegaIp}{{\Omega'}^{\mathrm{in}}}
\newcommand{\OmegaOp}{{\Omega'}^{\mathrm{out}}}
\newcommand{\weight}{w}
\newcommand{\qqnonneg}{\qq^{+}_0}
\newcommand{\qqpos}{\qq^{+}}
\newcommand{\qqinf}{\overline{\qq}}
\newcommand{\mc}{\mathrm{MC}}

\newcommand{\IO}{\kappa}
\newcommand{\IOI}{\theta^{\mathrm{in}}}
\newcommand{\IOO}{\theta^{\mathrm{out}}}
\newcommand{\shiftI}{\delta^{\mathrm{in}}}
\newcommand{\shiftO}{\delta^{\mathrm{out}}}
\newcommand{\shiftIP}{\delta^{\mathrm{in}+}}
\newcommand{\shiftOP}{\delta^{\mathrm{out}+}}
\newcommand{\shiftIN}{\delta^{\mathrm{in}-}}
\newcommand{\shiftON}{\delta^{\mathrm{out}-}}
\newcommand{\scale}{\gamma}
\newcommand{\polfeas}{\mathrm{PolFeas}}
\newcommand{\multipol}{\mathrm{Plu}}
\newcommand{\iopol}{\IO\mbox{-}\mathrm{Pol}}
\newcommand{\vmnn}[1]{\mathbb{#1}}
\newcommand{\vmc}{\mathrm{VMC}}
\newcommand{\incl}{\mathrm{incl}}
\newcommand{\GM}[1]{\mathrm{CM}[#1]}

\newcommand{\seltup}[1]{#1^{*}}
\DeclareMathOperator{\symdif}{\Delta}
\newcommand{\oddim}[1]{#1^2}

\newcommand{\ignore}[1]{}

%% file: abstract.tex
Following the success of the so-called algebraic approach to the study of decision constraint satisfaction problems (CSPs), exact optimization of valued CSPs, and most recently promise CSPs, we propose an algebraic framework for valued promise CSPs. 

To every valued promise CSP we associate an algebraic object, its so-called valued minion. Our main result shows that the existence of a homomorphism between the associated valued minions implies a polynomial-time reduction between the original CSPs. We also show that this general reduction theorem includes important inapproximability results, for instance, the inapproximability of almost solvable systems of linear equations beyond the random assignment threshold. 

%% file: paper.tex
\section{Introduction}

What mathematical structure captures efficient computation? Answering this
question is the holy grail of theoretical computer science. 
Constraint Satisfaction Problems, or CSPs for short, provide an excellent
framework
to attempt this ambitious research endeavour. On the one hand,
CSPs are general enough to include many fundamental problems of interest and
allow for general patterns to occur, which rarely happens when studying concrete
problems in isolation. On the other hand,  CSPs are structured enough so that
interesting and nontrivial results can be established. Indeed, while CSPs do not capture\footnote{Up to polynomial-time Turing reductions, CSPs on infinite domains \emph{do} capture all computational problems~\cite{Bodirsky08:icalp}, cf. also~\cite{Gillibert22:sicomp}.} all computational problems, both algorithmic and hardness techniques developed 
in the context of constraint satisfaction
are often used beyond the realm of CSPs.

Putting the area of constraint solving aside, there are two main strands of
research on the computational complexity of CSPs. The first strand studies decision CSPs on
finite~\cite{Feder98:monotone}
and infinite~\cite{bodirsky2021complexity} domains, exact solvability of optimization CSPs (known as valued
CSPs~\cite{DBLP:journals/siamcomp/CohenCCJZ13}), and most recently qualitative approximation of decision CSPs (known as
promise CSPs, or PCSPs for short~\cite{AGH17,BG21,BartoBKO21}). The highlights of this strand include, firstly, complexity classifications of 
CSPs, e.g.,
dichotomies for robust solvability of CSPs~\cite{Barto16:sicomp},
valued CSPs~\cite{tz16:jacm-complexity,KozikO15,Kolmogorov17:sicomp},
infinite-domain CSPs~\cite{Bodirsky10:jacm,Bodirsky15:jacm,BMM18}, promise CSPs~\cite{BG21,Ficak19:icalp},
and in particular a dichotomy for all finite-domain CSPs~\cite{Bulatov17:focs,Zhuk20:jacm}, which 
gave a positive answer to the long-standing Feder-Vardi conjecture~\cite{Feder98:monotone}.
Secondly, characterizations of the power of various algorithms,
e.g.,~\cite{Idziak10:siam,Berman10:few,Barto14:jacm,ktz15:sicomp,tz17:sicomp,cz23sicomp:clap,Bodirsky23:jacm}.

The second strand studies quantitative 
approximation of CSPs. The
highlights include, e.g., the PCP
theorem~\cite{Arora98:jacm-probabilistic,Arora98:jacm-proof,Dinur07:jacm},
H{\aa}stad's optimal inapproximability
results~\cite{DBLP:journals/jacm/Hastad01}, Raghavendra's result that a 
semidefinite programming relaxation is optimal for all CSPs~\cite{Raghavendra08:everycsp} under Khot's Unique Games
Conjecture~\cite{Khot02stoc}, inapproximability of certain valued CSPs (under the UGC)~\cite{KhotR08}, 
optimal inapproximability of certain MaxCSPs~\cite{Chan15:jacm},
or the recent line of work on inapproximability of perfectly satisfiable  MaxCSPs~\cite{Bhangale22:stoc,Bhangale23:stocII,Bhangale23:stocIII}.

While the two strands use different mathematical tools (algebraic vs. analytical), there are some common
features, e.g., dictatorship testing 
plays an important role 
in both PCSPs and approximability. Our
paper confirms that this is not a coincidence.

With the general goal to better 
understand what makes computational problems easy
or hard, we aim to provide uniform descriptions of algorithms, 
tractability boundaries, and reductions. For
CSPs studied in the first strand described above, all of these can be described uniformly by means of \emph{polymorphisms}, which can be,
informally, thought of as multivariate 
symmetries of solutions spaces of CSPs
(although the precise definitions and conditions depend on the type of
considered CSPs). 
Interestingly, it was observed a posteriori in~\cite{Brown-CohenRaghavendra15} that Raghavendra's
result from~\cite{Raghavendra08:everycsp}, which falls in the second strand, can be phrased in terms of (a certain
type of) polymorphisms, although it remained unclear whether and how polymorphisms determine complexity without the Unique Games Conjecture.
The notion of polymorphisms coming from~\cite{Brown-CohenRaghavendra15} is close to ours, cf. the discussion after~\cref{prop:alt-pol}.

\smallskip
In the present paper, we introduce and initiate the study of the very general framework of valued PCSPs.
It includes, as special cases, (non-valued) PCSPs (and thus also CSPs),
valued CSPs, approximation of CSPs (both constant factor and gap variants),
Gap Label Cover, and 
Unique Games.
The only previous works on valued PCSPs are the algorithmic results in~\cite{ViolaZ21,BartoButti2022} 
and the unpublished manuscript of Kazda developing an algebraic theory for constant factor approximation of valued PCSPs~\cite{Kazda21}, cf.~\cref{app:kazda}.

As our main result, we define a notion of polymorphisms for valued PCSPs and show that it leads to polynomial-time reductions. Thus, we take the first step in
providing a uniform description of reductions 
among the very large class of computational problems captured by valued PCSPs.

In order to help the reader and to  explain clearly the differences between the previous work on PCSPs and our more general setting of valued PCSPs, we recap in~\cref{sec:CrispPCSP} the basics of the algebraic theory for non-valued PCSPs. This should be useful in particular as we  work in the multi-sorted setting\footnote{Different variables can have different domains.} (cf. the discussion at the end of~\cref{subsec:PCSP}) and with slightly more general notions than is common in the literature.
In \cref{sec:ValuedPCSP} we define valued PCSPs and \emph{plurimorphisms}, the new notion of multivariate symmetry. 
Then, in~\cref{sec:valued-homos}, we prove our first main result, namely that a homomorphism between sets of plurimorphisms of two valued PCSPs implies a polynomial-time reduction between these valued PCSPs.  
The core of this reduction theorem is that  every valued PCSP is polynomial-time
equivalent to a valued version of the Minor Condition problem that played a
key role in the algebraic approach to non-valued PCSPs~\cite{BartoBKO21}. This
allows us to circumvent routes via definability as was done in the original
algebraic approach to decision non-promise CSPs and valued
CSPs~\cite{Bulatov05:classifying,DBLP:journals/siamcomp/CohenCCJZ13}. 
Finally, in~\cref{sec:examples} we give  examples of valued homomorphisms, most notably our second main result, which is a valued homomorphism that captures, e.g., H{\aa}stad's result on inapproximability of almost-satisfiable systems of linear equations~\cite{DBLP:journals/jacm/Hastad01}.

\section{Promise CSP} \label{sec:CrispPCSP}

In this section we review the basics of the theory of \emph{crisp} (non-valued) Promise CSPs, in a way that mimics our theory-building in the more general valued setting in the next section. The definitions and theorems 
essentially follow parts of~\cite{BartoBKO21} with some adjustments.%
\footnote{%
The adjustment is mostly in that the arity of a relation or a function can be any finite set $N$. It is more standard in the CSP literature to only use $n$-ary relations and functions for a non-negative integer $n$. We do not see any advantages for the latter (at least in our context) and a lot of disadvantages, such as the need to often choose enumerations, awkward expressions, unnecessary notions, unnecessary abusing notation, void calculations, etc.}

\subsection{Preliminaries}

For two sets $A$ and $Z$, the set $A^Z$ is the set of all functions from $Z$ to $A$. Sometimes it is more natural to regard elements $f \in A^Z$ as tuples of elements of $A$ indexed by elements of $Z$ (we also say a $Z$-tuple of elements of $A$). In such a case we use boldface and write, e.g., $\tuple{a} \in A^Z$. However, there are situations when both viewpoints (as a function or as a tuple) are used within one formula or a proof.

A $Z$-ary \emph{relation on $A$} is a subset $\phi$ of $A^Z$. For $\tuple{a} \in
A^Z$ and a $Z$-ary relation $\phi$, we usually write $\phi(\tuple{a})$ instead
of $\tuple{a} \in \phi$. 
In order to succinctly write down a $Z$-tuple, one can fix a linear order on $Z$ and write a tuple as a sequence of length $|Z|$, e.g., $\tuple{a} = (\tuple{a}(z_1), \tuple{a}(z_2), \dots, \tuple{a}(z_n))$, where $z_1, \dots, z_n$ is the enumeration of $Z$ in increasing order. 

Functions are composed from right to left: If $f: A \to B$ and $g: B \to C$ then
the composed function $A \to C$ is denoted by $g \circ f$ or just $gf$.
Note that for a $Z$-tuple $\tuple{a} \in A^Z$ and a function $f: A \to B$, the $Z$-tuple $f \circ \tuple{a} \in B^Z$ is the tuple obtained by applying $f$ to $\tuple{a}$ component-wise.

The class of finite sets is denoted by $\FinSet$. We denote by $[n]$ the set
$\{1,2,\ldots,n\}$.

\subsection{Relational structures}

For a set $\typ$, a $\tau$-sorted set $A$ is a collection of sets, one set $A_{t}$ for each sort $t \in \typ$, and a $\tau$-sorted function between two $\tau$-sorted sets $A$ and $B$ is a collection of functions $A_{t} \to B_{t}$, $t \in \tau$. We define these notions formally as follows.

\begin{definition}[Multi-sorted setting]
    Let $\typ$ be a set (of \emph{sorts symbols}). 
    A \emph{$\typ$-sorted} set is a set $A$ together with a mapping $\sort:A
    \to \tau$. For $t \in \tau$, the $t$-sort of $A$ is $A_{t} = \{a \in A \mid \sort(a) = t\}$.

    For two $\typ$-sorted sets $A,B$, a \emph{$\tau$-sorted function} from $A$
    to $B$ is a function $f: A \to B$ that preserves sorts, i.e., $\sort(f(a)) =
    \sort(a)$ for every $a \in A$. The set of $\tau$-sorted mappings from $A$ to
    $B$ is denoted by $B^A$, as above. Note that this set is not $\tau$-sorted.

    For a $\tau$-sorted set $A$ and a set $N$, their product is the $\tau$-sorted set $A \times N$ with $\sort(a,n) = a$ for every $a \in A$, $n \in N$. This time, by $A^N$ we denote the $\tau$-sorted set of those mappings $f: N \to A$ such that $f(N) \subseteq A_{t}$  for some $t$, with $\sort(f) = t$. 
\end{definition}

For two $\tau$-sorted sets $A$ and $Z$, we regard the elements of $A^Z$ also as $Z$-tuples and subsets of $A^Z$ as $Z$-ary relations. Similarly as before, $Z$-tuples can be presented as sequences of length $|Z|$ by fixing a linear order on $Z$.

\begin{definition}[Multi-sorted signature]
    A \emph{multi-sorted signature} $\Sigma$ is a triple $\Sigma=(\sig,\typ,\ar)$ where 
    $\sig$ is a set of \textit{relational symbols}, 
    $\typ$ is a set of \textit{sort symbols}, and $\ar$ assigns to each symbol $\phi \in \sig$ a finite $\typ$-sorted set $\ar(\phi)$, called  the \emph{arity} of $\phi$. 

    Such a signature is \emph{finite} if $\sig$ is finite.  
\end{definition}

Symbols $\Sigma, \sig, \typ, \ar$ are reserved for the objects above and we often keep the notation implicit. A signature is implicitly \textbf{multi-sorted and finite}. We will also implicitly assume that $\tau$ is finite.

\begin{definition}[Relational structure]
  Let $\Sigma$ be a signature. A \emph{structure in signature $\Sigma$}, or \emph{$\Sigma$-structure}, $\bfa$ consists of a $\typ$-sorted set $A$ called the \emph{domain} and an $\ar(\phi)$-ary relation $\phi^\bfa$ on $A$ (i.e., $\phi^{\bfa} \subseteq A^{\ar(\phi)}$) called the \emph{interpretation} of $\phi$ in $\bfa$ for each $\phi \in \sig$. Such a structure $\bfa$ is said to be \emph{finite} if $A$ is finite. 
\end{definition}

We shall use the same letter, but different fonts, to refer to a structure $\bfa$ (bold) and its domain $A$ (uppercase). A structure is implicitly \textbf{finite}. In order to avoid some edge cases, we also implicitly assume that every \textbf{relation} $\phi^{\bfa}$ in a structure \textbf{is nonenmpty}.

\subsection{Promise CSP}
\label{subsec:PCSP}

The Promise CSP over a pair of structures $(\bfa,\bfb)$ can be defined as  the problem of deciding whether a conjunctive formula is true in $\bfa$ or not even true in $\bfb$.%
\footnote{This is the decision version. The \emph{search version} is: given a conjunctive formula which is promised to be satisfiable in $\bfa$, find a satisfying assignment in $\bfb$. We only consider the decision version but results can be easily adjusted to the search one.} This problem only makes sense if each conjunctive formula true in $\bfa$ is also true in $\bfb$. Formal definitions are as follows. 

\begin{definition}[Conjunctive formula]
    Let $\Sigma$ be a signature and $X$ a $\typ$-sorted set. A \emph{conjunctive formula} over $X$ in the signature $\Sigma$ (or \emph{conjunctive $\Sigma$-formula}) is a formal expression $\Phi$ of the form 
    $$
    \Phi = \bigwedge_{i \in I} \phi_i(\tuple{x}_i), 
    $$
    where $I$ is a finite nonempty set, and $\phi_i \in \sig$, $\tuple{x}_i \in X^{\ar(\phi_i)}$ for all $i \in I$. The conjuncts are called \emph{constraints}.

    Given additionally a $\Sigma$-structure $\bfa$, the \emph{interpretation of $\Phi$ in $\bfa$}, or \emph{the $X$-ary relation defined in $\bfa$ by $\Phi$}, is the $X$-ary relation on $A$ defined by 
    $$
    \Phi^{\bfa}(h) \mbox{ iff } \bigwedge_{i \in I} \phi_i^{\bfa}(h\tuple{x}_i).
    $$ 
\end{definition}

We allow empty formulas ($I = \emptyset$) and interpret them $\Phi^{\bfa} = A^X$.

\begin{definition}[PCSP]
    A pair of relational structures $(\bfa,\bfb)$ over the same signature $\Sigma$ is a \emph{promise template} if
    $\Phi^{\bfa} \neq \emptyset$ implies $\Phi^{\bfb} \neq \emptyset$ for every conjunctive formula $\Phi$ in the signature $\Sigma$. 

Given a promise template $(\bfa,\bfb)$, the Promise Constraint Satisfaction
  Problem over $(\bfa,\bfb)$, denoted by $\PCSP(\bfa,\bfb)$, is the following
  problem.
\begin{enumerate}
    \item[\textsf{Input}] a finite $\typ$-sorted set $X$ and conjunctive $\Sigma$-formula $\Phi$ over $X$.
    \item[\textsf{Output}] \yes if $\Phi^{\bfa} \neq \emptyset$; $\no$ if
      $\Phi^{\bfb} = \emptyset$.\footnote{The promise is that we are in one of the two cases, i.e., not in the case that $\Phi^{\bfa}=\emptyset$ and $\Phi^{\bfb} \neq \emptyset$.}
\end{enumerate}
\end{definition}

In this context, $X$ is regarded as a set of variables and the $\tau$-sorted functions $h: X \to A$ as assignments of values in $A$ to variables. The fact that $\phi_i^{\bfa}(h\tuple{x}_i)$ means that the constraint $\phi_i(\tuple{x}_i)$ is satisfied in $\bfa$ by the assignment $h$. Thus elements of $\Phi^{\bfa}$ (or $\Phi^{\bfb}$) can be thought of as solutions of $\Phi$ in $\bfa$ (or $\bfb$).

The standard Constraint Satisfaction Problem over $\bfa$~\cite{Feder98:monotone} is $\PCSP(\bfa,\bfa)$, where typically only single-sorted signatures are considered. Here is a concrete example of a problem that falls into this framework.

\begin{example}[3LIN2] \label{ex:lin}
   Given a system of linear equations over the two-element field $\mathbb{Z}_2$ with exactly 3 variables in each equation, the task is to decide whether it has a solution.
   This problem can be phrased as $\PCSP(\bfa,\bfa)$, where $A = \{0,1\}$, the signature consists of two $[3]$-ary symbols $\phi_0, \phi_1$, and their interpretation is $\phi^{\bfa}_i(a_1,a_2,a_3)$ iff $a_1+a_2+a_3 = i \pmod{2}$. 

   We denote this PCSP as well as the template by $\lin$. We will also use this convention for other PCSPs. 
   Templates (and PCSPs) $\klin{k}$ for a positive $k$ are defined similarly.
\end{example}

An example of a ``truly'' promise problem is the following version of the approximate graph coloring problem.

\begin{example}[3- versus 5- graph coloring] \label{ex:AGC} 
Given a graph, the task is to accept if it is 3-colorable and reject if it is not 5-colorable. This is $\PCSP(\mathbf{K}_3,\mathbf{K}_5)$, where $\mathbf{K}_k$ denotes a $k$-clique, that is, a structure with a $k$-element domain and one binary relational symbol interpreted as the disequality relation on the domain.
\end{example}

The last example is a CSP, but requires two sorts instead of just one. It is a version of the Label Cover problem.

\begin{example}[$\lc_{D,E}$ -- Label Cover] \label{ex:label_cover}

Fix finite disjoint sets $D,E$. Given a bipartite (multi-)graph with vertex set $U \cup V$  and a constraint $\pi_{uv}:D \to E$ for each edge $\{u,v\}$ in the graph, the task is to decide whether all the constraints can be satisfied, i.e., whether there exist functions $h_D: U \to D$ and $h_E: V \to E$ such that $\pi_{uv}(h_D(u))=h_E(v)$ for every edge $\{u,v\}$. 

This problem is $\PCSP(\bfa,\bfa)$, where the sort symbols are $D$ and $E$, $A = D \cup E$ (with $\sort(d)=D$ for $d \in D$ and $\sort(e)=E$ for $e \in E$), the signature consists of 
all functions $\pi: D \to E$ of arity $[2]$ with $\sort(1) = D$, $\sort(2)=E$, interpreted as $\pi^{\bfa}(d,e)$ iff $\pi(d)=e$. 
We typically omit the superscript in $\pi^{\bfa}$, which should not cause a confusion because of the different number of arguments.
\end{example}

The multi-sorted setting is primarily introduced to include problems such as the Label Cover. Note however that single-sorted PCSPs have natural formulations as multi-sorted ones. For instance, $\lin$ from \cref{ex:lin} can be introduced using a 3-sorted signature, with the 3-sorted domain $\phi_0^{\bfa} \cup \phi_1^{\bfa} \cup \{0,1\}$ (where $\phi_i^{\bfa}$ is as in the example) and six binary symbols interpreted as the graphs of the projection mappings $\phi_i^{\bfa} \to \{0,1\}$. In fact, this transformation from single-sorted to multi-sorted is essentially the reduction from a PCSP to MC discussed at the end of this section. 

\subsection{Polymorphisms}

An $N$-ary polymorphism of $(\bfa,\bfb)$ is an $N$-ary function from $A$ to $B$ that \emph{preserves} every relation,   
that is, if we apply it component-wise to an $N$-tuple of tuples from $\phi^{\bfa}$, then we get a tuple from $\phi^{\bfb}$ for every $\phi$ in the common signature of $\bfa$ and $\bfb$. 
We phrase this property in terms of matrices. But first, let us discuss the terminology in the multi-sorted setting. 
Let $A$, $B$, $Z$ be $\tau$-sorted sets and $N$ be a set.

An \emph{$N$-ary function $f$ from $A$ to $B$} is a $\tau$-sorted function $A^N \to B$, i.e., an element of $B^{A^N}$. It can be regarded as a collection of functions $f_t: A_t^N \to B_t$, $t \in \tau$. 
When $|N|=1$ an $N$-ary function is 
called unary.

For $n \in N$, the \emph{$N$-ary projection to the $n$-th coordinate} is
denoted by $\proj{N}{n}$, i.e., $\proj{N}{n}: A^N \to A$ is defined by $\proj{N}{n}(\tuple{a}) = \tuple{a}(n)$. The set $A$ will be clear from the context. 
For $z \in Z$, the $Z$-ary projection to the $z$-th coordinate $\proj{Z}{z}: A^Z \to A$ is defined by the same formula. Note that its image is contained in $A_{\sort(z)}$.

An element $M \in A^{Z \times N}$ can be regarded as a matrix whose rows are
indexed by elements $z \in Z$, columns are indexed by elements $n \in N$, and
the $(z,n)$ entry is $M(z,n) \in A_{\sort(z)}$.  The $N$-tuple of columns is
denoted by $\cols(M) \in (A^Z)^N$ and, for $n \in N$, the $n$-th column is
denoted by $\col{n}{M} \in A^Z$. The $Z$-tuple of rows is denoted by $\rows(M) \in (A^N)^Z$ and the $z$-th row by $\row{z}{M} \in (A_{\sort(z)})^N \subseteq A^N$. 

\begin{definition}[Polymorphism]
  Let $\ab$ be a pair of $\Sigma$-structures and $N$ a finite set.
  An  \emph{$N$-ary relation-matrix pair} for $\bfa$ is a pair $(\phi, M)$,
  where $\phi \in \sigma$ and $M \in A^{\ar(\phi) \times N}$ is a matrix whose
  each column is in $\phi^{\bfa}$. We denote by $\Mat{\bfa}{N}$ the set of all $N$-ary relation-matrix pairs for $\bfa$.
  $$
  \Mat{\bfa}{N} = \{(\phi, M) \mid \phi \in \sigma, M \in A^{\ar(\phi) \times
  N}, \ \forall n \in N \  \col{n}{M} \in \phi^{\bfa}\}.
  $$
  An $N$-ary function $f$ from $A$ to $B$ is a \emph{polymorphism} of $\ab$ if 
  $$
  \forall (\phi,M) \in \Mat{\bfa}{N} \ \ \phi^{\bfb}(f \circ \rows(M)). 
  $$
  We denote the set of $N$-ary polymorphisms of $\ab$ by $\pol^{(N)} \ab$ and the collection of these sets%
  \footnote{It may seem that $\pol \ab$ is a monstrous object: for each finite set $N$ we have a set of $N$-ary functions from $A$ to $B$. However, note that $N$-ary polymorphisms fully determine $N'$-ary polymorphisms whenever $|N|=|N'|$. }
  by 
  $$
  \pol \ab = (\pol^{(N)} \ab)_{N \in \FinSet}.
  $$
\end{definition}

An $N$-ary polymorphism of $\ab$ gives us a way to combine $N$ tuples from a relation $\phi^{\bfa}$ to get a single tuple from $\phi^{\bfb}$. This extends to any conjunctive formula: if $\Phi$ is a conjunctive formula over $X$ and $M \in A^{X \times N}$ has all the columns in $\Phi^{\bfa}$, then $\Phi^{\bfb}(f \rows (M))$. In other words, polymorphisms give us a way to combine $\bfa$-solutions to get a $\bfb$-solution.
This is formalized in the following proposition.

\begin{proposition}[Polymorphisms combine solutions] \label{prop:combining_solutions}
Let $\ab$ be a pair of $\Sigma$-structures, $N$ a finite set, and $f: A^N \to B$. The following are equivalent.
\begin{enumerate}[label=\textnormal{(\roman*)}, ref=(\roman*)]
    \item \label{item:1:SolutionsToSolutions} $f$ is a polymorphism of $\ab$.
    \item \label{item:2:SolutionsToSolutions} For every finite set $X$, every conjunctive formula $\Phi$ over $X$, and every $M \in A^{X \times N}$ whose every column is in $\Phi^{\bfa}$, we have $\Phi^{\bfb}(f \rows(M))$.
\end{enumerate}    
\end{proposition}

\begin{proof}
   Assume \labelcref{item:1:SolutionsToSolutions}  and consider $M$ and $\Phi = \bigwedge_{i \in I} \phi_i(\tuple{x}_i)$ as in \labelcref{item:2:SolutionsToSolutions}. We need to show 
   that $\phi_i^{\bfb}(f \rows(M) \tuple{x}_i)$ for every $i \in I$. Consider the matrix $M_i \in A^{\ar(\phi_i) \times N}$ defined by $M_i(z,n) = M(\tuple{x}_i(z),n)$ for every $z \in \ar(\phi_i)$ and $n \in N$. Note that $\rows(M_i) = \rows(M) \tuple{x}_i$ and $\col{n}{M_i} = \col{n}{M} \tuple{x}_i$. For every $n \in N$, the $n$-th column of $M$ is in $\Phi^{\bfa}$, in particular $\phi_i^{\bfa}(\col{n}{M}\tuple{x}_i)$, therefore the $n$-th column of $M_i$ is in $\phi_i^{\bfa}$. As $f$ is a polymorphism, we obtain that $f\rows(M_i)=f \rows(M) \tuple{x}_i$ is in $\phi^{\bfb}$, as required.

   The opposite implication follows by considering  matrices with $X = \ar(\phi_i)$ for each $i\in I$. 
\end{proof}

The collection of polymorphisms of a pair $\ab$ is closed under taking minors in the sense of the following definition. We call such collections function minions.%
\footnote{Conversely, ``almost'' every function minion on finite sets is a
minion of polymorphisms. The caveat is that we would need to allow infinite
signatures and ignore functions of arity $\emptyset$.
} 

\begin{definition}[Function minion]
Let $A,B$ be $\tau$-sorted sets and $N$, $N'$ finite sets. 
For  $f: A^N \to B$ and $\pi:N \to N'$, the \emph{minor of $f$ given by $\pi$}, denoted by $f^{(\pi)}$, is the $N'$-ary function $f^{(\pi)}: A^{N'} \to B$ defined by
$$
f^{(\pi)}(\ba) = f(\ba \circ \pi) \quad \textnormal{ for every } \ba \in A^{N'}. 
$$ 
A collection $\mnn M = (\mnn M^{(N)})_{N \in \FinSet}$, where $\mnn M^{(N)}$ is a set of $N$-ary functions from $A$ to $B$, is a \emph{function minion on $(A,B)$} if $f^{(\pi)} \in \mnn M^{(N')}$ for every $N,N' \in \FinSet$, $f \in \mnn M^{(N)}$, and $\pi: N \to N'$. 
\end{definition}

\begin{example}
    If $f: A^{[3]} \to A$ and $\pi: [3] \to [2]$ is defined by
    $\pi(1)=\pi(3)=2$, $\pi(2)=1$, then $f^{(\pi)}(a_1,a_2) = f(a_2,a_1,a_2)$.
    Informally, a minor of $f$ is a function that can be obtained from $f$ by merging and permuting variables (and introducing dummy ones).
\end{example}

\begin{example} \label{ex:proj_minion}
    The collection given by $\mnn{P}^{(N)} = \{\proj{N}{n} \mid n \in N\}$ is an easy and important example of a function minion on $(A,A)$. Note that $(\proj{N}{n})^{(\pi)} = \proj{N'}{\pi(n)}$ for every $\pi: N \to N'$.
\end{example}

A fundamental role (though not always explicit) in the CSP theory, as well as for various variants of CSPs, is played by a specific conjunctive formula $\Phi$ on the set of variables $A^N$ for some finite set $N$. Note that assignments from the set of variables to $A$ (to $B$) are exactly the $N$-ary functions from $A$ to $A$ (to $B$). For a fixed pair $\ab$, the formula $\Phi$ is created by placing all the possible constraints with the restriction that $\Phi^{\bfa}(\proj{N}{n})$ for every $n \in N$. Then $\Phi^{\bfb}$ is exactly the set of $N$-ary polymorphisms of $\ab$. 

\begin{proposition}[Canonical formula] 
\label{prop:canonical_formula} 
For every pair $\ab$ of finite $\Sigma$-structures 
and $N$ a finite set, the $\Sigma$-formula\footnote{In the terminology used in e.g.~\cite{bodirsky2021complexity}, this is the canonical conjunctive query of the $N$-th power of $\bfa$. In~\cite{Jeavons97:closure}, this construction
  was called the indicator problem.}
    $$
    \Phi = \bigwedge_{(\phi,M) \in \Mat{\bfa}{N}} \phi(\rows(M))
    $$ 
    over the set of variables $A^N$ satisfies
  \begin{itemize}
      \item $\Phi^{\bfa}(\proj{N}{n})$ for every $n \in N$, and
      \item $\Phi^{\bfb} = \pol^{(N)} \ab$. 
\end{itemize}
\end{proposition}

\begin{proof}
    By definition of $\Phi^{\bfa}$, we have $\Phi^{\bfa}(\proj{N}{n})$ iff $\phi^{\bfa}(\proj{N}{n} \ \rows(M))$ for each $N$-ary relation-matrix pair $(\phi,M)$ for $\bfa$. Since $\proj{N}{n} \ \rows(M) = \col{n}{M} \in \phi^{\bfa}$, the first item holds.

    For the second item, we have for each $f: A^N \to B$ that $\Phi^{\bfb}(f)$ iff $\phi^{\bfb}(f \ \rows(M))$ for each $(\phi,M) \in \Mat{\bfa}{N}$, which happens, by definition, iff $f$ is a polymorphism of $\ab$ .
\end{proof}

Note that, as claimed above, $\Phi$ is created by placing all possible constraints so that the first item is satisfied. Indeed, any constraint $\phi(\tuple{x})$ over $A^N$ is equal to $\phi(\rows(M))$ for some $\phi \in \sig$ and $M \in A^{\ar(\phi) \times N}$. The fact that $\phi^{\bfa}(\proj{N}{n})$ is exactly saying that $\col{n}{M}$ is in $\phi^{\bfa}$, so satisfying  $\phi^{\bfa}(\proj{N}{n})$ for each $n\in N$, is equivalent to $M \in \Mat{\bfa}{N}$. 

Applying the canonical formula for a singleton set $N$ gives us a characterization of templates. \cref{item:3templates} in the proposition below is in fact often used as a definition of a template.

\begin{proposition}[Characterization of templates]  
    Let $\ab$ be a pair of finite $\Sigma$-structures. The following are equivalent.
    \begin{enumerate}[label=\textnormal{(\roman*)}, ref=(\roman*)]
        \item\label{item:1:CharOfTemplates} $\ab$ is a promise template.
        \item\label{item:2:CharOfTemplates} For each conjunctive $\Sigma$-formula $\Phi$ over the set of variables $A$, if 
        $\Phi^{\bfa}(\id_A)$, then $\Phi^{\bfb} \neq \emptyset$.        
        \item\label{item:3:CharOfTemplates} There exists a unary polymorphism of $\ab$. \label{item:3templates}
    \end{enumerate}
\end{proposition}

\begin{proof}
  The implication from \labelcref{item:3:CharOfTemplates} to
  \labelcref{item:1:CharOfTemplates} follows from~\cref{prop:combining_solutions} for a singleton set $N$,
  the implication from \labelcref{item:1:CharOfTemplates} to \labelcref{item:2:CharOfTemplates} is trivial, and the implication from \labelcref{item:2:CharOfTemplates} to \labelcref{item:3:CharOfTemplates} follows from~\cref{prop:canonical_formula} applied to a singleton set $N$.
\end{proof}

Starting from the canonical formula, the theory can now go in two directions. The
original approach for CSPs from~\cite{Jeavons97:closure,Bulatov05:classifying,BOP18,bodirsky2021complexity}
can be formulated with slight imprecision as follows.
If $\pol \aaa \subseteq \pol \aap$, then each relation in $\bfa'$ 
can be defined
by existentially quantifying the canonical formula (for $\aaa$) for a suitable $N$, which then
implies $\pcsp \aap \leq \pcsp \aaa$, where $\leq$ denotes the polynomial-time reducibility. 
A more detailed explanation is given in \Cref{app:definability}.

This direction can continue by replacing
definability with more expressive constructions and thus allowing us to replace
the inclusion $\pol \aaa \subseteq \pol \aap$ by weaker requirements, which in
turn gives us more reductions. 
One step in this process replaced  the inclusion
by the existence of so-called minion homomorphisms~\cite{BOP18} and this was generalized to PCSPs in~\cite{BartoBKO21} based on~\cite{pippenger2002galois,BG21}. 

The second direction that the theory can take, also based on the canonical formula, avoids the definability considerations. Instead, it proves the reduction theorem based on minion homomorphisms in a more direct way. This approach, discovered in~\cite{BartoBKO21}, is the one we follow in this work.

\subsection{Minion homomorphisms and reductions}

A minion homomorphism between function minions is a mapping of $N$-ary functions in the first minion to $N$-ary functions in the second minion that preserves taking minors. 
This concept does not depend on concrete functions in the minion, it only depends on the mappings $f \mapsto f^{(\pi)}$.
We therefore first introduce an abstraction of function minions that carries exactly this information. 

\begin{definition}[Minion]
An (abstract) \textit{minion} $\mnn M$ consists of a collection of sets $(\mnn M^{(N)})_{N \in \FinSet}$, together with a \textit{minor map} ${\mnn M}^{(\pi)}:  {\mnn M}^{(N)}\to  {\mnn M}^{(N')}$ for every function $\pi:N\to N'$, which satisfies that ${\mnn M}^{(\id_{N})}= \id_{{\mnn M}^{(N)}}$ for all finite sets $N$ and ${\mnn M}^{(\pi)} \circ {\mnn M}^{(\pi')} = \mnn M^{(\pi \circ \pi')}$ whenever such a composition is well-defined. When the minion is clear, 
we write $f^{(\pi)}$ for ${\mnn M}^{(\pi)}(f)$.

A minion $\mnn M$ is \emph{nontrivial} if $\mnn M^{(N)}$ is nonempty for every (equivalently some) 
nonempty $N$, and it is \textit{locally finite} if $\mnn M^{(N)}$ is finite for every $N$. 
\end{definition}

The most natural choice of morphisms between minions is minion homomorphisms defined as follows.\footnote{In the language of category theory, a minion is simply a functor from the category of finite sets to the category of sets (a minion corresponds to the functor $X \mapsto \mnn M^{(X)}$, $\pi \mapsto \mnn M^{(\pi)}$) and minion homomorphisms are natural transformations. Note that the projection minion from~\cref{ex:proj_minion} is naturally equivalent to the inclusion functor.} 

\begin{definition}[Minion homomorphism] Let $\mnn M $ and $\mnn M'$ be minions.
    A \textit{minion homomorphism} from $\mnn M $ to  $\mnn M'$ is a collection of functions $(\xi^{(N)}: \mnn M^{(N)} \to \mnn M'^{(N)})_{N \in \FinSet}$ that preserves taking minors, that is, $\xi^{(N')}(\mnn M^{(\pi)}(f))={\mnn M'}^{(\pi)}(\xi^{(N)}(f))$ for every $N, N' \in \FinSet$, $f \in \mnn{M}^{(N)}$, and $\pi:N \to N'$.
\end{definition}

The reduction theorem discussed above is the following.

\begin{theorem}[Reductions via minion homomorphism] \label{thm:crisp-reduction}
    Let $\ab$, $\abp$ be promise templates. If there is a minion homomorphism from $\pol\ab$ to $\pol\abp$, then $\PCSP\abp \leq \PCSP\ab$.\footnote{In fact, Theorem~1 even holds with a log-space reduction, but that will not concern us.}
\end{theorem}

This reduction theorem explains hardness for CSPs: if $\pcsp(\bfa,\bfa)$ cannot be solved in polynomial time by the algorithms in~\cite{Bulatov17:focs,Zhuk20:jacm}, then 
$\pol(\bfa,\bfa)$ has a minion homomorphism to the projection minion from \cref{ex:proj_minion} (with $|A|  \geq 2$), which has a minion homomorphism to every nontrivial minion. 

The modern proof of \cref{thm:crisp-reduction} is to show that each $\pcsp \ab$ is equivalent to a certain computational problem parameterized by the (abstract) minion of polymorphisms, called the \textit{minor condition problem}, and that a minion homomorphism (trivially) gives a reduction between such problems. 

\begin{definition}[Minor Condition Problem] \label{def:MC}
Given a nontrivial minion $\mnn M$ and an integer $k$, 
the Minor Condition Problem for $\mnn M$ and $k$, denoted by $\mc(\mnn M,k)$ is the following problem:

\begin{enumerate}
    \item[\textsf{Input}] 
    \begin{enumerate}[label=\textnormal{\arabic*.}]
     \item disjoint sets $U$ and $V$ (the sets of \emph{variables}),
     \item a set $D_x$ with $|D_x| \leq k$ for every $x \in U \cup V$ (the \emph{domain} of $x$),
     \item a set of formal expressions of the form $\pi(u) = v$, where $u \in U$, $v \in V$, and $\pi: D_u \to D_v$ (the \emph{minor conditions}). 
    \end{enumerate}
    \item[\textsf{Output}]
    \begin{itemize}
        \item[\yes] if there exists a function $h$ from $U \cup V$ with $h(x) \in D_x$ (for each $x \in U \cup V$) such that, for each minor condition $\pi(u) = v$, we have $\pi(h(u)) = h(v)$.
        \item[\no] if there does not exist a function $h$ from $U \cup V$ with $h(x) \in \mnn M^{(D_x)}$  such that, for each minor condition $\pi(u) = v$, we have $\mnn M^{(\pi)}(h(u)) = h(v)$.
    \end{itemize}
\end{enumerate}
\end{definition}

The name for the minor condition problem comes from the requirement $\mnn M^{(\pi)}(h(u)) = h(v)$: the element of $\mnn M$  assigned to $v$ must be the minor of the element assigned to $u$ given by $\pi$. Note also that $\pi(h(u)) = h(v)$ is equivalent to $\mnn P^{(\pi)}(h(u)) = h(v)$ for the projection minion $\mnn P$ from~\cref{ex:proj_minion}. 

Since $\mnn M$ is nontrivial, an instance cannot simultaneously be a $\yes$ and $\no$ instance. Indeed, if $h$ witnesses that an instance is a $\yes$ instance, then $x \mapsto f^{(\gamma_{h(x)})}$, where $f \in \mnn M^{([1])}$ and $\gamma_{h(x)}$ is the mapping $[1] \to D_x$  with $1 \mapsto h(x)$, witnesses that the instance is not a $\no$ instance.

Notice also that an instance of MC is very similar to an instance of LC from~\cref{ex:label_cover}. In fact, $\mc(\mnn M,k)$ can be phrased as a PCSP over a
certain multi-sorted template.

The reduction between two PCSPs in \cref{thm:crisp-reduction} based on a minion homomorphism is a composition
of three reductions: from PCSP to MC, from MC (over one minion) to MC (over another one), and from MC to PCSP.  
Overall, we have the following reductions (depicted as arrows) for templates $\ab$, $\abp$, their polymorphism minions $\mnn M$, $\mnn M'$, and a sufficiently large $k$:
$$\begin{tikzcd}
  \pcsp \abp \arrow[d, leftrightarrow]  & \pcsp \ab \arrow[d, leftrightarrow] \\
  \mc(\mnn M', k) \ar[r, ""] & \mc(\mnn M, k) 
\end{tikzcd}$$

We give proofs for each of these three reductions below, starting with the simplest.
\cref{thm:crisp-reduction} is an immediate consequence of these three reductions.

\begin{proposition}[Between MCs]
  Let $\mnn M$ and $\mnn M'$ be nontrivial minions such that there is a  minion homomorphism $\xi:\mnn M \to \mnn M'$. Then $\mc (\mnn M', k) \leq \mc (\mnn M, k)$ for any positive integer $k$. 
\end{proposition}

\begin{proof}
   We claim that the trivial reduction, i.e., not changing the instance, works.
   Clearly, $\yes$ instances are preserved. On the other hand, if $h$ witnesses
   that the instance is not a $\no$ instance of $\mc (\mnn M,k)$, then $h'$
   defined by $h'(x) = \xi(h(x))$ witnesses the instance is not a $\no$ instance
   of $\mc(\mnn M',k)$ since for every minor condition $\pi(u)=v$ we have 
   $$
   {\mnn M'}^{(\pi)}(h'(u)) = {\mnn M'}^{(\pi)}(\xi(h(u))) = \xi(\mnn M^{(\pi)}(h(u)))
   = \xi(h(v)) = h'(v).
   $$
\end{proof}

The following object is useful for the reduction from PCSP to MC, both in the crisp and valued settings. 

\begin{definition}[Canonical matrix] \label{def:canonicalMatric} 
Let $\psi \subseteq A^Z$ be a relation. The \emph{canonical matrix} for $\psi$ is the matrix 
$\GM{\psi} \in A^{Z \times \psi}$
defined by
$$
\GM{\psi}(z,\tuple{a}) = \tuple{a}(z)
\mbox{\quad for every }
z \in Z, \ \tuple{a} \in \psi.
$$
\end{definition}

First we note that, by construction, the canonical matrix has the following properties.

\begin{lemma} \label{lem:canonicalMatrix}
    Let $\psi \subseteq A^Z$ be a relation and $\GM{\psi}$ its canonical matrix. Then $\row{z}{\GM{\psi}}$ is the restriction of $\proj{\ar(\psi)}{z}$ to $\psi$ and $\col{\tuple{a}}{\GM{\psi}} = \tuple{a}$.  
\end{lemma}

\begin{proposition}[From PCSP to MC]\label{prop:PCSPleqMC}
    Let $\ab$ be a promise template and $\mnn M = \pol \ab$. If $k$ is a sufficiently large integer, then
    $\PCSP \ab \leq \mc( \mnn M, k)$.
\end{proposition}

\begin{proof}
From a conjunctive formula  $\Phi = \bigwedge_{i \in I} \phi_i(\tuple{x}_i)$ over $X$, 
where $\phi_i \in \sig$ and $\tuple{x}_i \in X^{\ar(\phi_i)}$, we create an
  instance of $\mc(\mnn M, k)$ as follows.
  \begin{enumerate}[label=\textnormal{\arabic*.}]
      \item $U = I$, $V = X$.
      \item $D_i = \phi_i^{\bfa}$, $D_x = A_{\sort(x)}$ for each $i\in I$, $x \in X$.
      \item For each $i \in I$  and $z \in \ar(\phi_i)$, we introduce the constraint $\pi_{i,z}(i) = \tuple{x}_i(z)$, where $\pi_{i,z}$ is the restriction of $\proj{\ar(\phi_i)}{z}: A^{\ar(\phi_i)} \to A$ to the set $\phi_i^{\bfa}$ with its codomain restricted to $A_{\sort(z)}$ so that $\pi_{i,z}: \phi_i^{\bfa} \to A_{\sort(z)}$. 
  \end{enumerate}    
Note that this is a valid instance as long as $k \geq |\phi^{\bfa}|$ and $k \geq |A_{\sort(x)}|$ for every  $\phi\in\sig$ and $x \in X$. We claim that for any such $k$, the above transformation, which is clearly polynomial-time, is a correct reduction.

In order to show completeness, assume that $h \in \Phi^{\bfa}$. We claim that
  the following function $h'$ is a witnessing assignment to the $\mc$ instance.
$$
h'(x) = h(x) \mbox{ for $x \in V = X$}
\mbox{\quad\quad and \quad\quad}
h'(i) = h \tuple{x}_i \mbox{ for $i \in U = I$}.
$$
Clearly $h'(x) \in D_x$ and, since $\phi_i(\tuple{x}_i)$ is satisfied in $\bfa$ by the assignment $h$, we also have $h'(i) \in D_{i}$.
For a minor condition $\pi_{i,z}(i)=\tuple{x}_i(z)$ where $z \in \ar(\phi_i)$, we have
$$
\pi_{i,z}(h'(i)) =
\proj{\ar(\phi_i)}{z}(h'(i)) = (h'(i))(z) = h \tuple{x}_i (z) =
h'(\tuple{x}_i(z)),
$$
which finishes the proof of completeness. 

For soundness of the reduction, assume that $h'$ is an assignment witnessing
  that the $\mc$ is not a $\no$ instance. We claim that the assignment $h: X \to
  B$ defined by $h(x) = (h'(x))(\incl_x)$, where $\incl_x: A_{\sort(x)} \to A$
  denotes the inclusion, witnesses that the original $\PCSP$ instance is not a
  $\no$ instance, i.e., that $h\tuple{x}_i \in \phi_i^{\bfb}$ for every constraint $\phi_i(\tuple{x}_i)$ in $\Phi$. 

Since $f := h'(i)$ is a polymorphism of $\ab$, we in particular have $f \ \rows(\GM{\phi_i^{\bfa}}) \in \phi_i^{\bfb}$ (where $\GM{\phi_i^{\bfa}}$ is the canonical matrix from~\cref{def:canonicalMatric}). We show that this tuple is equal to $h\tuple{x}_i$ by comparing coordinates. Let $z \in \ar(\phi_i)$. From the minor condition $\pi_{i,z}(i)=\tuple{x}_i(z)$ we obtain $f^{\pi_{i,z}} = h'(\tuple{x}_i(z))$. Applying this fact, \cref{lem:canonicalMatrix}, and definitions, we obtain
\arxiv{
\begin{align*}
h\tuple{x}_i(z) &=
h(\tuple{x}_i(z)) =
(h'(\tuple{x}_i(z)))(\incl_{\tuple{x}_i(z)}) =
f^{\pi_{i,z}}(\incl_{\tuple{x}_i(z)}) = 
f(\incl_{\tuple{x}_i(z)} \pi_{i,z}) \\ &=
f(\row{z}{\GM{\phi_i^{\bfa}}})=
(f \ \rows( \GM{\phi_i^{\bfa}})) (z)
\end{align*}
}
\conference{
\begin{align*}
h\tuple{x}_i(z) &=
h(\tuple{x}_i(z)) =
(h'(\tuple{x}_i(z)))(\incl_{\tuple{x}_i(z)}) \\ & = 
f^{\pi_{i,z}}(\incl_{\tuple{x}_i(z)}) = 
f(\incl_{\tuple{x}_i(z)} \pi_{i,z}) \\ &=
f(\row{z}{\GM{\phi_i^{\bfa}}})=
(f \ \rows( \GM{\phi_i^{\bfa}})) (z)
\end{align*}
}which finishes the proof.
\end{proof}

\begin{proposition}[From MC to PCSP] \label{prop:MCleqPCSP}
    Let $\ab$ be a promise template and $\mnn M = \pol \ab$. For any positive integer $k$, 
    $\mc( \mnn M, k) \leq \PCSP \ab$.
\end{proposition}

\begin{proof}
    We transform an instance of $\mc(\mnn M,k)$ (with notation as
    in~\cref{def:MC}) to an instance of $\PCSP \ab$ in two steps.

    In the first step, for every $x \in X := U \cup V$ we take a canonical formula $\Phi_x$ from~\cref{prop:canonical_formula} for the arity set $N = D_x$ (the set of variables is thus $A^{D_x}$). We make the variable sets disjoint and create $\Phi$ as a conjunction of all the $\Phi_x$. Specifically, let the variable set of $\Phi$ be $Y = \{(\tuple{a},x) \mid x \in U \cup V, \tuple{a} \in A^{D_x}\}$ and the variable set of $\Phi_x$ (after renaming variables) be $\{(\tuple{a},x) \mid \tuple{a} \in A^{D_x}\}$. Note that assignments $f:Y \to A$ correspond exactly to collections $(f_x: A^{D_x} \to A)_{x \in X}$ (specifically, $f_x(\tuple a) = f(x,\tuple{a})$) and, similarly, assignments $Y \to B$ correspond to collections $(f_x: A^{D_x} \to B)_{x \in X}$. We do not formally distinguish these objects in this proof.  From the properties of the canonical instances $\Phi_x$ we in particular have the following.
    \begin{enumerate}
        \item[(C)]\label{item:C} $\Phi^{\bfa}((f_x)_{x \in X})$ whenever each $f_x$ is a $D_x$-ary projection on $A$. 
        \item[(S)]\label{item:S} If $\Phi^{\bfb}((f_x)_{x \in X})$, then each $f_x$ is a $D_x$-ary polymorphism of $\ab$. 
    \end{enumerate}

    In the second step, we create from $\Phi$ the resulting instance $\Psi$ of $\PCSP \ab$ by identifying, for each minor condition $\pi(u)=v$ and each $\tuple{a} \in A^{D_v}$, the variables $(\tuple{a}\pi,u)$ and $(\tuple{a},v)$. Now assignments for $\Psi$ from the new set of variables to $A$ correspond to those assignments $(f_x: A^{D_x} \to A)_{x \in X}$ for $\Phi$ that have equal values on all pairs of identified variables, i.e., to those that satisfy $f_u(\tuple{a}\pi) = f_v(\tuple{a})$; in other words, 
    \begin{itemize}
        \item[(M)]  $f_u^{(\pi)} = f_v$ for each minor condition $\pi(u)=v$.
    \end{itemize}
    Of course, an analogous observation remains valid for assignments to $B$. 

    The transformation is clearly polynomial-time 
    and it is a reduction: If $h: X \to A$ is a witness that the original $\mc$
    instance is a $\yes$ instance, then $(\proj{D_x}{h(x)})_{x \in X}$ corresponds to a satisfying assignment for $\Psi$. Indeed, (M) is satisfied because $(\proj{D_u}{h(u)})^{(\pi)} = \proj{D_u}{\pi(h(u))} = \proj{D_u}{h(v)}$ and the constraints are satisfied because of (C).
    On the other hand, if $(f_x)_{x \in X}$ corresponds to a satisfying assignment of the $\pcsp$, then $f_u^{(\pi)} = f_v$ by (M) and each $f_x$ is a polymorphism by (S). 
\end{proof}

\section{Valued Promise CSP} \label{sec:ValuedPCSP}

The generalization of PCSP to the valued setting is obtained by replacing relations by valued relations, that is, mappings $A^Z \to \qq \cup \{-\infty\}$, and suitably adjusting the concepts.  The crisp PCSPs can be modelled as Valued PCSPs with $\{-\infty,0\}$-valued relations. 

This section covers the basics up to a valued and improved version of canonical formulas. A generalization of minion homomorphisms and the main reduction theorem are given in \cref{sec:valued-homos} and examples of valued homomorphisms are shown in \cref{sec:examples}.

\subsection{Preliminaries}

We denote by $\qqpos$ ($\qqnonneg$) the set of positive (nonnegative) rational numbers and by $\qqinf$ the set of rational numbers together with an additional symbol $-\infty$. We naturally extend the operations and order, leaving $0 \cdot -\infty$ undefined. 

We will work with probability distributions on finite sets with rational probabilities, so we can formally regard a probability distribution on $N$ as a function $\mu: N \to \qqnonneg$ such that $\sum_{n \in N} \mu(n)=1$. We denote by $\Delta N$ the set of probability distributions on $N$. The \textit{support} of a probability distribution $\mu \in \Delta N$ is the set $\supp(\mu)=\{n\in N \mid \mu(n)>0\}$.

If $f: N \to N'$ and $\mu \in \Delta N$, we define $f(\mu) \in \Delta N'$ in the natural way $(f(\mu))(n') = \sum_{n; f(n)=n'} \mu(n)$, that is, $n'$ can be sampled according to $f(\mu)$ by sampling $n$ according to $\mu$ and computing $n'=f(n)$. 
We also use the notation $F(\mu)$ when $\mu \in \Delta N$ and $F$ is a probability distribution on a set of mappings $N \to N'$, i.e., to sample $F(\mu)$ we independently sample $n \sim \mu$, $f \sim F$ and compute $f(n)$. 

Given $\mu \in \Delta N$ and a function $f: N \to \qq$, we denote by $\ex{n \sim \mu} f(n)$ the expected value of $f(n)$ when $n$ is sampled according to $\mu$, i.e., $\ex{n \sim \mu} f(n) = \sum_{n \in N} \mu(n)f(n)$.

A basic tool for some of the proofs is Farkas' lemma~\cite{Schrijver86}. The following formulation will be convenient for us. In the statement, juxtaposition denotes the standard matrix multiplication, ${}^T$ is used for the transposition, and $\tuple{x}\geq0$ means that all the components are  nonnegative.

\begin{theorem}[Farkas' lemma] \label{thm:Farkas}
    Let $I,J$ be finite sets, $F \in \qq^{I \times J}$, and $\tuple{q} \in \qq^{I}$. The following are equivalent.
    \begin{enumerate}[label=\textnormal{(\roman*)}] 
    \item $\exists \tuple{y}  \in (\qqnonneg)^{J} \ F\tuple{y} \leq \tuple{q}$. 
    \item $\forall \tuple{x} \in (\qqnonneg)^{I} \ (F^T\tuple{x} \geq 0 \implies  \tuple{q}^T\tuple{x} \geq 0).$
    \end{enumerate}
\end{theorem}

\subsection{Valued Structures and the Valued PCSP}

\begin{definition}[Valued relational structure]
  Let $\tau$ be a set (of sorts), and let $Z$ and $A$ be $\tau$-sorted sets. A \emph{$Z$-ary valued relation on $A$}, or a \emph{$Z$-ary payoff function on $A$}, 
  is a function $\phi: A^Z \to \overline{\qq}$. The \textit{feasibility set} of
  $\phi$, denoted by $\feas(\phi)$, is the pre-image of $\qq$ under $\phi$. 
  
  Let $\Sigma$ be a signature. A \emph{valued $\Sigma$-structure} $\bfa$ consists of a $\typ$-sorted set $A$ called the \emph{domain} and an $\ar(\phi)$-ary valued relation $\phi^\bfa$ on $A$, the \emph{interpretation of $\phi$ in $\bfa$}, for every $\phi \in \sig$. Such a structure $\bfa$ is said to be finite if $A$ is finite.   For a rational number $c$ we write $\bfa \leq c$ if $\phi^{\bfa}(\tuple{a}) \leq c$ for every $\phi \in \sig$ and $\tuple{a} \in A^{\ar(\phi)}$.

  For a valued $\Sigma$-structure $\bfa$, the \emph{feasibility structure},
  denoted by $\feas(\bfa)$, is the (non-valued) $\Sigma$-structure obtained by replacing each $\phi^{\bfa}$ by $\feas(\phi^{\bfa})$. 
\end{definition}

Similarly as in the crisp case, we implicitly assume that in a valued structure every $\feas(\phi^{\bfa})$ is nonempty.

\begin{definition}[Payoff formula]
    Let $\Sigma$ be a signature and $X$ a finite $\typ$-sorted set. A \emph{payoff formula} over $X$ in the signature $\Sigma$, or a \emph{payoff $\Sigma$-formula}, is a formal expression of the form 
    $$\Phi = \sum_{i \in I} w_i \phi_i(\tuple{x}_i), $$ 
    where $I$ is a finite nonempty set, and $w_i \in \qqnonneg$  (\emph{weights}), $\phi_i \in \sig$, $\tuple{x}_i \in X^{\ar(\phi_i)}$ for all $i \in I$.
    The \emph{weight of $\Phi$} is $\weight(\Phi) = \sum_{i \in I} w_i$.

    Given additionally a valued $\Sigma$-structure $\bfa$, the \emph{interpretation of $\Phi$ in $\bfa$}, or \emph{the $X$-ary valued relation defined in $\bfa$ by $\Phi$}, is the $X$-ary valued relation on $A$ defined by 
    $$
    \Phi^{\bfa}(h) = \sum_{i \in I} w_i \phi_i^{\bfa}(h\tuple{x}_i),
    $$ 
    where  summands $0 \cdot -\infty$ are evaluated as $-\infty$ (but we keep $0 \cdot -\infty$ undefined in different contexts).
\end{definition}

We allow empty formulas $\Phi$, and define $\weight(
\Phi)=0$ and $\Phi^{\bfa}(h) = 0$.

Note that the convention that $0 \cdot -\infty = -\infty$ ensures that $\feas(\Phi^{\bfa})$ is equal to the interpretation of $\bigwedge_{i \in I} \phi_i(\tuple{x}_i)$ in $\feas(\bfa)$ and, for $h \in \feas(\Phi^{\bfa})$, the sum defining $\Phi^{\bfa}(h)$ does not contain any infinities.

\begin{definition}[Valued PCSP]
A \textit{valued promise template} is a quadruple $\abcs$ where
 \begin{itemize}
     \item $\bfa$, $\bfb$ are valued relational structures in the same signature $\Sigma$, and
     \item $c,s \in \qq$ are the \textit{completeness} and \textit{soundness} parameters respectively
 \end{itemize}
    such that 
    $\exists h \ \Phi^{\bfa}(h) \geq c \,\weight(\Phi)$ implies $\exists h \ \Phi^{\bfb}(h) \geq s \,\weight(\Phi)$ for every payoff $\Sigma$-formula $\Phi$.

Given a valued promise template $\abcs$, the \emph{Promise Constraint
  Satisfaction Problem over $\abcs$}, denoted by $\PCSP\abcs$, is the following problem. 
\begin{enumerate}
    \item[\textsf{Input}] a finite $\tau$-sorted set $X$ and a payoff $\Sigma$-formula $\Phi$ over $X$.
    \item[\textsf{Output}] \yes if $\exists h \ \Phi^{\bfa}(h) \geq c\,
      \weight(\Phi)$; $\no$ if $\forall h \ \Phi^{\bfb}(h) < s
      \,\weight(\Phi)$.\footnote{The promise is that we are in one of the two cases, i.e., not in the case that $\forall h \ \Phi^{\bfa}(h)<c\,\weight(\Phi)$ and $\exists h \ \Phi^{\bfb}(h)\geq s w(\Phi)$.}
\end{enumerate}
\end{definition}

Let us start the discussion about this generalization of PCSPs by giving examples of problems included in this framework. 

First observe that Valued PCSPs indeed generalize crisp PCSPs: for a crisp template $\abp$ we define a valued promise template $(\bfa,\bfb,0,0)$ by setting $\phi^{\bfa}(\tuple{a})=0$ if $\tuple{a} \in \phi^{\bfa'}$ and $\phi^{\bfa}(\tuple{a})=-\infty$ otherwise for all $\phi \in \sig$, $\ba \in A^{\ar(\phi)}$, and similarly for $\bfb'$. Clearly, $\pcsp \abp$ is equivalent to $\pcsp (\bfa,\bfb,0,0)$. 

Another natural valued promise template associated to a crisp template $\abp$ is $\abcs$, where $\phi^{\bfa}(\tuple{a})=1$ if $\tuple{a} \in \phi^{\bfa'}$, $\phi^{\bfa}(\tuple{a})=0$ otherwise, and $c \geq s$ are the completeness and soundness parameters. PCSPs over such templates include e.g.  approximation problems for MaxCSPs, such as the following concrete problems.

\begin{example}[$\lin(c,s)$] \label{ex:valued-lin}
    Given a weighted system of linear equations over $\mathbb{Z}_2$ with exactly 3 variables in each equation, accept if there exists an assignment that satisfies a $c$-fraction of the equations (taking weights into account), and reject if there is no assignment that satisfies an $s$-fraction of the equations.

    This problem is $\PCSP(\bfa,\bfa,c,s)$ where $A=\{0,1\}$ and the signature consists (as in~\cref{ex:lin}) of two $[3]$-ary symbols $\phi_0,\phi_1$ interpreted as $\phi_i^{\bfa}(a_1,a_2,a_3)=1$ if $a_1+a_2+a_3 = i \pmod 2$ and $\phi_i^{\bfa}(a_1,a_2,a_3)=0$ otherwise.
    
    We denote this PCSP as well as the template by $\lin(c,s)$.
    Note that $\lin(1,1)$ is another formulation of $\lin$. 
\end{example}

The following maximization version of~\cref{ex:AGC}, first introduced in~\cite{NakajimaZ23},
 nicely combines the promise and valued frameworks. 

\begin{example}[Maximum 3- versus 5-coloring of graphs] \label{ex:maxAGC}
Given an edge-weighted graph $G$, the task is to accept if $G$ admits a 3-coloring with a $c$-fraction of non-monochromatic edges, and reject if $G$ does not admit a 5-coloring with an $s$-fraction of non-monochromatic edges. This is $\PCSP(\mathbf{K}_3,\mathbf{K}_5,c,s)$, where $\mathbf{K}_k$ is the $k$-clique, interpreted here as having payoff 1 on the edges and 0 on non-edges. 
\end{example}

Another example that fits in our framework is a variant of Example~\ref{ex:maxAGC} concerned with a 3- vs 5- coloring of a large induced subgraph of a given graph~\cite{Hecht23:approx}.

A gap version of~\cref{ex:label_cover}, the Gap Label Cover problem, is a starting point for many NP-hardness results in approximation.

\begin{example}[$\glc_{D,E}(c,s)$: Gap Label Cover] \label{ex:gap_label_cover}
Fix disjoint finite sets $D,E$ and rationals $1 \geq c \geq s > 0$. Given a weighted
bipartite graph with vertex set $U \cup V$ and a constraint $\pi_{uv}: D \to E$ for each edge $\{u,v\}$, accept if a $c$-fraction (taking weights into account) of the constraints can be satisfied, and reject if not even an $s$-fraction of the constraints can be satisfied. 

This problem is $\pcsp(\bfa,\bfa,c,s)$, where the sort symbols are $D$ and $E$, $A  = D \cup E$, the signature consists of all functions $\pi: D \to E$ of arity [2] ($\sort(1)=D$, $\sort(2)=E$), interpreted as $\pi^{\bfa}(d,e)=1$ if $\pi(d)=e$ and $\pi^{\bfa}(d,e)=0$ otherwise.

A consequence of the PCP theorem~\cite{Arora98:jacm-proof,Dinur07:jacm} and the Parallel Repetition theorem~\cite{Raz98} is that for every $\epsilon>0$ there exist $D,E$ such that $\glc_{D,E}(1,\epsilon)$ is NP-hard. 
\end{example}

Problems with $\bfa \leq c$ are said to have \emph{perfect completeness}.
By giving up perfect completeness in the Gap Label Cover and restricting the functions $\pi: D \to E$ to be bijections, we obtain the well-known Unique Games problem, a starting point of many conditional NP-hardness results. 

\begin{example}[Unique Games] \label{ex:unique_games}
We fix disjoint sets $D$ and $E$ such that $|D|=|E|$ and $\epsilon>0$, and define $\bfa$ as in \cref{ex:gap_label_cover} but only using bijective $\pi:D \to E$.  

The Unique Games Conjecture of Khot~\cite{Khot02stoc} states that for every $\epsilon >0$ there exist $D$ and $E$ such that $\pcsp(\bfa,\bfa,1-\epsilon,\epsilon)$ is NP-hard.
\end{example}

Nice examples where infinite and nonzero finite payoffs both appear are the vertex cover and independent set problems in graphs. While they are in some sense complementary,\footnote{A set of vertices is independent iff its complement is a vertex cover.} it is known that these two problems differ significantly with respect to approximability: vertex cover admits a $2$-approximation whereas there is no constant factor approximation for independent set. The following examples show the optimization versions of these problems.

\begin{example}[Independent Set]
\label{ex:ind_set}
An \textit{independent set} in a graph $G$ is a subset $S$ of the vertices of $G$ such that every edge of the graph is incident to \textit{at most} one vertex in $S$. 
In the Independent Set problem with parameter $1 \geq c > 0$, the task is, given a vertex-weighted graph $G$, to accept if $G$ has an independent set of fractional size at least $c$, and reject otherwise.

Independent Set fits in our framework as $\pcsp(\bfa,\bfa,c,c)$ where $1 \geq c > 0$ (the lower bound on the weight of the independent set), $A=\{0,1\}$, and the signature consists of a unary relational symbol $\phi$ interpreted as $\phi^\bfa(a)=a$ (enforcing that the fractional size of the independent set is at least $c$) and a binary relational symbol $\psi$ interpreted as $\psi^\bfa(1,1)=-\infty$ and $\psi^\bfa(a_1,a_2)=0$ for all other values of $a_1,a_2$ (enforcing that if the subset of the vertices that are assigned $1$ yields a finite payoff, then it is an independent set).
\end{example}

\begin{example}[Vertex Cover] \label{ex:vertex_cover}
A \textit{vertex cover} of a graph $G$ is a subset $S$ of the vertices of $G$ such that every edge of the graph is incident to \textit{at least} one vertex in $S$. In the Vertex Cover problem with parameter $c$, the task is to accept if a vertex-weighted graph $G$ has a vertex cover of fractional size at most $c$, and reject otherwise.

Vertex cover is a minimization problem. However, it can be phrased in our framework as $\pcsp(\bfa,\bfa,-c,-c)$,
where the domain and signature are as in \cref{ex:ind_set} but the symbols are interpreted as $\phi^\bfa(a)=-a$,   
$\psi^\bfa(0,0)=-\infty$, and $\psi^\bfa(a_1,a_2)=0$ otherwise.
\end{example}

\medskip
We now discuss several possible variations and modifications of the definition of valued PCSPs, ordered by the significance of the difference they would cause. 

First, we have decided for the maximization version of the definition. The corresponding minimization problem can be obtained by multiplying all payoff functions as well as $c$ and $s$ by $-1$ (cf. \cref{ex:vertex_cover}), so results for our version can be easily transferred to the minimization version and vice versa.

Second, note that by shifting (and/or scaling) the payoff functions in $\bfa$ and $\bfb$ and modifying $c$ and $s$ in the same way, we get an equivalent problem. It would therefore be possible to fix $c,s$, e.g., to $c=s=0$ and define a template just as a pair $\ab$. Our choice here was inspired by a more natural formulation of problems such as $\lin(c,s)$.

Third, a natural version of the definition is to require $\Phi$ to be
\emph{normalized}, that is, $w(\Phi)=1$. An instance can then be regarded as a
probability distribution on constraints; $\Phi^{\bfa}(h)$ can be interpreted as
the expected value of $\phi^{\bfa}(h\tuple{x}_i)$ when constraint
$\phi(\tuple{x}_i)$ is selected according to this distribution. Note that an
equivalent normalized instance can be obtained by dividing all the weights by $w(\Phi)$, unless $w(\Phi)=0$, i.e., all the weights are zero. Therefore this alternative formulation differs from our formulation only very slightly.

Fourth, a more substantial change would be to require that all the weights be equal, say to 1. We regard the presented version as slightly more natural. Note however that it is often the case that positive (algorithmic) results work even for the weighted version and negative (hardness) results already for the non-weighted one, by emulating weights via repeated constraints. 

Fifth, the most substantial change would be to not fix $c,s$ in advance and
rather make them part of the instance. An important intermediate choice is to fix
$s/c$: a template would be a triple $(\bfa,\bfb,\IO)$, an instance would include
$c$ (not $s$), and $\yes$ and $\no$ would be defined in the same way as in the
definition with $s=\IO c$. Such a framework includes constant factor
approximation problems for MaxCSP; for $\IO=1$ and $\bfa=\bfb$ this framework
essentially coincides with general-valued
CSPs~\cite{DBLP:journals/ai/CohenCJK06,Kolmogorov17:sicomp}. In fact, the
algebraic framework discovered for $\IO=1$ and general $\bfa,\bfb$ by
Kazda~\cite{Kazda21} was among the starting points for this work. We give basics of this framework in~\cref{app:kazda}.

\subsection{Polymorphisms} \label{subsec:valued_polymorphisms}

A natural generalization of an $N$-ary operation from $A$ to $B$ to the valued world consists of a probability distribution on $N$ and a probability distribution on a set $\mathcal{F}$ of (normal) $N$-ary operations from $A$ to $B$. In our situation $\mathcal{F}$ will be the set of all $N$-ary polymorphisms of the pair of feasibility structures corresponding to a pair $\ab$ of valued structures. We therefore denote
$$
\polfeas \ab = \pol(\feas(\bfa),\feas(\bfb))
$$
and introduce the following concept. 

\begin{definition}[Weighting] Let $\mnn M$ be a minion and $N$ a finite set. An \emph{$N$-ary weighting of $\mnn M$} is a pair 
$$
\Omega=(\OmegaI,\OmegaO) \quad \textnormal{ where } \quad \OmegaI \in \Delta N, \ \OmegaO \in \Delta \mnn{M}^{(N)}.
$$
\end{definition}

Relation-matrix pairs for valued structures are introduced in an analogous fashion as in the crisp case. Given such a pair and a weighting $\Omega$ of $\polfeas\ab$ we have two naturally associated rationals: the expected payoff in $\bfa$ of the $n$-th column, when $n$ is selected according to $\OmegaI$; and the expected payoff in $\bfb$ of the tuple obtained by applying $f$ to the rows of the matrix, when $f$ is selected according to $\OmegaO$.  

\begin{definition}[Relation-matrix pairs, input and output payoffs]
Let $\ab$ be a pair of valued $\Sigma$-structures, $\mnn M = \polfeas \ab$,  and $N$ a finite set. We define
$$
  \Mat{\bfa}{N} = \{(\phi, M) \mid \phi \in \sigma, M \in A^{\ar(\phi) \times N}, \ \forall n \in N \  \col{n}{M} \in \feas(\phi^{\bfa})\}.
$$
For an $N$-ary weighting $\Omega$ of $\mnn M$ and $(\phi,M) \in \Mat{\bfa}{N}$ we define
\arxiv{
$$
\OmegaI[\phi,M] = \Ex{n \sim \OmegaI} \,\phi^\bfa(\col{n}{M}) \quad\textnormal{and}\quad 
\OmegaO[\phi,M] = \Ex{f\sim \OmegaO} \,\phi^\bfb(f \,\rows (M)).
$$
}\conference{
$$
\OmegaI[\phi,M] = \! \! \Ex{n \sim \OmegaI} \! \! \phi^\bfa(\col{n}{M}) \quad \! \textnormal{and} \quad \!  
\OmegaO[\phi,M] = \! \! \Ex{f\sim \OmegaO} \! \! \phi^\bfb(f \,\rows (M)).
$$
}For an $N$-ary weighting $\Omega$ of $\mnn M$ and functions $\alpha:N \to \qq$, $\beta: \mnn M^{(N)} \to \qq$ we define 
$$ 
\OmegaI[\alpha] =  \Ex{n \sim \OmegaI} \, \alpha(n) 
\quad\textnormal{and}\quad 
\OmegaO[\beta] = \Ex{f\sim \OmegaO} \, \beta(f).
$$
\end{definition}

For a weighting $\Omega$, each relation-matrix pair thus gives us a point $(\OmegaI[\phi,M],\OmegaO[\phi,M])$ in the plane $\qq^2$. We call $\Omega$ a $\IO$-polymorphism if all these points lie on or above the line with slope $\IO$ going through $(c,s)$.

\begin{definition}[Polymorphisms]
   Let $\ab$ be a pair of valued $\Sigma$-structures, $\mnn M = \polfeas \ab$, and $c,s \in \qq$.  
   \begin{itemize}
   \item Let  $\IO \in \qqnonneg$. An $N$-ary weighting $\Omega$ of $\mnn M$ is
   a \emph{$\IO$-polymorphism} of $\abcs$ if 
   $$
   \forall (\phi,M) \in \Mat{\bfa}{N} \ \ \OmegaO[\phi,M] - s \geq \IO (\OmegaI[\phi,M]-c).
   $$
   \item An $N$-ary weighting $\Omega$ of $\mnn M$ is
   a \emph{polymorphism} of $\abcs$ if it is a $\IO$-polymorphism for some $\IO \in \qqnonneg$. 
    \item A finite family $(\Omega_j)_{j \in J}$ of weightings of $\mnn M$ of arities $\caln = (N_j)_{j \in J}$  is an $\caln$-ary \emph{plurimorphism}  of $\abcs$ if there exists $\IO \in \qqnonneg$ such that every $\Omega_j$ is a $\IO$-polymorphism. 
   \end{itemize}
   
   We will denote by $\iopol^{(N)} \abcs$, $\pol^{(N)} \abcs$, and $\multipol^{(\caln)} \abcs$ the sets of all $N$-ary $\IO$-polymorphism, $N$-ary polymorphisms, and $\caln$-ary plurimorphisms, respectively, 
   and by $\iopol \abcs$, $\pol \abcs$, and $\multipol \abcs$  the collections of the corresponding morphisms indexed by their arities.
\end{definition}

We remark that one could also define ``$\infty$-polymorphisms'' corresponding to vertical lines: the inequality would be $0 \geq \OmegaI[\phi,M]-c$, so having such a polymorphism is equivalent to $\bfa \leq c$. This would remove some (weak) assumptions in the following, but would cause some other inconveniences.

For a polymorphism $\Omega$, all points in $\qq^2$ determined by relation-matrix pairs lie above or on a line going through $(c,s)$ with a nonnegative slope $\IO$. In particular, these points avoid the region $R=\{(x,y) \mid x \geq c, y < s\}$ and so does any convex combination of these points (since half-planes are convex). It is easy to see that, conversely, if the convex hull of these points avoids $R$, then $\Omega$ is a polymorphism. This is phrased more generally for plurimorphisms in item \labelcref{alt-pol:iii} of the following proposition, in the language of probabilities. It is also geometrically clear that it is enough to require that the convex hulls of two points 
avoid $R$, leading to item \labelcref{alt-pol:ii}.

\begin{proposition}[Alternative definitions of plurimorphisms] \label{prop:alt-pol} 
   Let $\ab$ be a pair of valued $\Sigma$-structures, $\mnn M = \polfeas \ab$, $c,s \in \qq$, and  $(\Omega_j)_{j \in J}$ a finite family of weightings of $\mnn M$ of arities $(N_j)_{j \in J}$.  The following are equivalent. 
    \begin{enumerate}[label=\textnormal{(\roman*)}]
        \item $(\Omega_j)_{j \in J}$ is a plurimorphism of $\abcs$.
        \item \label{alt-pol:ii} Each pair $(\Omega_{j},\Omega_{j'})$ with $j,j' \in J$ is a plurimorphism of $\abcs$.
        \item \label{alt-pol:iii} For every probability distribution
          $$
          \mu \in \Delta \{(j,\phi,M) \mid j \in J, (\phi,M) \in \Mat{\bfa}{N_j}\} 
          $$
          we have that
          $$
          \Ex{ (j,\phi,M) \sim \mu} \ \OmegaI_j[\phi,M]  \geq c 
          \implies
          \Ex{ (j,\phi,M) \sim \mu} \ \OmegaO_j[\phi,M] \geq s.           
          $$
    \end{enumerate}
\end{proposition}

\begin{proof}
 We first state the items in the proposition in a geometric way using the following  notation.
 \begin{itemize}
     \item $R = \{(x,y) \in \qq^2 \mid x \geq c, y < s\}$.
     \item $H_{\IO} = \{(x,y) \in \qq^2 \mid y-s \geq \IO(x-c)\}$ for $\IO \in \qqnonneg$.
     \arxiv{
     \item $P[j,\phi,M] = (\OmegaI_j[\phi,M],\OmegaO_j[\phi,M]) \in \qq^2$ for $j \in J$, $(\phi,M) \in \Mat{\bfa}{N_j}$.
     }
     \conference{
     \item $P[j,\phi,M] = (\OmegaI_j[\phi,M],\OmegaO_j[\phi,M]) \in \qq^2$ for $j \in J$,\\ $(\phi,M) \in \Mat{\bfa}{N_j}$.
     }
 \end{itemize}
The geometric translation is as follows. We also add an additional item \labelcref{it:alt-4} that is equivalent to \labelcref{it:alt-3} with $\mu$ required to have at most two-element support.
   \begin{enumerate}[label=\textnormal{(\roman*)}] 
   \item \label{it:alt-1} 
     There exists $\IO \in \qqnonneg$ such that all the points $P[j,\phi,M]$ lie in $H_{\IO}$. 
   \item \label{it:alt-2} 
    For every $j,j' \in J$ there exists $\IO \in \qqnonneg$ such that all the points $P[j,\phi,M]$ and $P[j',\phi',M']$ lie in $H_{\kappa}$.
   \item \label{it:alt-3}  The convex hull of all the points $P[j,\phi,M]$ does not intersect $R$.
   \item \label{it:alt-4} 
   For all pairs of points $P[j,\phi,M]$, $P[j',\phi',M']$, their  convex hull (which is a line segment) does not intersect $R$.  
   \end{enumerate}

 The implications from \labelcref{it:alt-1} to \labelcref{it:alt-2} and from \labelcref{it:alt-3} to \labelcref{it:alt-4} are trivial. Note that for every $\IO \in \qqnonneg$, the halfplane $H_{\IO}$ is convex and disjoint from $R$. It follows that \labelcref{it:alt-1} implies \labelcref{it:alt-3} and \labelcref{it:alt-2} implies \labelcref{it:alt-4}. It is therefore enough to verify that \labelcref{it:alt-4} implies \labelcref{it:alt-1}.

 Assume \labelcref{it:alt-4}, in particular no point $P[j,\phi,M]=(P_x,P_y)$ lies in $R$, i.e., $P_x \geq c$ implies $P_y \geq s$. Assume also that there is some point with $P_y < s$ (thus $P_x < c$) as otherwise $\kappa=0$ proves \labelcref{it:alt-1}.
Let  $\IO \in \qqnonneg$ be the minimum value such that all the points $P[j,\phi,M]$ that lie in $H_{\infty} := \{(x,y) \in \qq^2 \mid x < c\}$ also lie in $H_{\IO} \cap H_{\infty}$, which exists since we have a finite number of points at least one of which is in $H_{\infty}$, and the sets $H_{\IO} \cap H_{\infty}$ cover $H_{\infty}$ and  get inclusion-wise larger as $\IO$ increases. Let $P[j,\phi,M]=:(P_x,P_y)$ witnesses that $\kappa$ cannot be further decreased, that is, $P_x < c$ and $P_y - s = \kappa(P_x - c)$. 

We claim that every  point $P[j',\phi',M'] =: (P'_x,P'_y)$ lies in $H_{\kappa}$, which gives \labelcref{it:alt-1}. The claim holds when this point is in $H_{\infty}$ by definition of $\kappa$, so we assume $P'_x \geq c$. Since $P_x < c$, the line segment between $(P_x,P_y)$ and $(P'_x,P'_y)$ intersects the line $\{(x,y) \in \qq^2 \mid x = c\}$ and, by \labelcref{it:alt-4}, it intersects it at $y \geq s$. Therefore the line segment, in particular the point $(P'_x,P'_y)$, is on or  above the line connecting $(P_x,P_y)$ and $(c,s)$ --- the line $\{(x,y) \in \qq^2 \mid y - s \geq \kappa(P_x-c)\}$. In other words, $(P'_x,P'_y)$ lies in $H_{\kappa}$, as claimed. 
\end{proof}

We remark that the notion of $(c,s)$-approximate polymorphism of Brown-Cohen and Raghavendra from~\cite{Brown-CohenRaghavendra15} (implied by Definitions 1.6 and 1.9 in their paper) is essentially introduced as in item 
\labelcref{alt-pol:iii} of~\cref{prop:alt-pol} (for $|J|=1$ and $\bfa=\bfb$ and uniform distribution $\OmegaI$).

\smallskip

In the crisp case, we observed that polymorphisms give us a way to combine solutions in $\bfa$ to obtain solutions in $\bfb$. Item \labelcref{alt-pol:iii} with $|J|=1$ can be used to show a valued version of this fact: if $\Omega$ is a polymorphism of $\abcs$, $\Phi$ is a normalized payoff formula over $X$, and $M \in A^{X \times N}$ is such that the expected payoff in $\bfa$ of the $n$-th column when $n \sim \OmegaI$ is at least $c$, then the expected payoff in $\bfb$ of $f \rows(M)$ when $f \sim \OmegaO$ is at least $s$. 
We formalize this in the result below, which can be seen as a valued version of \cref{prop:combining_solutions}. For simplicity, we only formulate one implication and restrict to polymorphisms (rather than  plurimorphisms).

\begin{proposition}[Polymorphisms and payoffs] \label{prop:valued_combining_solutions}
Let $\ab$ be a pair of valued $\Sigma$-structures,  $N$ a finite set, $c,s \in \qq$, $\IO \in \qqnonneg$, and $\Omega$ an $N$-ary $\IO$-polymorphism of $\abcs$. 

Then for every finite set $X$, every payoff $\Sigma$-formula $\Phi$ over $X$, and every $M \in A^{X \times N}$ such that every column is in $\feas(\Phi^{\bfa})$ we have that
$$
\Ex{f \sim \OmegaO} \ \Phi^{\bfb}(f \rows(M))  - s \,\weight(\Phi)
\geq
\IO \left(\Ex{n \sim \OmegaI} \ \Phi^{\bfa}(\col{n}{M}) - c \,\weight(\Phi)\right).
$$
In particular,
          $$
          \Ex{n \sim \OmegaI} \ \Phi^{\bfa}(\col{n}{M}) \geq c \,\weight(\Phi)
          \implies
          \Ex{f \sim \OmegaO} \ \Phi^{\bfb}(f \rows(M))  \geq s \,\weight(\Phi).           $$
\end{proposition}

\begin{proof}
   Consider $M$ and $\Phi = \sum_{i \in I} w_i\phi_i(\tuple{x}_i)$ as in the statement. 
   For $i \in I$ we define a matrix $M_i'$ as in the proof of \cref{prop:combining_solutions}, that is, $M_i' \in A^{\ar(\phi_i
   ) \times N}$ has $M(\tuple{x}_i(z),n)$ at position $(z,n)$. Recall that $\col{n}{M'_i} = \col{n}{M}\tuple{x}_i$ and $\rows(M'_i) = \rows(M)\tuple{x}_i$. 
   We have
   \begin{align*}
       \Ex{n \sim \OmegaI} \ \Phi^{\bfa}(\col{n}{M})  
         &=  \Ex{n \sim \OmegaI} \sum_{i \in I} w_i \phi_i^{\bfa}(\col{n}{M}\tuple{x}_i)      \\
         &=  \sum_{i \in I} w_i \Ex{n \sim \OmegaI} \phi_i^{\bfa}(\col{n}{M}\tuple{x}_i) \\
         &= \sum_{i \in I} w_i \Ex{n \sim \OmegaI} \phi_i^{\bfa}(\col{n}{M'_i}) \\
         &= \sum_{i \in I} w_i \OmegaI[\phi_i,M_i'] 
   \end{align*}
   and similarly 
      $$  \Ex{f \in \OmegaO} \Phi^{\bfb}(f\rows(M))
        =
         \sum_{i \in I} w_i \OmegaO[\phi_i,M_i']. 
    $$
   The inequality now follows from $\OmegaO[\phi_i,M_i'] - s \geq \IO(\OmegaI[\phi_i,M_i']-c)$.

   We remark that an alternative proof of the second part is to normalize $\Phi$ and apply item \labelcref{it:alt-3} of \cref{prop:alt-pol} to the  probability distribution on $\Mat{\bfa}{N}$ sampled by taking the relation-matrix pair $(\phi_i,M_i')$ with probability $w_i$.
   \end{proof}

Unlike in the crisp case, we do not introduce a concept of valued function
minion. The reason is that we currently do not know for sure what the right
choice of closure properties would be, so that valued function minions would be
exactly collections of plurimorphisms of templates. 
The obvious properties come from item \ref{alt-pol:ii} and the fact that $\IO$-polymorphisms are closed under convex combinations and taking minors (defined naturally, see \cref{sec:valued-homos}). 

\subsection{Canonical payoff formulas}

A natural valued refinement of canonical formulas in~\cref{prop:canonical_formula} is the
following fact. It is useful for characterizing templates and definability
(see \Cref{app:definability}),  but the main theorem requires a more
complex version of canonical formulas, presented in~\cref{prop:canonical_payoff_formulas}. 

\begin{proposition}[Canonical payoff formula] \label{prop:canonical-payoff-baby}
    Let $\ab$ be a pair of valued $\Sigma$-structures, $\mnn M = \polfeas \ab$, $c,s \in \qq$, $N$ a finite set, $\alpha: N \to \qq$, and $\beta: \mnn M^{(N)} \to \qq$. 
    Suppose further that if $\bfa \leq c$, then $\alpha \leq c$ (i.e., $\alpha(n)\leq c$ for all $n \in N$). 
    Then the following are equivalent.         
    \begin{enumerate}[label=\textnormal{(\roman*)}]
        \item \label{item:1:canonical-payoff-baby} For each $\IO \in \qqnonneg$ and each $\Omega \in \iopol^{(N)} \abcs$, 
        $\OmegaO[\beta]-s \geq \IO (\OmegaI[\alpha]-c)$.  
        \item\label{item:2:canonical-payoff-baby} There exists a payoff formula $\Phi$ over the set of variables $A^{N}$ such that
\begin{align*}
    & \forall n\in N  & \Phi^{\bfa}(\proj{N}{n}) -  c \,\weight (\Phi) & \geq \alpha(n) -c \\
    & \forall f \in \mnn M^{(N)}  &\Phi^{\bfb}(f) - s \,\weight (\Phi) & \leq \beta(f) - s \\
    &         &  \feas(\Phi^{\bfb}) &= \mnn M^{(N)}. 
\end{align*}
    \end{enumerate}
\end{proposition}

\begin{proof}
       We first observe that item \labelcref{item:2:canonical-payoff-baby} is equivalent to the following condition.
   \begin{enumerate}[label=\textnormal{(\roman*)}] \setcounter{enumi}{2}
    \item\label{item:3:canonical-payoff-baby} There exist $w_{\phi,M} \in \qqnonneg$, where $(\phi,M) \in \Mat{\bfa}{N}$,  such that the payoff formula
    \begin{equation*} \label{eq:payoff-baby-rewritten}
         \Phi = \sum_{(\phi,M) \in \Mat{\bfa}{N}} w_{\phi,M} \ \phi(\rows(M))
    \end{equation*}
    satisfies all the inequalities (so we skip the requirement on $\feas(\Phi^{\bfb})$). 
   \end{enumerate}
   Indeed, if \labelcref{item:3:canonical-payoff-baby}, then $\feas(\Phi^{\bfb})$ is, as we noted after defining interpretations of payoff formulas, the interpretation of $\bigwedge_{(\phi,M)} \phi(\rows(M))$ in $\feas(\bfb)$, which is $\mnn M^{(N)}$ by \cref{prop:canonical_formula}. On the other hand, every constraint over the set of variables $A^N$ is of the form $\phi(\rows(M))$ for some matrix $M$. If \labelcref{item:2:canonical-payoff-baby}, then the first type of inequalities ensures that only $\phi(\rows(M))$ with $M \in \Mat{\bfa}{N}$ show up in $\Phi$. By summing up weights and giving weight zero to constraints that do not show up, we obtain $\Phi$ as in \labelcref{item:3:canonical-payoff-baby}.
   
Condition \labelcref{item:3:canonical-payoff-baby} is equivalent, by definitions, to the following system of linear inequalities with 
unknowns $w_{\phi,M} \in \qqnonneg$.
      \begin{align*}
          \forall n \in N \quad \quad
          \sum_{(\phi,M)}
          -(\phi^{\bfa}(\col{n}{M})-c) w_{\phi,M}
           &\leq - (\alpha(n) - c)
          \\
          \forall f \in \mnn M^{(N)}\quad \quad
          \sum_{(\phi,M)}
           (\phi^{\bfb}(f \rows(M)) - s) w_{\phi,M}
           &\leq \beta(f) - s.
\end{align*}
Note that, since $\Mat{\bfa}{N}$ only contains matrices whose columns are in $\feas(\phi^\bfa)$, all the coefficients in the above system of inequalities are finite.

Let $F$ be the coefficient matrix of the system and $\tuple{q}$ the right-hand side vector. Note that $F$ can be naturally regarded as a rational matrix of type $(N \cup \mnn M^{(N)}) \times \Mat{\bfa}{N}$ (where the union should formally be disjoint). Schematically, the system $F\tuple{y} \leq \tuple{q}$ is

\NiceMatrixOptions
 {
   custom-line =  
    {
      command = hdashedline , 
      letter = I , 
      tikz = dashed ,
      total-width = \pgflinewidth 
    }
 }

$$\begin{pNiceArray}{ccc}[first-col, first-row,margin]
\quad \quad & \quad & (\phi,M) & \quad \\
\quad \quad  & \quad & \Vdots & \quad \\
 n \  \quad  & \Cdots & -(\phi ^{\bfa}(\col{n}{M})-c)  & \Cdots \\
\quad \quad  &  \quad & \Vdots & \quad  \\
\arrayrulecolor{gray}\hline 
\quad \quad & \quad & \Vdots & \quad \quad \\
f \ \quad &  \Cdots & \phi^{\bfb}(f\rows(M)) - s  & \Cdots \quad \\
\quad \quad &\quad & \Vdots & \quad 
\end{pNiceArray} \ \tuple{y} \ \leq \ \begin{pNiceArray}[margin]{c}
\Vdots \\
  -(\alpha(n)-c)\\
\Vdots \\
\arrayrulecolor{gray}\hline 
\Vdots \\
  \beta(f)-s\\
\Vdots\\
\end{pNiceArray}$$

By the Farkas' lemma (\cref{thm:Farkas}), the above system of inequalities (and then item \labelcref{item:2:canonical-payoff-baby}) is equivalent to 
\begin{align*} \tag{FE} \label{eq:FarkasEq}
\forall \tuple{x} \in (\qqnonneg)^{N \cup \mnn M^{(N)}} (F^T\tuple{x} \geq 0 \implies  \tuple{q}^T\tuple{x} \geq 0).
\end{align*}
Vectors $\tuple{x}$ correspond to pairs of vectors $(\tuple{x}^{\mathrm{in}} \in (\qqnonneg)^N$, $\tuple{x}^{\mathrm{out}} \in (\qqnonneg)^{\mnn M^{(N)}}$) and these can be written as $(\IOI \OmegaI,\IOO\OmegaO)$ where $\IOI, \IOO \in \qqnonneg$, $\OmegaI \in \Delta N$, and $\OmegaO \in \Delta \mnn M^{(N)}$. Condition \labelcref{eq:FarkasEq} is thus
\arxiv{
\begin{align} \label{eq:theta-in-theta-out}
\forall \IOI, &\IOO \in \qqnonneg \ 
\forall \Omega \textnormal{ $N$-ary weighting of $\mnn M$} \\ 
   &\Big(\forall (\phi,M) \in \Mat{\bfa}{N} \quad \  \IOO(\OmegaO[\phi,M]-s) \geq \IOI(\OmegaI[\phi,M] - c) \Big) \nonumber \\
  &\implies
    \IOO(\OmegaO[\beta]-s) \geq \IOI(\OmegaI[\alpha]-c). \nonumber
\end{align}
}
\conference{
\begin{align} \label{eq:theta-in-theta-out}
& \forall \IOI, \ \IOO \in \qqnonneg \ 
\forall \Omega \textnormal{ $N$-ary weighting of $\mnn M$} \\ 
& \left(\forall (\phi,M) \! \in \!  \Mat{\bfa}{N} \ \IOO(\OmegaO[\phi,M] \! -\! s) \geq \IOI(\OmegaI[\phi,M]\! - \! c) \right) \nonumber \\
  & \ \implies
  \IOO(\OmegaO[\beta]-s) \geq  \IOI(\OmegaI[\alpha]-c) . \nonumber
\end{align}
}
For $\IOO > 0$, the condition is exactly saying that for each $\IO \in \qqnonneg$ and $\Omega \in \iopol^{(N)} \abcs$, we have 
        $\OmegaO[\beta]-s \geq \IO (\OmegaI[\alpha]-c)$, where $\IO = \IOI/\IOO$, which is exactly \labelcref{item:1:canonical-payoff-baby}. It remains to observe that \labelcref{eq:theta-in-theta-out} is void for
 $\IOO = 0$ and $\IOI > 0$. Indeed,  the left-hand side of the implication in \labelcref{eq:theta-in-theta-out} holds only if $\bfa \leq c$ (by considering matrices with all the columns equal). In that case we have $\alpha \leq c$ by assumptions of the proposition, therefore the right-hand side of the implication holds as well.
\end{proof}

A characterization of templates is, similarly as in the crisp setting, a consequence of~\cref{prop:valued_combining_solutions} and the canonical formulas of~\cref{prop:canonical-payoff-baby}.

\begin{proposition}[Characterization of templates] \label{prop:valued-char-of-templates} 
    Let $\ab$ be a pair of valued $\Sigma$-structures and  $c,s \in \qq$. The following are equivalent.
    \begin{enumerate}[label=\textnormal{(\roman*)}]
        \item \label{item:1:ValuedCharTempl} $\abcs$ is a valued promise template.
        \item \label{item:2:ValuedCharTempl} For each payoff formula $\Phi$ over the set of variables $A$
        $$
        \Phi^{\bfa}(\id_A) \geq c \weight(\Phi) \implies \exists h\in B^A \ \Phi^{\bfb}(h) \geq s \weight(\Phi).
        $$
        \item \label{item:3:ValuedCharTempl} There exists a unary polymorphism of $\abcs$.
    \end{enumerate}
\end{proposition}

\begin{proof}
  The implication from \labelcref{item:3:ValuedCharTempl} to
  \labelcref{item:1:ValuedCharTempl} follows from~\cref{prop:valued_combining_solutions} for a singleton set $N$ and
  the implication from \labelcref{item:1:ValuedCharTempl} to \labelcref{item:2:ValuedCharTempl} is trivial.
  
  We now assume \labelcref{item:2:ValuedCharTempl} and aim to prove \labelcref{item:3:ValuedCharTempl} using~\cref{prop:canonical-payoff-baby} for a singleton $N = \{n\}$. Let  $\alpha(n) = c$ and choose $\beta: \polfeas \ab \to \qq$  arbitrarily so that  $\beta(f) < s$ for each $f$. Note that the assumption $\bfa \leq c \implies \alpha \leq c$ is satisfied. 
  We identify $A^N$ with $A$ so that $\Phi$ can be regarded as being over the set of variables $A$.
  
  We claim that there is no formula $\Phi$ as in item \labelcref{item:2:canonical-payoff-baby} of \cref{prop:canonical-payoff-baby}. Indeed, if there was such a formula $\Phi$, then $\Phi^{\bfa}(\proj{N}{n})-cw(\Phi) \geq \alpha(n) - c$, that is, $\Phi^{\bfa}(\id_A) \geq cw(\Phi)$. By the assumption we then get $\Phi^{\bfb}(f) \geq sw(\Phi)$ for some $f \in B^A$, in particular $f \in \feas(\Phi^{\bfb}) = \mnn M^{(N)}$. Such an $f$ cannot satisfy the inequality $\Phi^{\bfb}(f)-sw(\Phi) \leq \beta(f) - s$ by the choice of $\beta$. 

  We conclude that \labelcref{item:1:canonical-payoff-baby} of \cref{prop:canonical-payoff-baby} is not satisfied, in particular, there exists a unary polymorphism of $\abcs$.
\end{proof}

Now we state the mentioned more complex version of canonical formula required for the main theorem.
The difference is that instead of having one instance $\Phi$ as in item \labelcref{item:2:canonical-payoff-baby} of~\cref{prop:canonical-payoff-baby}, we simultaneously create multiple instances $(\Phi_j)_{j \in J}$ and allow suitable scaling and shifts. Moreover, in order to slightly simplify our formulation of the valued version of the minor condition problem, we also shift the $\alpha$ by $c$ and $\beta$ by $s$.  

\begin{proposition}[Improved canonical payoff formulas] \label{prop:canonical_payoff_formulas} 
 Let
 \begin{itemize}
 \item $\ab$ be a pair of $\Sigma$-structures, $\mnn M = \polfeas \ab$, $c,s \in \qq$, 
 \item $(N_j)_{j \in J}$ a family of finite sets (arities) with $J$ finite, 
 \item $(\alpha_j)_{j \in J}$ a family of functions $\alpha_j: N_j \to \qq$, and 
 \item $(\beta_j)_{j \in J}$ a family of functions $\beta_j: \mnn M^{(N_j)} \to \qq$.
 \end{itemize}
The following are equivalent.
\begin{enumerate}[label=\textnormal{(\roman*)}]
    \item \label{item:1:canonical_payoff_formulas} For each $\IO \in \qqnonneg$ and each family $(\Omega_j)_{j \in J}$ of $\IO$-polymorphisms of $\abcs$, 
    $$ 
    \IO \sum_{j \in J} \OmegaI_j[\alpha_j] \geq 0 \implies
    \sum_{j \in J} \OmegaO_j[\beta_j] \geq 0 .
    $$
    \item \label{item:2:canonical_payoff_formulas}
      There exist a family of payoff formulas $(\Phi_j)_{j \in J}$ with $\Phi_j$ over the set of variables $A^{N_j}$, a number $\scale \in \qqnonneg$ (scaling factor), and families of rational numbers $(\shiftI_j,\shiftO_j)_{j \in J}$ (input and output shifts) such that 
      \begin{align*}
          & \forall j \in J \ \forall n\in N_j  
            & \Phi^{\bfa}_j(\proj{N_j}{n}) - c \,\weight (\Phi_j) & \geq \scale \alpha_j(n) + \shiftI_j  \\
          & \forall j \in J \ \forall f \in \mnn M^{(N_j)}  
            & \Phi^{\bfb}_j(f) - s \,\weight (\Phi_j)& \leq \beta_j(f) - \shiftO_j  \\
          & & \sum_{j \in J} \shiftI_j & \geq 0 \\
          & & \sum_{j \in J} \shiftO_j & \geq 0 \\
          &\forall j \in J & \feas(\Phi_j^{\bfb}) &= \mnn M^{(N_j)} .
\end{align*}
\end{enumerate}
Moreover, for a fixed $\ab$, if there is an upper bound on the sizes of the $N_j$ then the payoff formulas $(\Phi_j)_{j \in J}$, the scaling factor, and the shifts can be computed from $\alpha_j,\beta_j,N_j$  (or decided that such formulas do not exist) in polynomial time in the size of the input. 
\end{proposition}

\begin{proof}
  Similarly as in the proof of~\cref{prop:canonical-payoff-baby}, 
   we first observe that item \labelcref{item:2:canonical_payoff_formulas} is equivalent to the following condition.
   \begin{enumerate}[label=\textnormal{(\roman*)}] \setcounter{enumi}{2}
    \item\label{item:3:canonical_payoff_formulas} There exist $w_{j,\phi,M} \in \qqnonneg$, where $j \in J$, $(\phi,M) \in \Mat{\bfa,N_j}$, $\scale \in \qqnonneg$, and $\shiftI_j, \shiftO_j \in \qq$ such that the payoff formulas
    $$
    \Phi_j = \sum_{(\phi,M) \in \Mat{\bfa}{N_j}} w_{j,\phi,M} \ \phi(\rows(M))
    $$
    satisfy all the inequalities (skipping the requirement on $\feas(\Phi_j^{\bfb})$). 
   \end{enumerate}

Condition \labelcref{item:3:canonical_payoff_formulas} is equivalent, by definitions, to the following system of linear inequalities with nonnegative rational unknowns $w_{j,\phi,M}$, $\scale$, $\shiftIP_j$, $\shiftIN_j$, $\shiftOP_j$, $\shiftON_j$ and (finite) rational coefficients. 
\arxiv{
\begin{align*}
          \forall j,n && \quad 
          \sum_{(\phi,M)}
          -(\phi^{\bfa}(\col{n}{M})-c) w_{j,\phi,M}
          \  & + & 
          (\shiftIP_j - \shiftIN_j) 
         \ & + & \! \! \alpha_j(n) \scale & \leq 0
          \\
          \forall j,f && \quad
          \sum_{(\phi,M)}
           (\phi^{\bfb}(f \rows(M)) - s) w_{j,\phi,M}
           \  & + & 
           (\shiftOP_j - \shiftON_j)  &  \ &
           & \ \leq \beta_j(f)\\
         &&  &  & \! \! \sum_{j \in J} -(\shiftIP_j - \shiftIN_j) & \ &  & \ \leq 0 \\
         &&  &  &  \! \! \sum_{j \in J} -(\shiftOP_j - \shiftON_j) & \ &  & \ \leq 0 
\end{align*}
}

\conference{
\begin{align*}
    \forall j,n & 
    \sum_{(\phi,M)}
    -(\phi^{\bfa}(\col{n}{M})-c) w_{j,\phi,M} + (\shiftIP_j - \shiftIN_j) +  \alpha_j(n) \scale
     \leq 0 \\
    \forall j,f &
    \sum_{(\phi,M)}
    (\phi^{\bfb}(f \rows(M)) - s) w_{j,\phi,M} + (\shiftOP_j - \shiftON_j)
    \leq \beta_j(f)\\
    & \ \ \sum_{j \in J} -(\shiftIP_j - \shiftIN_j)
     \leq 0 & \\
    & \ \ \sum_{j \in J} -(\shiftOP_j - \shiftON_j) 
     \leq 0 &
\end{align*}
}

This system $F\tuple{y} \leq \tuple{q}$ is shown in \cref{fig:improved-farkas} for the special case where $|J|=2$. 
Note that, when the cardinality of all the sets $N_j$ is bounded by a constant, the size of the system is polynomial (in fact, linear) in $|J|$. Then, when a (nonnegative) solution to $F\tuple{y} \leq \tuple{q}$ exists, it can be found in polynomial time.

\begin{figure*}
\[
\begin{pNiceArray}{ccc|ccc|cc|cc|cc|cc|c}[small,first-col, first-row,margin]
 \quad &\quad &  (j,\phi,M) & \quad &\quad &\quad &\quad & \! \delta_j^{\ii+} \! & \delta_j^{\ii-} \! \! &\quad &\quad & \! \delta_j^{\oo+} \! \! &  \! \delta_j^{\oo-} \! \! &\quad &\quad & \gamma\\
\quad &\quad &  \vdots & \quad &\quad &\quad &\quad & \vdots &\vdots &\quad &\quad &\vdots &\vdots  &\quad &\quad & \vdots\\[4pt]
 j,n \ \ & \cdots & \! \! -(\phi ^{\bfa}(\col{n}{M})-c) \! \! & \cdots & 0 &\cdots & 0 &1 &-1& 0 & 0 & 0 & 0 & 0 & 0 &\alpha_j(n)\\
 \quad &\quad &  \vdots & \quad &\quad &\quad &\quad & \vdots &\vdots &\quad &\quad &\vdots &\vdots  &\quad &\quad & \vdots\\[3pt]
\arrayrulecolor{gray} \hline 
\quad &\quad & \vdots  & \quad &\quad &\quad &\quad & \vdots &\vdots &\quad &\quad &\vdots &\vdots  &\quad &\quad & \vdots\\[3pt]
 j',n' \ & 0 &\cdots \quad \quad \quad \cdots & 0  & \cdots & \cdots & \cdots  &0 &0& 1 & -1 & 0 & 0 & 0 & 0 &\alpha_{j'}(n')\\
 \quad &\quad & \vdots  & \quad &\quad &\quad &\quad & \vdots &\vdots &\quad &\quad &\vdots &\vdots  &\quad &\quad & \vdots\\[3pt]
\arrayrulecolor{black}\hline 
\quad &\quad &  \vdots & \quad &\quad &\quad &\quad & \vdots &\vdots &\quad &\quad &\vdots &\vdots  &\quad &\quad & \vdots\\[4pt]
j,f \ \quad &  \cdots & \! \! \phi^{\bfb}(f\rows(M)) - s  \! \! & \cdots & 0 &\cdots & 0 &0 &0& 0 & 0 & 1 & -1 & 0 & 0 &0\\
\quad &\quad &  \vdots & \quad &\quad &\quad &\quad & \vdots &\vdots &\quad &\quad &\vdots &\vdots  &\quad &\quad & \vdots \\[3pt]
\arrayrulecolor{gray} \hline 
\quad &\quad & \vdots  & \quad &\quad &\quad &\quad & \vdots &\vdots &\quad &\quad &\vdots &\vdots  &\quad &\quad & \vdots \\[4pt]
 j',f' \ & 0 &\cdots \quad \quad \quad \cdots & 0  & \cdots & \cdots & \cdots  &0 &0& 0 & 0 & 0 & 0 & 1 & -1 &0 \\
 \quad &\quad & \vdots  & \quad &\quad &\quad &\quad & \vdots &\vdots &\quad &\quad &\vdots &\vdots  &\quad &\quad & \vdots\\[3pt]
\arrayrulecolor{black} \hline \\[-2pt]
\quad & 0& \cdots \quad \cdots \quad \cdots & 0 & 0 &\cdots & 0 &-1 &1& -1 & 1 & 0 & 0 & 0 & 0 &0 \\[3pt]
\quad & 0& \cdots \quad \cdots \quad \cdots & 0 & 0 &\cdots & 0 &0 &0& 0 & 0 & -1 & 1 & -1 & 1 &0\\
\end{pNiceArray} \ \tuple{y} \ \leq \ \begin{pNiceArray}[small, margin]{c}
\vdots\\[4pt]
 0\\
\vdots\\[3pt]
\arrayrulecolor{gray} \hline 
\vdots \\[4pt]
0 \\
\vdots \\[3pt]
\arrayrulecolor{black}\hline 
\vdots \\[4pt]
  \beta_j(f)\\
\vdots\\[3pt]
\arrayrulecolor{gray}\hline 
\vdots \\[4pt]
  \beta_{j'}(f')\\
\vdots\\[3pt]
\arrayrulecolor{black}\hline \\[-2pt]
0\\[3pt]
0\\
\end{pNiceArray} \]
\caption{System of inequalities for improved canonical formulas for $J=\{j,j'\}$}
\label{fig:improved-farkas}
\end{figure*}

By  Farkas' lemma (\cref{thm:Farkas}), the above system  (and then item \labelcref{item:2:canonical_payoff_formulas}) is equivalent to 
\begin{align*} \tag{FE*} \label{eq:FarkasEq-improved}
\forall \tuple{x} \in (\qqnonneg)^{K} \  (F^T\tuple{x} \geq 0 \implies  \tuple{q}^T\tuple{x} \geq 0),
\end{align*}
where $K = \bigcup_{j} (N_j \, \cup \, \mnn M^{(N_j)} ) \, \cup \, \{\ii,\oo\}$.

Such vectors $\tuple{x}$ correspond to collections of pairs of vectors $(\tuple{x}_j^{\mathrm{in}} \in (\qqnonneg)^{N_j}$, $\tuple{x}_j^{\mathrm{out}} \in (\qqnonneg)^{\mnn M^{(N_j)}}$), which can be written as $(\IOI_j \OmegaI_j,\IOO_j\OmegaO_j)$ where $\IOI_j, \IOO_j \in \qqnonneg$, $\OmegaI_j \in \Delta N_j$, and $\OmegaO_j \in \Delta \mnn M^{(N_j)}$, together with a pair of rationals $\rho^\ii, \rho^\oo \in \qqnonneg$. 
The inequalities $F^T\tuple{x} \geq 0$ can then be written as
\begin{align*}
    \forall (j,\phi,M) && \quad 
       - \IOI_j (\OmegaI_j[\phi,M] - c) + \IOO_j (\OmegaO_j[\phi,M]-s) & \geq 0 \\
    \forall j && \quad
      \IOI_j - \rho^\ii &= 0 \\
    \forall j && \quad 
      \IOO_j - \rho^\oo &=0 \\
    && \quad
      \sum_j \IOI_j \OmegaI_j[\alpha_j] & \geq 0 
\end{align*}
and $\tuple{q}^T \tuple{x} \geq 0$ as
$$
\sum_j \IOO_j \OmegaO_j[\beta_j] \geq 0.
$$
Condition \labelcref{eq:FarkasEq-improved} is thus equivalent to
\begin{align*} 
&\forall  \rho^\ii, \rho^\oo \in \qqnonneg  \ \ 
\forall (\Omega_j)_{j \in J} \textnormal{ where $\Omega_j$ is an $N_j$-ary weighting of $\mnn M$} \\ 
  &   \Big( \big( \forall j \ \forall (\phi,M) \! \in \! \Mat{\bfa}{N_j} \ \rho^\oo(\OmegaO_j[\phi,M] \! - \! s) \! \geq  \!\rho^\ii(\OmegaI_j[\phi,M] \! - \! c)  \big)  \\
& \quad \ \wedge \ \rho^\ii \sum_{j\in J} \OmegaI_j[\alpha_j] \geq 0 \Big) \\
&  \implies
   \rho^\oo\sum_{j\in J} \OmegaO_j[\beta_j] \geq 0. 
\end{align*}
Noting that the right-hand side of the implication is always satisfied when $\rho^\oo=0$, we can equivalently write (with $\kappa = \rho^\ii/\rho^\oo$) 
\begin{align*} 
\forall \IO \in \qqnonneg & \ 
\forall (\Omega_j)_{j \in J} \textnormal{ where $\Omega_j$ is an $N_j$-ary weighting of $\mnn M$} \\ 
   &  ( \forall j \ \Omega_j \in \iopol^{(N_j)} \ab ) \\
  &\implies
   \left( \kappa \sum_{j\in J} \OmegaI_j[\alpha_j] \geq 0 \ \implies \sum_{j\in J} \OmegaO_j[\beta_j] \geq 0\right), 
\end{align*} 
which is exactly item \labelcref{item:1:canonical_payoff_formulas}.
\qedhere
\end{proof}

\section{Valued minion homomorphisms and reductions} \label{sec:valued-homos}

We define a valued minion as an abstraction of a set of plurimorphisms. As discussed in the last paragraph of~\cref{subsec:valued_polymorphisms}, it is not clear what closure properties we should require, so the definition should be regarded as temporary. For now, we take a liberal approach and only require a closure property that will be useful to us: the closure under minors.

\begin{definition}[Minor of a weighting] Let $\mnn{M}$ be a minion, $N$, $N'$ finite sets, $\Omega=(\OmegaI,\OmegaO)$ an $N$-ary weighting of $\mnn{M}$, and $\pi:N \to N'$. The \textit{minor of $\Omega$ given by $\pi$}  
is the $N'$-ary weighting defined by $\Omega^{(\pi)}=(\pi(\OmegaI),\mnn M^{(\pi)} (\OmegaO))$.
\end{definition}

\begin{definition}[Valued minion] \label{def:valued_minion}
Let $\mnn M$ be a minion. A \emph{valued minion over $\mnn M$} is a collection $\vmnn M = (\vmnn M^{(\caln)})$ indexed by finite families of finite sets $\caln = (N_j)_{j \in J}$ such that 
\begin{enumerate}
    \item elements of $\vmnn M^{(\caln)}$ are families $(\Omega_j)_{j \in J}$ where each $\Omega_j$ is an $N_j$-ary weighting of $\mnn M$;
    \item $\vmnn M$ is closed under taking minors in the following sense:
    if $(\Omega_j)_{j \in J} \in \vmnn M^{(\caln)}$,  then also $(\Omega'_{j'})_{j' \in J'} \in \vmnn M^{(\caln')}$ whenever each $\Omega'_{j'}$ is a minor of some $\Omega_j$ (given by some $\pi: N_j \to N'_{j'}$). 
\end{enumerate}
A valued minion $\vmnn M$ is \emph{nontrivial} if $\vmnn M^{(N)}$ is nonempty for every nonempty $N$, and it is \textit{locally finite} if so is $\mnn M$.
\end{definition}

Observe that a minor of a $\IO$-polymorphism of a valued promise template is again a $\IO$-polymorphism. It follows that the collection of all plurimorphisms of a valued promise template is a valued minion over the minion of polymorphisms of its feasibility template.

Some polymorphisms of the feasibility template have zero probability in any $\kappa$-polymorphism of the valued template; they play no role in this sense. In the definition of valued minion homomorphisms, it is advantageous to restrict ourselves to polymorphisms that play some role,\footnote{%
In particular, this solves an issue raised  in the conference version of this paper~\cite[Example 5.4]{Barto24:lics}: the reductions between two obviously equivalent valued PCSPs could not be explained by valued minion homomorphisms, cf. \cref{ex:crips-vs-valued}.} 
for which we need the following definition.

\begin{definition} 
Let $\vmnn M$ be a valued minion over an abstract minion $\mnn M$.     The \textit{support minion} of $\vmnn M$, denoted $\supp( \vmnn M)$, is the minion given by $\supp( \vmnn M)^{(N)} =  \cup_{\Omega \in \vmnn{M}^{(N)}} \supp(\OmegaO)$ 
    (and the minor maps are restrictions of the minor maps of $\mnn M$). 
\end{definition}

Note that, since $\vmnn M$ is closed under taking minors, a minor of an element of $\supp(\vmnn M)$ is again an element thereof. Thus, the definition makes sense.

\begin{definition}[Valued minion homomorphisms] \label{def:valued_minion_homo}
    Let $\vmnn M$, $\vmnn M'$ be valued minions over  minions $\mnn M$ and $\mnn M'$, respectively. A \emph{valued minion homomorphism} $\vmnn M \to \vmnn M'$ is a probability distribution $\Xi$ on the set of minion homomorphisms $\supp(\vmnn M) \to \supp(\vmnn M')$ such that, for every finite set $J$, every family of finite sets $\caln = (N_j)_{j \in J}$, and every $(\Omega_j)_{j \in J} \in \vmnn M^{(\caln)}$, we have  $(\Xi(\Omega_j))_{j \in J} \in \vmnn M'^{(\caln)}$, where $\Xi(\Omega_j) = (\OmegaI_j, \Xi(\OmegaO_j))$.
\end{definition}

  Note that we slightly abuse the notation in the expression $\Xi(\OmegaO_j)$ since $\OmegaO_j$ is formally a probability distribution on the whole set $\mnn M^{(N_j)}$, not just $\supp(\vmnn{M})^{(N_j)}$. However, the probabilities outside of the support are zero, so this imprecision is harmless. In any case, $\Xi(\OmegaO_j)$ is the probability distribution that is sampled by sampling $\xi \sim \Xi$, sampling $f \sim \OmegaO_j$ (which is in $\supp(\vmnn{M})^{(N_j)}$ with probability 1), and computing $\xi(f)$.

We are ready to state the main theorem of this paper. 

\begin{theorem}[Reductions via valued minion homomorphism] \label{thm:valued-reduction-via-homos}
    Let $\abcs$ and $\abcsp$ be valued promise templates 
    such that the former one has a $\no$ instance.
    If there is a valued minion homomorphism from $\multipol\abcs$ to $\multipol\abcsp$, then $\PCSP\abcsp \leq \PCSP\abcs$.     
\end{theorem}

The proof uses the same strategy as in the crisp case. We introduce a valued version of the minor condition problem (VMC) and prove that each PCSP is equivalent to a VMC. Moreover, valued minion homomorphisms give us reductions between VMCs, cf. the following figure, in which $\abcs$, $\abcsp$ are valued promise templates, 
 $\mnn M = \polfeas \ab$, $\mnn M'=\polfeas \abp$,
 $\vmnn M = \multipol \abcs$, $\vmnn M' = \multipol \abcsp$,
 and $k$ is sufficiently large.

$$\begin{tikzcd}
  \pcsp \abcsp \arrow[d, leftrightarrow]  & \pcsp \abcs \arrow[d, leftrightarrow] \\
  \vmc(\vmnn M', \mnn M', k) \ar[r, ""] & \vmc(\vmnn M, \mnn M, k) 
\end{tikzcd}$$

\subsection{Valued Minor Conditions} 

\begin{restatable}[Valued Minor Condition Problem]{definition}{DefVMC} \label{def:VMC}
Given a nontrivial, locally finite valued minion $\vmnn M$ over a minion $\mnn M$ and
  an integer $k$, the Valued Minor Condition Problem for $\mnn M$, $\vmnn M$,
  and $k$, denoted by $\vmc(\mnn M,\vmnn M,k)$, is the following problem.

\begin{enumerate}
    \item[\textsf{Input}] 
    \begin{enumerate}[label=\textnormal{\arabic*.}]
     \item disjoint sets $U$ and $V$ (the sets of \emph{variables}),
     \item a set $D_x$ with $|D_x| \leq k$ for every $x \in U \cup V$ (the \emph{domain} of $x$),
     \item a set of formal expressions of the form $\pi(u) = v$, where $u \in U$, $v \in V$, and $\pi: D_u \to D_v$ (the \emph{minor conditions}), 
     \item for each $u \in U$, a pair of functions $\alpha_u:D_u \to \qq$, $\beta_u:\mnn M^{(D_u)} \to \qq$ (the \emph{input and output payoff functions}) which satisfy the following condition. 
       \begin{equation*} \tag{$\star$}\label{star-condition} 
         \forall (\Omega_u)_{u \in U} \in \vmnn M^{(D_u)_{u \in U}} \ \sum_{u \in U} \OmegaI_u[\alpha_u] \geq 0 
         \!\implies\!
         \sum_{u \in U} \OmegaO_u[\beta_u] \geq 0.\end{equation*}
    \end{enumerate}
    \item[\textsf{Output}] 
    \begin{enumerate}
        \item[\yes] if there exists a function $h$ from $U \cup V$ with $h(x) \in D_x$ (for each $x \in U \cup V$) such that, for each minor condition $\pi(u) = v$, we have $\pi(h(u)) = h(v)$, and $\sum_{u \in U} \alpha_u(h(u)) \geq 0$.
        \arxiv{
        \item[\no] if there does not exist a function $h$ from $U \cup V$ with $h(x) \in \mnn M^{(D_x)}$  such that, for each minor condition $\pi(u) = v$, we have $\mnn M^{(\pi)}(h(u)) = h(v)$, and $\sum_{u \in U} \beta_u(h(u)) \geq 0$.
        }
        \conference{
        \item[\no] if there does not exist a function $h$ from $U \cup V$ with $h(x) \in \mnn M^{(D_x)}$  such that, for each minor condition $\pi(u) = v$, we have $\mnn M^{(\pi)}(h(u)) = h(v)$, and\\$\sum_{u \in U} \beta_u(h(u)) \geq 0$.
        }
    \end{enumerate}
\end{enumerate}
\end{restatable}

Note that unlike in the crisp case, the VMC is not a PCSP over a valued promise template, at least not in an obvious way, because of the promise~\labelcref{star-condition}.

Similarly to the crisp case, an instance cannot simultaneously be a $\yes$ and $\no$ instance. Indeed, let $h$ witness that an instance of VMC is a $\yes$ instance and let $\Omega \in \vmnn M^{([1])}$ (which exists since $\vmnn M$ is nontrivial). We define $(\Omega'_u)_{u \in U} \in \vmnn M^{{(D_u)}_{u \in U}}$ by $\Omega'_u = \Omega^{(\gamma_{h(u)})}$, where $\gamma_{h(u)}$ is the mapping $[1] \to D_u$  with $1 \mapsto h(u)$. A straightforward calculation shows that $\sum_{u} \OmegaIp_u[\alpha_u] = \sum_u \alpha_u(h(u))$, which we know is at least 0. From \labelcref{star-condition} we get $\sum_u \OmegaOp_u[\beta_u] \geq 0$. The expression is equal to $\ex{f \sim \Omega} \sum_u \beta_u(f^{(\gamma_{h(u)})})$, therefore there exists $f \in \mnn{M}^{([1])}$ such that $h'(x):=f^{(\gamma_{h(x)})}$ ($x \in U \cup V$)  witnesses that the instance is not a $\no$ instance.

Finally, note that $\multipol \abcs$ for a valued promise template $\abcs$ is locally finite and nontrivial by \cref{prop:valued-char-of-templates}.

\medskip

The proof of \cref{thm:valued-reduction-via-homos} is based on three reductions: from PCSP to VMC, from VMC to PCSP, and between VMCs. The last one is the simplest.

\begin{proposition}[Between VMCs]{}\label{prop:BetweenVMCs} 
  Let $\vmnn M$, $\vmnn M'$ be nontrivial, locally finite valued minions over minions $\mnn M$ and $\mnn M'$, respectively, such that there exists a valued minion homomorphism $\vmnn M \to \vmnn M'$. Then $\vmc (\mnn M', \vmnn M', k) \leq \vmc (\mnn M, \vmnn M, k)$ for any positive integer $k$. 
\end{proposition}

\begin{proof}
Consider an instance of $\vmc(\mnn{M}',\vmnn{M}',k)$ with output functions denoted by $\beta'_u$ for each $u\in U$. Let $C= \max\{\beta'_u(f) \mid u\in U, f \in \mnn{M}^{(D_u)}\}$.
We produce an instance of
$\vmc(\mnn{M},\vmnn{M},k)$ that is unchanged except for changing the output payoff functions as follows. 
For each $u \in U$ we  
  define $\beta_u(f)=\Ex{\xi\sim\Xi} \beta'_u(\xi^{(D_u)}(f))$ for $f \in \supp(\vmnn M)^{(D_u)}$, and $\beta_u(f)=-(|CU| +1)$ for $f \in \mnn M^{(D_u)} \setminus \supp(\vmnn M)^{(D_u)}$.

We need to check that \labelcref{star-condition} is preserved. Let $(\Omega_u)_{u \in U} \in \vmnn M^{(D_u)_{u \in U}}$ and suppose that $\sum_{u \in U} \OmegaI_u[\alpha_u] \geq 0$. Since $\Xi$ is a valued minion homomorphism, $(\Xi(\Omega_u))_{u \in U} \in \vmnn M'^{(D_u)_{u \in U}}$ and hence by \labelcref{star-condition}, $\sum_{u \in U} \Xi(\OmegaO_u)[\beta'_u] \geq 0.$ 
Since $\beta_u(f) = \Ex{\xi\sim\Xi} \beta'_u (\xi^{(D_u)}(f))$ for every $f$ in the support of $\OmegaO_u$, we obtain
$$  \sum_{u \in U} \OmegaO_u[\beta_u] =  
     \sum_{u \in U}\Ex{f \sim \OmegaO_u} \Ex{\xi \sim \Xi}  \beta'_u(\xi^{(D_u)}(f)) = 
    \sum_{u \in U} \Ex{f' \sim \Xi(\OmegaO_u)} \beta'_u(f') =
    \sum_{u \in U} \Xi(\OmegaO_u)[\beta'_u] \geq 0, 
$$     
which completes the proof of \labelcref{star-condition}.

Since the completeness of the reduction is trivial, it remains to verify the soundness. Let
  $h$ be such that $\sum_{u\in U}\beta_u(h(u)) \geq 0$. By the choice of $C$, this implies that $h(u) \in \supp(\vmnn M)$ for each $u \in U$.  Then $\ex{\xi \sim \Xi} \sum_{u \in U}(\beta'_u \circ \xi^{(D_u)})(h(u)) \geq 0$, therefore there exists a
  minion homomorphism $\xi:\supp(\vmnn M) \to \supp(\vmnn M')$ such that $\sum_{u\in U}\beta'_u((\xi^{(D_u)} \circ h)(u)) \geq 0$,
  so $\xi \circ h$ is a witness that the original instance was not a $\no$ instance
  of the $\vmc$ problem.
\end{proof}

\subsection{From PCSP to VMC} \label{subsec:from_pcsp_to_vmc}

\begin{proposition}[From PCSP to VMC] \label{prop:PCSPleqVMC} 
    Let $\abcs$ be a valued promise template, $\mnn M = \polfeas \ab$, and $\vmnn M = \pol \abcs$. If $k$ is a sufficiently large integer, then
    $\PCSP \abcs \leq \vmc( \mnn M, \vmnn M, k)$.
\end{proposition}

\def\FromPCSPtoVMCreduction{
From a payoff formula $\Phi = \sum_{i \in I} w_i \phi_i(\tuple{x}_i)$ over $X$ (where $w_i \in \qqnonneg$), 
we create an instance of $\vmc(\mnn M, \vmnn M, k)$ as follows. The sets $U, V, D_x$ 
and minor conditions are created by applying the reduction in the proof of~\cref{prop:PCSPleqMC} to the crisp template $(\feas(\bfa),\feas(\bfb))$ and conjunctive formula $\bigwedge_{i \in I} \phi_i(\tuple{x}_i)$. That is, for a large enough $k$ we define these objects as follows.
  \begin{enumerate}[label=\textnormal{\arabic*.}]
      \item $U = I$, $V = X$.
      \item $D_i = \feas(\phi_i^{\bfa})$, $D_x = A_{\sort(x)}$ for each $i \in I$, $x \in X$.
      \item For each $i \in I$  and $z \in \ar(\phi_i)$, we introduce the constraint $\pi_{i,z}(i) = \tuple{x}_i(z)$, where $\pi_{i,z}$ is the domain-codomain  restriction of $\proj{\ar(\phi_i)}{z}$ to $\feas(\phi_i^{\bfa})$ and $A_{\sort(z)}$.
  \end{enumerate}    
The input and output payoff functions are defined as follows.

\begin{enumerate}[label=\textnormal{\arabic*.}] \setcounter{enumi}{3}
\item %
$\alpha_i(\tuple{a}) = w_i(\phi_i^{\bfa}(\tuple{a})-c)$,\\
$\beta_i(f) = w_i (\phi_i^{\bfb}(f \ \rows(\GM{\feas(\phi_i^{\bfa})}))-s)$\\
for each $i\in I$, $\tuple{a} \in \feas(\phi_i^{\bfa})$, and $f \in \mnn M^{(D_i)}$. 
\end{enumerate}

For this to be a valid
  instance, we require $k \geq |\feas(\phi^{\bfa})|$ and $k \geq |A_{\sort(x)}|$ for every  $\phi\in\sig$ and $x \in X$. 
  
  We need to verify condition~\labelcref{star-condition}. In fact, a stronger condition holds:
 \begin{equation*} \tag{$\star\star$} \label{stronger-star-condition}
     \forall \IO \in \qqnonneg \ 
     \forall i\in I \ 
     \forall \Omega \in \iopol^{(D_i)} \ab \   
     \OmegaO[\beta_i] \geq \IO \OmegaI[\alpha_i] 
    \end{equation*}
Notice that condition \labelcref{stronger-star-condition} is indeed stronger than \labelcref{star-condition}: assuming \labelcref{stronger-star-condition}, 
$(\Omega_i)_{i \in I} \in \vmnn M^{(D_i)_{i \in I}}$, and $\sum_{i \in I} \OmegaI_i[\alpha_i] \geq 0$, we obtain $\sum_{i \in I} \OmegaO_i[\beta_i] \geq \sum_{i \in I} \IO \OmegaI_i[\alpha_i] \geq 0$, as required. 
}

\begin{proof}
\FromPCSPtoVMCreduction

To check condition \labelcref{stronger-star-condition}, take
 $\IO$, $i$, and $\Omega$ as in that condition and denote $M = \GM{\feas(\phi_i^{\bfa})}$. We get
 \arxiv{
 \begin{align*}
 \OmegaO[\beta_i] &= \Ex{f \sim \OmegaO} w_i(\phi_i^{\bfb}(f \rows(M)) - s)
= w_i (\Ex{f \sim \OmegaO} \phi_i^{\bfb}(f \rows(M)) - s)
\\ &= w_i (\OmegaO[\phi_i,M]-s)
\geq w_i \IO (\OmegaI[\phi_i,M]-c])
= w_i \IO(\Ex{\tuple{a} \sim \OmegaI} \phi_i^{\bfa}(\col{\tuple{a}}{M}) - c)
\\ &= w_i \IO(\Ex{\tuple{a} \sim \OmegaI} \phi_i^{\bfa}(\tuple{a}) - c) 
= \IO \Ex{\tuple{a} \sim \OmegaI} w_i(\phi_i^{\bfa}(\tuple{a})-c)
=\IO \OmegaI[\alpha_i]
 \end{align*}
 }
 \conference{
 \begin{align*}
 \OmegaO[\beta_i] &= \Ex{f \sim \OmegaO} w_i(\phi_i^{\bfb}(f \rows(M)) - s)
\\ & = w_i (\Ex{f \sim \OmegaO} \phi_i^{\bfb}(f \rows(M)) - s)
= w_i (\OmegaO[\phi_i,M]-s)
\\ & \geq w_i \IO (\OmegaI[\phi_i,M]-c])
= w_i \IO(\Ex{\tuple{a} \sim \OmegaI} \phi_i^{\bfa}(\col{\tuple{a}}{M}) - c)
\\ &= w_i \IO(\Ex{\tuple{a} \sim \OmegaI} \phi_i^{\bfa}(\tuple{a}) - c) 
= \IO \Ex{\tuple{a} \sim \OmegaI} w_i(\phi_i^{\bfa}(\tuple{a})-c)
=\IO \OmegaI[\alpha_i]
 \end{align*}
 }

To show completeness of this reduction, let $h$ be such that
  $\Phi^{\bfa}(h) \geq cw(\Phi)$. Then, the following function $h'$ witnesses that the resulting $\vmc$ instance is a $\yes$ instance.
$$h'(x) = h(x) \mbox{ for $x \in V = X$} \mbox{\quad\quad and \quad\quad} h'(i) = h \tuple{x}_i \mbox{ for $i \in U = I$}.$$
Clearly $h'(x) \in D_x$ and, since $\Phi^{\bfa}(h) \geq -\infty$, we also have $h'(i) \in D_{i}$ for each $i\in I$.
For a minor condition $\pi_{i,z}(i)=\tuple{x}_i(z)$ where $z \in \ar(\phi_i)$, we have
$$
\pi_{i,z}(h'(i)) =
\proj{\ar(\phi_i)}{z}(h'(i)) = (h'(i))(z) = h \tuple{x}_i (z) =
h'(\tuple{x}_i(z)).
$$
Moreover, the input payoff functions satisfy
$$
\sum_{i\in I}\alpha_i(h'(i)) = \sum_{i\in I}w_i(\phi^\bfa_i(h\tuple{x}_i)-c) \geq c w(\Phi) -c w(\Phi) = 0.
$$

For soundness, let $h'$ be a witness that the $\vmc$ instance is not a $\no$
instance. We claim that the assignment $h: X \to B$ defined by $h(x) =
(h'(x))(\incl_x)$, where $\incl_x: A_{\sort(x)} \to A$ denotes the inclusion,
witnesses that the original $\PCSP$ instance is not a $\no$ instance, i.e., that $\sum_{i\in I} \phi^\bfb_i(h\tuple{x}_i) \geq s w(\Phi)$. 
For $i\in I$ and $f := h'(i)$, we shall show that $h\tuple{x}_i = f \ \rows(\GM{\feas(\phi_i^{\bfa})})$ by comparing coordinates. Let $z \in \ar(\phi_i)$. From the minor condition $\pi_{i,z}(i)=\tuple{x}_i(z)$ we obtain $f^{(\pi_{i,z})} = h'(\tuple{x}_i(z))$. Applying this fact, \cref{lem:canonicalMatrix}, and definitions, we obtain
\arxiv{
\begin{align*}
h\tuple{x}_i(z) &=
h(\tuple{x}_i(z)) =
(h'(\tuple{x}_i(z)))(\incl_{\tuple{x}_i(z)}) =
f^{(\pi_{i,z})}(\incl_{\tuple{x}_i(z)}) = 
f(\incl_{\tuple{x}_i(z)} \pi_{i,z}) \\ &=
f(\row{z}{\GM{\feas(\phi_i^{\bfa})}}) =
(f \ \rows( \GM{\feas(\phi_i^{\bfa})})) (z).
\end{align*}
}
\conference{
\begin{align*}
h\tuple{x}_i(z) &=
h(\tuple{x}_i(z)) =
(h'(\tuple{x}_i(z)))(\incl_{\tuple{x}_i(z)}) =
f^{(\pi_{i,z})}(\incl_{\tuple{x}_i(z)}) \\ &= 
f(\incl_{\tuple{x}_i(z)} \pi_{i,z}) =
f(\row{z}{\GM{\feas(\phi_i^{\bfa})}}) \\ &=
(f \ \rows( \GM{\feas(\phi_i^{\bfa})})) (z).
\end{align*}
}
Then, we have that
\arxiv{$$ \sum_{i\in I}w_i\phi^\bfb(h\tuple{x}_i) = \sum_{i\in I} w_i \phi^\bfb(h'(i) \ \rows( \GM{\feas(\phi_i^{\bfa})}) )= \sum_{i\in I} \beta_i(h'(i)) + s w(\Phi) \geq s w(\Phi)$$} 
\conference{
\begin{align*}
 \sum_{i\in I}w_i\phi^\bfb(h\tuple{x}_i) &  = \sum_{i\in I} w_i \phi^\bfb(h'(i) \ \rows( \GM{\feas(\phi_i^{\bfa})}) ) \\
 & = \sum_{i\in I} \beta_i(h'(i)) + s w(\Phi) \geq s w(\Phi)   
\end{align*}
}
which completes the proof.
\end{proof}

It would be possible to make a version of VMC based on \labelcref{stronger-star-condition}, $\IO$-polymorphisms for various $\IO$, and appropriately defined valued minions and homomorphism (and this is in fact what is done for the constant factor approximation setting in \cref{app:kazda} for a fixed $\IO$).
The chosen version, albeit more complicated because of the concept of plurimorphisms, requires less information about $\IO$-polymorphisms and gives a stronger reduction result.

\subsection{From VMC TO PCSP}

The idea of the reduction from $\vmc$ to $\pcsp$ is similar as in the proof of
\cref{prop:MCleqPCSP} but we use improved canonical payoff formulas
(\cref{prop:canonical_payoff_formulas}) instead of canonical formulas (\cref{prop:canonical_formula}).
A technical issue is that condition ($\star$) only guarantees condition \labelcref{item:1:canonical_payoff_formulas} in~\cref{prop:canonical_payoff_formulas} for $\IO > 0$. This causes a slight complication.

\begin{proposition}[From VMC to PCSP] \label{prop:VMCleqPCSP} 
    Let $\abcs$ be a valued promise template
    such that $\pcsp \abcs$ has a $\no$ instance, 
    $\mnn M = \polfeas \ab$, and $\vmnn M = \pol \abcs$.  For any positive integer $k$, 
    $\vmc( \mnn M, \vmnn M, k) \leq \PCSP \abcs$. 
\end{proposition}

\begin{proof}
  Consider an instance of $\vmc$ as in~\cref{def:VMC}.

  We try to find a collection of payoff formulas $(\Phi_u)_{u \in U}$ and rationals $\scale \geq 0$, $\shiftI_u$, $\shiftO_u$ such that all the properties in item \labelcref{item:2:canonical_payoff_formulas}  of~\cref{prop:canonical_payoff_formulas} are satisfied (where $J=U$ and $N_u = D_u$). By that proposition, one can find  such a collection or decide that it does not exist, in polynomial time.

  Assume first that we found such formulas.
  We define $\Phi_v$ for $v \in V$ as
  $\Phi_v = \sum_{(\phi,M) \in \Mat{\bfa}{D_v}} 0 \cdot \phi(\rows(M))$ and do the same
  as in the proof of~\cref{prop:MCleqPCSP}: First,
  we make the variable sets of all the $\Phi_x$, $x \in X := U \cup V$ disjoint,  
  take their sum,  and thus obtain a payoff formula
  $\Phi$ over the set of variables $Y = \{(\tuple{a},x) \mid x \in X, \tuple{a} \in A^{D_x}\}$. Second, we identify for each minor condition $\pi(u)=v$ and each $\tuple{a} \in A^{D_v}$ the variables $(\tuple{a}\pi,u)$ and $(\tuple{a},v)$, summing up weights accordingly, and get a payoff formula
  $\Psi$. Recall from the proof of~\cref{prop:MCleqPCSP} that
  assignments for $\Psi$ to $A$ correspond exactly to those collections $(f_x: A^{D_x} \to A)_{x \in X}$ such that
    \begin{itemize}
        \item[(M)]  $f_u^{(\pi)} = f_v$ for each minor condition $\pi(u)=v$.
    \end{itemize}
    and analogously for $B$.

  Properties in item \labelcref{item:2:canonical_payoff_formulas} of \cref{prop:canonical_payoff_formulas} guarantee
  the correctness of this reduction. Indeed, if a function $h$ witnesses that
  the $\vmc$ instance is a $\yes$ instance, then the assignment $h'$ corresponding to $(\proj{D_x}{h(x)})_{x \in X}$ satisfies the minor conditions by (by (M); see the proof of \cref{prop:MCleqPCSP}) and 
  \begin{align*}
  \Psi^{\bfa} (h') 
  &= \sum_{u \in U} \Phi_u^{\bfa}(\proj{D_u}{h(u)})
  \geq  \sum_{u \in U} (\gamma\alpha_u(h(u)) + \shiftI_u + cw(\Phi_u)) \\
  &\geq \sum_{u \in U} \gamma \alpha_u(h(u)) + c \weight(\Phi) \geq c \weight(\Phi) = c \weight(\Psi).
  \end{align*}
  On the other hand, if $\Psi^{\bfb}(h') \geq s \weight(\Psi)$, then the corresponding $(h_x)_{x \in X}$ consists of polymorphism that satisfy the minor conditions (by (M); see again the proof of \cref{prop:MCleqPCSP}) and
  \begin{align*}
  \sum_{u\in U} \beta_u(h_u) 
  &\geq \sum_{u \in U} (\Phi^{\bfb}_u(h_u) - s \weight(\Phi_u) + \shiftO_u)
  \\ &= \Psi^{\bfb}(h') - s \weight(\Psi) + \sum_{u \in U} \shiftO_u
  \geq 0.
  \end{align*}

  Assume now that there are no formulas satisfying item \labelcref{item:2:canonical_payoff_formulas} in~\cref{prop:canonical_payoff_formulas}. This implies that item \labelcref{item:1:canonical_payoff_formulas} is not satisfied; let $\IO$ and $\Omega$ witness it. From ($\star$) we get $\IO=0$ and $\sum_{u \in U} \OmegaI_u[\alpha_u]<0$. 
  Consider any $(\Omega'_u)_{u \in U}$ obtained by changing each $\OmegaI_u$ to some arbitrary probability distribution.
  Clearly, each $\Omega_u'$ is still a 0-polymorphism. Therefore, using
  ($\star$) again, we get $\sum_{u \in U} ({\Omega_u'})^{\mathrm{in}}[\alpha_u]<0$. But then our $\vmc$ instance cannot be a
  $\yes$ instance, so we can simply return any $\no$ instance to the $\pcsp$, which exists by the assumption.
\end{proof}

\section{Examples of homomorphisms} \label{sec:examples}

Examples of valued minion homomorphisms of course include minion homomorphisms for crisp templates. More precisely, if $\abzz$ and $\abzzp$ are templates such that all symbols in all four structures are interpreted as $(-\infty,0)$-valued relations and $\xi$ is a minion homomorphism from $\polfeas \ab$ to $\polfeas \abp$, then the probability distribution $\Xi$ that assigns probability one to $\xi$ is a valued minion homomorphism from $\multipol \abzz$ to $\multipol \abzzp$. Notice that, in this case, the support of  $\multipol\abzz$ coincides with $\polfeas \ab$.

A crisp template can be in fact represented in many equivalent ways. This is discussed in the following example.

\begin{example} \label{ex:crips-vs-valued}

If a valued promise template $\abcs$ has perfect completeness and soundness, that is, $\bfa \leq c$ and $\bfb \leq s$, then $\pcsp \abcs$ is clearly equivalent to $\pcsp \abp$, where $A=A'$, $B=B'$ and relational symbols are interpreted as
$$
\phi^{\bfa'} = \{ \tuple{a} \in A^{\ar(\phi)} \mid \phi^{\bfa}(\tuple{a}) = c \}, \quad 
\phi^{\bfb'} = \{ \tuple{b} \in B^{\ar(\phi)} \mid \phi^{\bfb}(\tuple{b}) = s \}.
$$
Many formally different valued templates $\abcs$ determine the same crisp template $\abp$, for instance, 
$\lin(1,1)$ from~\cref{ex:valued-lin} and $\lin$ from \Cref{ex:lin} considered $(-\infty,0)$-valued with $c=s=0$ (as above). Observe that all such templates $\abcs$ have the same $\supp(\multipol \abcs)$ (equal to $\pol \abp$) and the same plurimorphisms when restricted to the support. Therefore, the reduction between such two versions of the same problem is covered by our reduction theorem as witnessed by the valued minion homomorphism $\Xi$ that assigns probability one to the identity on the support.

We remark that this would not be the case if in \cref{def:valued_minion_homo} the minions $\mnn M$, $\mnn M'$ were not restricted to the supports $\supp(\vmnn M)$, $\supp(\vmnn M')$. For instance, there is even no minion homomorphism from  $\polfeas (\lin(1,1))$ (which consists of all Boolean functions) to $\polfeas (\lin)$ (which consists of parity functions depending on odd number of coordinates), since e.g. every minion homomorphism maps commutative binary operations (which the first minion has) to commutative binary operations (which the second minion does not have).

This issue will be absent in all the other situations in this section. For convenience we will define the probability distribution $\Xi$ on all the minion homomorphisms $\mnn M \to \mnn M'$. The $\Xi$ from the official definition is obtained by restricting the minion homomorphisms to the supports. 
\end{example}

Recall that the reduction theorem fully explains hardness for crisp CSPs (but not for crisp PCSPs~\cite{BartoBKO21}).
In the remainder of this section we discuss two types of situations in which a reduction is explained by the reduction theorem in the valued setting.

\subsection{Gadget reductions}

We start with a very simple example of a gadget reduction.

\begin{example}[$\lin(c,s) \leq \klin{5}(c,s)$, $c \geq s$] Rename relational symbols for $\lin(c,s)$ (see \cref{ex:lin}) to $\phi'_0$, $\phi'_1$ to distinguish them from relational symbols for $\klin{5}(c,s)$. 
The reduction is: replace every constraint $\phi'_i(x_1,x_2,x_3)$ in the input payoff formula by $\phi_i(x_1,x_2,x_3,x_3,x_3)$. 
\end{example}

In the example, the gadget for $\phi'_i(x_1,x_2,x_3)$ is the payoff formula $\phi_i(x_1,x_2,x_3,x_3,x_3)$. More generally, one can use more complex gadgets, possibly including additional variables. Such reductions are covered by the reduction theorem via particularly strong valued minion homomorphisms: assigning probability one to the inclusion, see \cref{app:definability}.

Sometimes a reduction can be obtained by replacing constraints by gadgets and  merging some of the additional variables, cf.~\cite{Bellare_et_al}. The following is an example
that can be explained by the reduction theorem. 

\begin{example}[$\lin(c,s) \leq \klin{4}(c,s)$, $c \geq s$] 
We replace every constraint $\phi_i(x_1,x_2,x_3)$ by $\phi_i(x_1,x_2,x_3,z)$, where $z$ is a fresh variable common to all the constraints. The reduction works since if an assignment for the new instance assigns 0 to $z$, then forgetting $z$ gives an assignment for the original instance with the same payoff; and if $z$ is assigned $1$, then we additionally flip the values $0 \leftrightarrow 1$. 
\end{example}

We verify that this reduction is explained by a valued minion homomorphism.
Let $\Xi$ assign probability one to a single $\xi$ defined as follows. For $f:\{0,1\}^N\to\{0,1\}$ 
    we define $\xi(f) = f$ if $f(0,0, \dots, 0)=0$ and $\xi(f) = t \circ f$ otherwise, where $t: \{0,1\} \to \{0,1\}$ is the transposition.
We claim that 
     $\Xi$ is a valued minion homomorphism from $\multipol(\klin{4}(c,s))$ to $\multipol(\lin(c,s))$; it even  maps $\IO$-polymorphisms to $\IO$-polymorphisms for each $\IO$. 

\begin{proof}
         Consider any relation-matrix pair $(\phi'_i,M')$ for $\lin(c,s)$ and the relation-matrix pair $(\phi_i,M)$ for $\klin{4}(c,s)$ where $M$ is $M'$ padded with a constant zero row.  
     The point $(\OmegaI[\phi_i,M],\OmegaO[\phi_i,M])$ 
     coincides with $(\OmegaI[\phi'_i,M'],\Xi(\OmegaO)[\phi'_i,M'])$, so $\Xi$ indeed maps $\IO$-polymorphisms to $\IO$-polymorphisms. To see this, observe first that  $\OmegaI[\phi_i,M]=\OmegaI[\phi'_i,M']$ since adding trailing zeroes does not affect the parity of a tuple. For $\OmegaO$, consider $f:\{0,1\}^N\to\{0,1\}$. If $f(0,\ldots,0)=0$, we have that $\phi'_i(\xi(f) \rows(M')) = \phi'_i(f \rows(M')) = \phi_i(f \rows(M))$. On the other hand,  if $f(0,\ldots,0)=1$, we have $\phi_i'(\xi(f) \rows(M')) = 1 - \phi_i'(f \rows(M')) = \phi_i(f \rows(M))$. Noticing that $\Xi(\OmegaO)[\phi'_i,M']=\Ex{f\sim\OmegaO}\phi'_i(\xi(f)\rows(M'))$ completes the proof.
\end{proof}

A systematic investigation of gadget-like reductions is a task that we leave for future work.

\subsection{Homomorphisms to Gap Label Cover}

Recall from~\cref{ex:gap_label_cover} that for every $1 \geq \epsilon > 0$, there exist $D,E$ such that $\glc_{D,E}(1,\epsilon)$ is NP-hard. This result is a starting point for many inapproximability results, including those in an influential paper by H{\aa}stad~\cite{DBLP:journals/jacm/Hastad01}.

The following proposition gives a sufficient condition for a reduction from the Gap Label Cover, in particular, it isolates the core of H{\aa}stad's results. The reduction and its correctness is simple (similar to the reduction from MC to PCSP in~\cref{prop:MCleqPCSP}) and not needed in this paper, so we leave it to the reader.

The statement uses instances over the set of variables $A^D \cup A^E$. Note that assignments $A^D \cup A^E \to A'$ exactly correspond to pairs of functions $(f_D: A^D \to A', f_E:A^E \to A')$ and we write such assignments in this way. 

\def\RedToGLCassumptions{
  \begin{itemize}
      \item mappings $\Lambda_D: \mnn{M}^{(D)} \to \Delta D$ and $\Lambda_E: \mnn{M}^{(E)} \to \Delta E$, 
      \item for every $\pi: D \to E$ a normalized payoff formula $\Phi_{\pi}$ over the set of variables $A^{D} \cup A^{E}$, and 
      a linear nondecreasing function $\gamma_{\pi}: \rr \to \rr$ with $\gamma_{\pi}(s) \geq \epsilon$
  \end{itemize}
  such that for every $\pi: D \to E$
  \begin{enumerate}[label=\textnormal{\arabic*.}]
      \item \label{red_from_label_cover_i} $\Phi^{\bfa}_{\pi}(\proj{D}{d},\proj{E}{\pi(d)}) \geq c$ for every $d \in D$, and
      \item \label{red_from_label_cover_ii} $\gamma_{\pi}(\Phi^{\bfb}_{\pi}(f_D,f_E)) \leq \ex{\substack{d \sim \Lambda_D(f_D)\\e\sim \Lambda_E(f_E)}} \pi(d,e)$ for every $f_D \in \mnn{M}^{(D)}$, $f_E \in \mnn{M}^{(E)}$.
  \end{enumerate}
}

\begin{proposition}[Reductions from GLC] \label{prop:red_from_label_cover}
  Let $\abcs$ be a valued promise template, $\mnn{M} = \polfeas \ab$, $D$ and $E$ finite disjoint sets, and $\epsilon \in \rr$. Suppose that there exist
  \RedToGLCassumptions
  Then $\glc_{D,E}(1, \epsilon) \leq \PCSP \abcs$.
\end{proposition}

\begin{example}[$\lin(1-\delta,1/2+\delta)$]
  The first inapproximability result in~\cite{DBLP:journals/jacm/Hastad01} follows from the fact that for each $1/4 \geq \delta>0$, there exists $\epsilon$ (namely $16\delta^3$) such that the template
  $\lin(1-\delta,1/2+\delta)$ satisfies the conditions of~\cref{prop:red_from_label_cover} for every $D$ and $E$.

  The mapping $\Lambda_D$ (and similarly $\Lambda_E$) is a composition of two mappings. The first one is ``folding'' from $\mnn M^{(D)}$ to the set $\mnn{F}$ of \emph{folded} functions, i.e., those satisfying $f(\tuple{a}) = 1 - f(1-\tuple{a})$ (where $(1-\tuple{a})(z)=1-\tuple{a}(z)$ at each coordinate $z$). The second one assigns to $f \in \mnn{F}$ a probability distribution on $D$ based on the size of the Fourier coefficients of $f$. The definition of $\Phi_{\pi}$ is according to the ``long code test'' in~\cite{DBLP:journals/jacm/Hastad01}.
  Details are worked out in~\cref{app:hastad}.
\end{example}

We will show in~\cref{thm:homo_to_glc_sufficient_condition} that the reduction in~\cref{prop:red_from_label_cover} is explained by a valued minion homomorphism. But first we spell out and somewhat simplify the condition for homomorphisms into the plurimorphisms of $\glc_{D,E}(1,\epsilon)$. The simplification is that we only need to consider polymorphisms, not plurimorphisms; it follows from the proof that this is always the case when considering homomorphisms to PCSPs with perfect completeness.

Note that $N$-ary functions in $\polfeas(\glc_{D,E}(1,\epsilon))$ correspond exactly to pairs $(p_D: D^N \to D,p_E: E^N \to E)$, and every pair correspond to such a function since all tuples are feasible.

\begin{proposition}[Simplified homomorphisms to GLC] \label{prop:homo_to_glc_simplified} 
  Let $\abcs$ be a valued promise template, $\mnn{M} = \polfeas \ab$, $D$, $E$ finite sets, $\epsilon \in \rr$, and $\Xi$ a probability distribution on the set of minion homomorphisms $\mnn{M} \to \polfeas(\glc_{D,E}(1,\epsilon))$. The following are equivalent.
  \begin{enumerate}[label=\textnormal{(\roman*)}]
      \item \label{item:1:homo_to_glc_simplified} $\Xi$ is a valued minion homomorphism from $\multipol \abcs$ to $\multipol (\glc_{D,E}(1, \epsilon))$.
      \item \label{item:2:homo_to_glc_simplified} For every $N \in \FinSet$, $\Omega \in \pol^{(N)} \abcs$, $\pi: D \to E$, $\tuple{d} \in D^N$, and $\tuple{e} \in E^N$
      \arxiv{
      $$
      \forall n \in \supp(\OmegaI) \  \pi(\tuple d(n)) = \tuple e(n)
      \quad \implies \quad 
      \Ex{(p_D,p_E) \sim \Xi(\OmegaO)} \ \pi(p_D(\tuple{d}),p_E(\tuple{e})) \geq \epsilon.
      $$
      }
      \conference{
      $$
      \forall n \in \supp(\OmegaI) \  \pi(\tuple d(n)) = \tuple e(n)
      \ \Rightarrow  \quad \quad
      \Ex{\mathclap{ \quad (p_D,p_E) \sim \Xi(\OmegaO)}}  \ \ \ \pi(p_D(\tuple{d}),p_E(\tuple{e})) \geq \epsilon.
      $$
      }
  \end{enumerate}
\end{proposition}

\begin{proof}
The direction from \labelcref{item:1:homo_to_glc_simplified}~to~\labelcref{item:2:homo_to_glc_simplified} is straightforward: 
 we use item \labelcref{alt-pol:iii} in \cref{prop:alt-pol} for the polymorphism $\Omega \in \pol^{(N)}\abcs$ (pedantically, the plurimorphism consisting of a single polymorphism)  and the probability distribution giving probability one to the relation-matrix pair $(\pi,M)$, where $M$ is the matrix whose rows are $\tuple{d}$ and $\tuple{e}$.

For the opposite direction, we note that if $(\Omega_j)_{j\in J}$ is a plurimorphism of $\abcs$ and $\mu$ is a probability distribution as in item \labelcref{alt-pol:iii} of \cref{prop:alt-pol} which satisfies $\Ex{(j,\pi,M)\sim\mu}\OmegaI_j[\pi,M]\geq 1,$ then, by perfect completeness of $\glc_{D,E}(1,\epsilon)$, the left-hand side of the implication in \labelcref{item:2:homo_to_glc_simplified} must be satisfied for every $(j,\pi,M)\in\supp(\mu)$ with $\rows(M)=(\tuple{d},\tuple{e})$. The assumption thus gives us $\Xi(\OmegaO_j)[\pi,M] \geq \epsilon$ and this inequality remains true when we take expectation over $(j,\pi,M) \sim \mu$, obtaining $\Xi( (\Omega_j)_{j\in J}) \in \multipol(\glc_{D,E}(1, \epsilon))$, as required.
\end{proof}

\begin{theorem}[Reductions from GLC via homomorphisms] \label{thm:homo_to_glc_sufficient_condition} 
   Under the assumptions of~\cref{prop:red_from_label_cover}, there exists
   a valued minion homomorphism from $\multipol \abcs$ to $\multipol (\glc_{D,E}(1,\epsilon))$.
\end{theorem}

\def\RedToGLCconstruction{
    We start by defining a probability distribution $\Xi$ on the set of minion homomorphisms $\mnn{M} \to \polfeas(\glc_{D,E}(1,\epsilon))$. 
    A minion homomorphism $\xi$ is sampled from $\Xi$ as follows. 
    \begin{itemize}
        \item Pick $\lambda_D: \mnn{M}^{(D)} \to D$ by sampling $\lambda_D(f) \in D$ according to $\Lambda_D(f)$, independently for each $f \in \mnn{M}^{(D)}$.
        \item Pick $\lambda_E: \mnn{M}^{(E)} \to E$ similarly using $\Lambda_E(f)$.
        \item Define $\xi$ by $\xi^{(N)}(f) = (\xi^{(N)}_D(f),\xi^{(N)}_E(f))$ for every $N \in \FinSet$ and $f \in \mnn{M}^{(N)}$, where for $\tuple{d} \in D^N$, $\tuple{e} \in E^N$ we define
        $$
        \xi^{(N)}_D(f)(\tuple{d}) = \lambda_D(f^{(\tuple{d})}), \quad 
        \xi^{(N)}_E(f)(\tuple{e}) = \lambda_E(f^{(\tuple{e})}).
        $$
    \end{itemize}
}
\begin{proof}
Let $\Lambda_D, \Lambda_E$, $\Phi_{\pi}$, $\gamma_\pi$ be as in~\cref{prop:red_from_label_cover}.
\RedToGLCconstruction

   Note that $\xi$ indeed preserves minors: for any $\pi:N\to N'$, $f\in \mnn{M}^{(N)}$ and $\tuple{d} \in D^{N'}$ we have
    \arxiv{$$(\xi^{(N)}_D(f))^{(\pi)}(\tuple{d}) = \xi^{(N)}_D(f)(\tuple{d}\pi) =\lambda_D(f^{(\tuple{d}\pi)})= \lambda_D(f^{(\pi)^{(\tuple{d})}}) =\xi^{(N')}_D(f^{(\pi)})(\tuple{d})$$}\conference{\begin{align*}
        (\xi^{(N)}_D(f))^{(\pi)}(\tuple{d}) & = \xi^{(N)}_D(f)(\tuple{d}\pi) =\lambda_D(f^{(\tuple{d}\pi)}) \\ &= \lambda_D(f^{(\pi)^{(\tuple{d})}}) =\xi^{(N')}_D(f^{(\pi)})(\tuple{d})
    \end{align*}}and similarly for $\xi_E$, therefore $(\xi^{(N)}(f))^{(\pi)} = \xi^{(N')}(f^{(\pi)})$.

   We remark that the construction of $\xi$ from $\lambda = (\lambda_D,\lambda_E)$ is not ad hoc: every minion homomorphism $\xi$ can be constructed in this way. We refer to~\cite[Lemma 4.4]{BartoBKO21} for details in the one-sorted case.

   To see that $\Xi$ is a valued minion homomorphism, we verify the condition in item \labelcref{item:2:homo_to_glc_simplified} of~\cref{prop:homo_to_glc_simplified}. Fix $N \in \FinSet$, $\Omega \in \pol^{(N)} \abcs$, $\pi: D \to E$, $\tuple{d} \in D^N$, and $\tuple{e} \in E^N$. We assume that $\pi(\tuple d(n)) = \tuple e(n)$ for every $n \in \supp(\OmegaI)$ and aim to
     prove the inequality
    $$
    \epsilon \leq \Ex{(p_D,p_E) \sim \Xi(\OmegaO)} \pi(p_D(\tuple{d}),p_E(\tuple{e})).
    $$
    Since $\Xi(\OmegaO)$ is $\Xi$ applied to the probability distribution $\Omega$, and $\xi^{(N)}(f) = (\xi^{(N)}_D(f),\xi^{(N)}_E(f))$, the expression on the right hand side is equal to 
    $$
    \Ex{f \sim \OmegaO} \Ex{\xi \in \Xi} \pi(\xi^{(N)}_D(f)(\tuple{d}),\xi^{(N)}_E(f)(\tuple{e})).
    $$
    By definition of $\xi$, we have $\xi^{(N)}_D(f)(\tuple{d}) = \lambda_D(f^{(\tuple{d})})$ and similarly for $E$. By definition of $\lambda_D$ and $\lambda_E$ we can then simplify the expression to
    \begin{equation} \label{eq:ss}
    \Ex{f \sim \OmegaO} \Ex{\substack{d' \sim \Lambda_D(f^{(\tuple{d})}) \\ e' \sim \Lambda_E(f^{(\tuple{e})})}} \pi(d',e').
    \end{equation}

  In order to verify that this value is at least $\epsilon$, we apply item \labelcref{alt-pol:iii} in~\cref{prop:alt-pol} to a probability distribution $\mu$ on $\Mat{\bfa}{N}$ that we define using $\Phi := \Phi_{\pi} = \sum_{i \in I} w_i \phi_i(\tuple{x}_i)$ as follows. 
    \begin{itemize}
        \item Pick $i \in I$ with probability $w_i$ (recall that $\Phi$ is normalized). 
        \item Define $(\phi,M) = (\phi_i,M_i)$ where for $z \in \ar(\phi)$ we set $\row{z}{M_i} = \tuple{x}_i(z)\circ \tuple{d}$ if $\tuple{x}_i(z) \in A^D$ and $\row{z}{M_i} = \tuple{x}_i(z) \circ \tuple{e}$ if $\tuple{x}_i(z) \in A^E$.
    \end{itemize}
    We now verify that $\ex{(\phi,M) \sim \mu} \OmegaI[\phi,M] \geq c$. Consider
    any $i \in I$ and $n \in \supp(\OmegaI)$. Note that, for any $z \in
    \ar(\phi_i)$, the $z$-th component of the $n$-th column of $M_i$ is $\tuple{x}_i(z)(\tuple{d}(n))$ if $\tuple{x}_i(z) \in A^D$ and $\tuple{x}_i(z)(\tuple{e}(n))$ if $\tuple{x}_i(z) \in A^E$, which is equal
    to the $z$-th component of $(\proj{D}{\tuple{d}(n)},\proj{E}{\tuple{e}(n)}) \circ \tuple x_i$. Moreover, by the assumptions on $\tuple{d},\tuple{e},\pi$ we have $\pi(\tuple{d}(n))=\tuple{e}(n)$. We obtain
    \begin{align*}
    \Ex{(\phi,M) \sim \mu} \ \OmegaI[\phi,M] &=
    \Ex{n  \sim \OmegaI} \ \Ex{(\phi,M) \sim \mu} \ \phi^{\bfa}(\col{n}{M}) \\
    &= \Ex{n  \sim \OmegaI} \ \sum_{i \in I} \ w_i \phi^{\bfa}(\col{n}{M_i}) \\
    &= \Ex{n \in \OmegaI} \ \sum_{i \in I} \ w_i \phi^{\bfa}((\proj{D}{\tuple{d}(n)},\proj{E}{\tuple{e}(n)}) \circ \tuple x_i) \\
    &= \Ex{n \sim \OmegaI} \ \Phi^{\bfa}(\proj{D}{\tuple{d}(n)},\proj{E}{\pi(\tuple{d}(n))}),
    \end{align*}
    which is indeed at least $c$ by the assumptions (namely, property \labelcref{red_from_label_cover_i} in~\cref{prop:red_from_label_cover}).

    By item \labelcref{alt-pol:iii} in~\cref{prop:alt-pol} we  obtain $s \leq \ex{(\phi,M) \sim \mu} \ \OmegaO[\phi,M]$.
    By definitions $f \circ \rows(M_i) = (f^{(\tuple{d})},f^{(\tuple{e})}) \circ
    \tuple{x}_i$ for any $f \in \mnn{M}$ and $i \in I$. Indeed, assuming
    $\tuple{x}_i(z) \in A^D$, the $z$-th component of the left-hand side is $f(\tuple{x}_i(z) \circ \tuple{d})$ by definition of $M_i$, which is equal to  $f^{(\tuple{d})} (\tuple{x}_i(z))$ by definition of minors; and similarly if $\tuple{x}_i(z) \in A^E$. Therefore, we get    
    \begin{align*}
    s &\leq \Ex{(\phi,M) \sim \mu} \ \OmegaO[\phi,M]  \\
    &= \Ex{f \sim \OmegaO} \ \Ex{(\phi,M) \sim \mu} \  \phi^{\bfb}(f \circ \rows (M)) \\
    &= \Ex{f \sim \OmegaO} \ \sum_{i \in I} w_i \phi^{\bfb}((f^{(\tuple{d})},f^{(\tuple{e})}) \circ \tuple{x}_i) \\
    &= \Ex{f \sim \OmegaO} \ \Phi^{\bfb}(f^{(\tuple{d})},f^{(\tuple{e})})
    \end{align*}
    By applying the nondecreasing linear function $\gamma := \gamma_{\pi}$, using linearity of expectation, and property \labelcref{red_from_label_cover_ii} in~\cref{prop:red_from_label_cover} we get
    \begin{align*}
        \gamma(s) &\leq  
        \gamma \left(\ex{f \sim \OmegaO} \ \Phi^{\bfb}(f^{(\tuple{d})},f^{(\tuple{e})})\right) \\
        &=\Ex{f \sim \OmegaO} \ \gamma(\Phi^{\bfb}(f^{(\tuple{d})},f^{(\tuple{e})})) \\
        &\leq \Ex{f \sim \OmegaO} \ \Ex{\substack{d' \sim \Lambda_D(f^{(\tuple d)})\\e' \sim \Lambda_E(f^{(\tuple e)})}}\ \pi(d',e'),
    \end{align*}
    therefore \labelcref{eq:ss} is indeed greater than or equal to $\gamma(s) \geq \epsilon$. 
 \end{proof}

\section{Conclusion}

Our main result, \cref{thm:valued-reduction-via-homos}, shows that computational complexity is determined by symmetries in the vast framework of valued PCSPs. We see this result as a step towards the general goal of providing uniform descriptions of algorithms, tractability boundaries, and reductions. 

Crisp non-promise CSPs already include many important combinatorial problems. Valued PCSPs generalize this framework in two directions: towards approximation (promises) and optimization (values). A further vast enlargement in the combinatorial direction would be provided by incorporating interesting classes of infinite structures. In fact, \cite{ViolaZ21,BodirskySL24:lics} already contribute to this project in the constant factor setting.

There are however many basic questions and theory-building tasks left open already for finite-domain valued PCSPs: 
 to characterize gadget reductions (or versions of definability) in terms of polymorphisms or plurimorphisms;  
to characterize plurimorphism valued minions of templates;   
to clarify whether plurimorphisms are necessary to determine computational complexity or enough information is provided already by polymorphisms;
to develop methods for proving nonexistence of homomorphisms; 
to revisit the valued CSP dichotomy  without fixed threshold~\cite{Kolmogorov17:sicomp} and Raghavendra's result on unique games hardness of approximation for all MaxCSPs~\cite{Raghavendra08:everycsp}; 
among others.
An interesting special case for a full complexity classification is the valued non-promise CSPs with fixed threshold. 

The most exciting (and likely challenging) research goal to us is to improve the reduction theorem so that it explains the PCP theorem~\cite{Dinur07:jacm} (hardness of Gap Label Cover) or even the Unique Games Conjecture~\cite{Khot02stoc} (hardness of Unique Games), or some special cases  such as the 2-to-2 Conjecture (now theorem \cite{Khot23:annals}). Crucially, both Gap Label Cover and Unique Games are within our framework, and we can now thus at least specify the aim: to weaken the concept of valued minion homomorphism so that, e.g., plurimorphisms of $\glc_{D,E}(1, \epsilon)$ homomorphically map to the projection minion. A reason for cautious optimism is the recent ``Baby PCP'' paper~\cite{BartoK22gap} that contributes to this effort in the crisp setting.

%% file: appendix.tex
\section{Definability} \label{app:definability}

The starting point of the algebraic approach to the CSP was the discovery~\cite{Jeavons97:closure}  that so-called primitive positive (pp-) definability of a CSP template in another template, which entails a simple gadget reduction between the corresponding CSPs, can be characterized in terms of polymorphisms.

This fact can be generalized to crisp PCSPs~~\cite{pippenger2002galois,BG21}.
The purpose of this appendix is to show how this first step generalizes to valued PCSPs. We again start by discussing the crisp setting and then move on to the valued one.

\subsection{Crisp case}

Informally, a pp-definition of a relation in a structure $\bfa$ is a definition that uses relations in $\bfa$, the equality relation, conjunction, and existential quantification. For example, if the signature of $\bfa$ contains a binary symbol $\phi$ and a ternary symbol $\phi'$, we may pp-define a ternary relation $\psi$ by the formula
$$
\psi(x,y,z) \quad \equiv \quad 
\exists u,v \ \phi(y,x) \wedge \phi(y,v) \wedge \phi'(v,z,u) \wedge (u=v) \wedge (x=z).
$$
The equalities can be eliminated by merging variables unless they are between free variables. In the above example, we can define the same relation by the formula
$$
(x=z) \wedge \left(\exists u \ \phi(y,x) \wedge \phi(y,u) \wedge \phi'(u,x,u)\right).
$$
A convenient way to formally deal with existential quantification and equalities between free variables is to define a pp-definition of a relation $\psi$ as a conjunctive formula $\Phi$ over a set of variables $X$ together with a mapping $\iota: \ar(\psi) \to X$  which describes how the coordinates of $\psi$ are matched to the free variables of $\Phi$. In the above example, let $\ar(\psi)=\{1,2,3\}$ and take 
$$
X = \{x,y,u\}, \quad
\Phi = \phi(y,x) \wedge \phi(y,u) \wedge \phi'(u,x,u) , \quad
\iota(1) = x, \iota(2)=y, \iota(3)=x.
$$
The relation defined in $\bfa$ by $\Phi$ and $\iota$ is  
$$
\psi = \{\tuple{a}\iota \mid \exists \tuple{a} \in A^X \ \Phi^{\bfa}(\tuple{a})\}.
$$

We say that a CSP template $(\bfa',\bfa')$ is pp-definable from a CSP template $(\bfa,\bfa)$, where $A=A'$, if each relation in $\bfa'$ is pp-definable from $\bfa$. Crucially, this happens if, and only if, every polymorphism of $(\bfa,\bfa)$ is a also a polymorphism of $(\bfa',\bfa')$. Equivalently, we have a particularly simple minion homomorphism $\pol (\bfa,\bfa) \to \pol (\bfa',\bfa')$, the inclusion. We remark that the polynomial reduction from $\PCSP(\bfa',\bfa')$ to $\PCSP(\bfa,\bfa)$ can then be directly done by replacing conjuncts of the input formula by their pp-definitions (gadgets) $\Phi, \iota$. 

\cref{prop:pp-def-crisp} proves a generalization of this fact to promise templates. The appropriate generalization of pp-definability is as follows: a PCSP template $(\bfa',\bfb')$ is pp-definable from $(\bfa,\bfb)$ if for each symbol $\psi$ in the signature of $(\bfa',\bfb')$ there is $\Phi$ and $\iota$ such that the relation defined in $\bfa$ by $\Phi,\iota$ contains $\psi^{\bfa'}$, and the relation defined in $\bfb$ by $\Phi,\iota$ is contained in $\psi^{\bfb'}$.

The formula $\Phi$, which we need to produce for each $\psi$ assuming the inclusion of polymorphism minions, is a canonical formula from \cref{prop:canonical_formula}. We restate the proposition for convenience.

\begin{proposition}[Canonical formula, restated] \label{prop:canonical-formula-restated} 
For every pair $\ab$ of finite $\Sigma$-structures 
and $N$ a finite set, there exists a $\Sigma$-formula $\Phi$
    over the set of variables $A^N$ that satisfies
  \begin{itemize}
      \item $\Phi^{\bfa}(\proj{N}{n})$ for every $n \in N$, and
      \item $\Phi^{\bfb} = \pol^{(N)} \ab$. 
\end{itemize}
\end{proposition}

\begin{proposition}[Polymorphisms and pp-definitions] \label{prop:pp-def-crisp}
    Let $\ab$ and $\abp$ be promise templates in signatures $\Sigma, \Sigma'$ such that $A=A'$ and $B=B'$. The following are equivalent.
        \begin{enumerate}[label=\textnormal{(\roman*)}, ref=(\roman*)]
        \item
        $\abp$ is pp-definable from $\ab$, that is, for every $\psi \in \Sigma'$ of arity $Z=\ar(\psi)$ there exist a $\Sigma$-formula $\Phi$ over a  set of variables $X$ and a mapping $\iota \colon Z \to X$ such that
    \begin{itemize}
        \item $\forall \tuple{a}' \in A^{Z} \quad 
          \psi^{\bfa'}(\tuple{a}') \ \Rightarrow \ 
          \left(\exists \tuple{a} \in A^{X} \ \tuple{a}'=\tuple{a}\iota \,\wedge\,  \Phi^{\bfa}(\tuple{a})\right) $, and 
        \item $\forall \tuple{b}' \in B^{Z} \quad 
          \psi^{\bfb'}(\tuple{b}') \ \Leftarrow \ 
          \left(\exists \tuple{b} \in B^{X} \ \tuple{b}' = \tuple{b}\iota \,\wedge\, \Phi^{\bfb}(\tuple{b})  \right)$.
    \end{itemize}
        
        \item
        $\pol \ab \subseteq \pol \abp$.
    \end{enumerate}

\end{proposition}

\begin{proof}
For the forward direction, assume that $\abp$ is pp-definable from $\ab$, take $f \in \pol^{(N)} \ab$, and consider $(\psi,M') \in \Mat{\bfa'}{N}$ with arity $Z = \ar(\psi)$. We need to show that $\psi^{\bfb'}(f \,\rows(M'))$.

Let $X$, $\iota$, $\Phi$ be the witnesses of pp-definability of $\psi$ as in the statement.
All columns of $M'$ are in $\psi^{\bfa'}$, therefore by the first item there exists a matrix $M \in A^{X \times N}$ such that, for all $n \in N$, $\col{n}{M}$ is in $\Phi^{\bfa}$ and $\col{n}{M'} = \col{n}{M} \iota$; written equivalently, $\rows(M') = \rows(M) \iota$.
Since $f$ is a polymorphism of $\ab$, we get from~\Cref{prop:combining_solutions} that $\tuple{b} := f \, \rows(M)$ is in $\Phi^{\bfb}$. It follows from the second item that $\tuple{b}\iota = f \, \rows(M) \iota = f\, \rows(M')$ is in $\psi^{\bf B'}$, as required.

For the backward direction, assume that $\pol \ab \subseteq \pol \abp$ and consider $\psi \in \Sigma'$ of arity $Z$. We need to find $\Phi$ and $\iota$ satisfying the two items.

We set $N = \psi^{\bfa'}$, take the canonical formula $\Phi$ over the set of variables $X := A^N$ from~\cref{prop:canonical-formula-restated}, and define $\iota = \rows \ \GM{\psi^{\bfa'}}$ (note that indeed $\iota$ is a mapping from $Z = \ar(\psi)$ to $X = A^{\psi^{\bfa'}}$ as it should be). To verify the first item, let $\tuple{a}' \in A^Z$ be in $\psi^{\bfa'}$ and take $\tuple{a} = \proj{N}{\tuple{a}'}$. By the first item in~\cref{prop:canonical-formula-restated} this tuple is in $\Phi^{\bfa}$. We have $\tuple{a} \iota = \proj{N}{\tuple{a}'} \, \rows (\GM{\psi^{\bfa'}}) = \tuple{a}'$, as required. For the second item, let $\tuple{b} \in B^X$. From the second item in~\cref{prop:canonical-formula-restated} we learn that $\tuple{b}$ is a polymorphism of $\ab$, so it is also a polymorphism of $\abp$ by the assumption. 
Therefore $\tuple{b} \, \rows(\GM{\psi^{\bfa'}}) = \tuple{b} \iota$ is in $\psi^{\bfb'}$, which was to be demonstrated. 
\end{proof}

\subsection{Valued case}

A generalization of pp-definitions to the valued setting is obtained by replacing conjunction by sum and existential quantification by maximization (and dealing with equalities, weights, and thresholds). An analogue of the crisp running example is the ternary payoff function defined by
$$
(x=z) \wedge \left(\max_u \ 2\phi(y,x) + 5\phi(y,u) + 6\phi'(u,x,u)\right)
$$
formally captured by a payoff formula $\Phi$ over $X=\{x,y,u\}$ and with the same $\iota: \ar(\psi) \to X$ as above. 
\begin{align*}
\Phi &= 2\phi(y,x) + 5\phi(y,u) + 6\phi'(u,x,u)  \\
\psi(\tuple{a}') &= \max_{\tuple{a} \in A^X, \ \tuple{a}\iota = \tuple{a}'} \Phi^{\bfa}(\tuple{a})
\end{align*}

Similarly to the crisp case, such definitions will be produced using canonical payoff formulas from~\cref{prop:canonical-payoff-baby}. We state a slightly generalized version with different thresholds on the right hand side of the inequalities. The proof is essentially the same -- just replace $c/s$ by $c'/s'$ on the right hand side of the system of inequalities $F\tuple{y} \leq \tuple{q}$.

\begin{proposition}[Canonical payoff formula with changed thresholds] \label{prop:canonical-payoff-baby-thresholds}
    Let $\ab$ be a pair of valued $\Sigma$-structures, $\mnn M = \polfeas \ab$, $c,c',s,s' \in \qq$, $N$ a finite set, $\alpha: N \to \qq$, and $\beta: \mnn M^{(N)} \to \qq$. 
    Suppose further that if $\bfa \leq c$, then $\alpha \leq c'$ (i.e., $\alpha(n)\leq c'$ for all $n \in N$). 
    Then the following are equivalent.         
    \begin{enumerate}[label=\textnormal{(\roman*)}]
        \item \label{item:1:canonical-payoff-baby-thresholds} For each $\IO \in \qqnonneg$ and each $\Omega \in \iopol^{(N)} \abcs$, 
        $\OmegaO[\beta]-s' \geq \IO (\OmegaI[\alpha]-c')$.  
        \item\label{item:2:canonical-payoff-baby-thresholds} There exists a payoff formula $\Phi$ over the set of variables $A^{N}$ such that
\begin{align*}
    & \forall n\in N  & \Phi^{\bfa}(\proj{N}{n}) -  c \,\weight (\Phi) & \geq \alpha(n) -c' \\
    & \forall f \in \mnn M^{(N)}  &\Phi^{\bfb}(f) - s \,\weight (\Phi) & \leq \beta(f) - s' \\
    &         &  \feas(\Phi^{\bfb}) &= \mnn M^{(N)}. 
\end{align*}
    \end{enumerate}
\end{proposition}

\begin{proposition} \label{prop:valued-definable}
    Let $\abcs$, $\abcsp$ be valued promise templates in signatures $\Sigma, \Sigma'$ such that $A=A'$ and $B=B'$ and such that if $\bfa \leq c$, then $\bfa' \leq c'$. The following are equivalent.
        \begin{enumerate}[label=\textnormal{(\roman*)}, ref=(\roman*)]
        \item\label{item:vdi} $\abcsp$ is valued pp-definable 
        from $\abcs$, that is,
     for every $\psi \in \Sigma'$ of arity $Z=\ar(\psi)$ there exists a payoff $\Sigma$-formula $\Phi$ over a set of variables $X$ and a mapping $\iota \colon Z \to X$  such that
    \begin{itemize}
        \item $\forall \tuple{a}' \in A^{Z} \quad 
        \psi^{\bfa'}(\tuple{a}') - c' \leq 
        \max_{\tuple{a} \in A^{X}, \, \tuple{a}\iota = \tuple{a}'} \Phi^{\bfa}(\tuple{a}) - c\,\weight(\Phi)$, and 
        \item $\forall \tuple{b}' \in B^{Z} \quad 
        \psi^{\bfb'}(\tuple{b}') - s'  \geq
        \max_{\tuple{b} \in B^X, \, \tuple{b}\iota = \tuple{b}'} \Phi^{\bfb}(\tuple{b}) - s\,\weight(\Phi)$.
    \end{itemize}
        \item\label{item:vdii} $\polfeas \ab \subseteq \polfeas \abp$ and, for every $\IO \in \qqnonneg$, each $\IO$-polymorphism of $\ab$ is a $\IO$-polymorphism of $\abp$ (where the undefined output probabilities are implicitly zero).
        In particular, $\Xi$ with probability one on the inclusion is a valued minion homomorphism from $\multipol \abcs$ to $\multipol \abcsp$.
 \end{enumerate}  
\end{proposition}

 \begin{proof}
    For the forward direction, we first observe that the inequalities imply that $(\feas(\bfa'),\feas(\bfb'))$ is pp-definable from $(\feas(\bfa),\feas(\bfb))$. Consequently, $\polfeas \ab \subseteq \polfeas \abp$.
    Let now $\Omega$ be an $N$-ary $\IO$-polymorphism of $\abcs$ and $(\psi,M') \in \Mat{\bfa'}{N}$. We need to show that $\OmegaO[\psi,M'] - s' \geq \IO (\OmegaI[\psi,M']-c')$.

Let $Z=\ar(\phi)$, $X$, $\iota$, $\Phi$ be as in the statement.
By the first item there exists a matrix $M \in A^{X \times N}$ such that, for all $n \in N$,  $\col{n}{M'} = \col{n}{M} \iota$ (equivalently $\rows(M') = \rows(M) \iota$) and 
$$
\psi^{\bfa'}(\col{n}{M'}) - c' \leq \Phi^{\bfa}(\col{n}{M}) - c\, \weight(\Phi),
$$
in particular, every column of $M$ is in $\feas(\bfa)$. Taking the expected value with $n \sim \OmegaI$, we obtain
$$
\OmegaI[\psi,M'] - c' \leq \Ex{n \sim \OmegaI} \Phi^{\bfa}(\col{n}{M}) - c \,\weight(\Phi).
$$

    For every $f \in \polfeas^{(N)} \ab$ we apply the second item to $\tuple{b} =  f \,\rows(M)$ (note that $\tuple{b}\iota = f\,\rows(M')$) and get $\psi^{\bfb'}(f \,\rows(M')) - s' \geq \Phi^{\bfb}(f \,\rows(M)) - s \,\weight(\Phi)$. Applying additionally \cref{prop:valued_combining_solutions}
   gives us
   
   \begin{align*}
   \OmegaO[\psi,M'] - s' 
   &= \Ex{f \sim \OmegaO} \psi^{\bfb'}(f \,\rows(M')) - s'
      \geq \Ex{f \sim \OmegaO} \Phi^{\bfb}(f \,\rows(M)) - s\, \weight(\Phi) 
   \\ & \geq \IO \left(\Ex{n \sim \OmegaI} \Phi^{\bfa}(\col{n}{M}) - c \,\weight(\Phi)\right)
     \geq \IO (\OmegaI[\psi,M'] - c').
   \end{align*}

For the backward direction, assume the inclusions in item \ref{item:vdii} and consider $\psi \in \Sigma'$ of arity $Z$. We need to find $\Phi$ and $\iota$ satisfying the two requirements for valued pp-definability in item~\ref{item:vdi}.

We define 
$$
\mnn M = \polfeas\ab, \quad N = \feas(\psi^{\bfa'}), \quad M = \GM{\feas(\psi^{\bfa'})} 
$$
and $\alpha \colon N \to \qq$, $\beta \colon \mnn M^{(N)} \to \qq$ by 
$$
\alpha(\ba')=\psi^{\bfa'}(\ba'), \quad
\beta(f)=\psi^{\bfb'}(f \, \rows(M)).
$$
Note that $\beta$ is correctly defined since every $f \in \mnn M$ is also in $\polfeas \abp$, so it gives a feasible tuple when applied to the rows of $M$. 
Also observe that item~\ref{item:1:canonical-payoff-baby-thresholds} of \cref{prop:canonical-payoff-baby-thresholds} is satisfied. Indeed, any $\Omega \in \iopol^{(N)} \abcs$ is also a $\IO$-polymorphism of $\abcsp$ by the assumption. Moreover, $\OmegaO[\beta] = \OmegaO[\psi,M]$ and $\OmegaI[\alpha]=\OmegaI[\psi,M]$, so the inequality follows.

Now we take the  formula $\Phi$ from item~\ref{item:2:canonical-payoff-baby-thresholds} in \cref{prop:canonical-payoff-baby-thresholds} over the set of variables $X := A^N$, and define $\iota = \rows(M)$. 
We need to verify the two inequalities required for pp-definability.
For the first one, let $\tuple{a}' \in A^Z$.
If $\tuple{a}'$ is not in $\feas(\psi^{\bfa'})$, then the inequality is clearly satisfied. Otherwise, we take $\tuple{a} = \proj{N}{\tuple{a}'}$ (so $\tuple{a}\iota = \tuple{a}'$) and apply the first property in \cref{prop:canonical-payoff-baby-thresholds}. We obtain $\Phi^{\bfa}(\tuple{a}) - c \,\weight(\Phi) \geq \alpha(\tuple{a}')-c' = \psi^{\bfa'}(\tuple{a}') - c'$, as required. 
For the second inequality, let $\tuple{b} \in B^X$. We  assume that $\tuple{b}$ is in $\feas(\Phi^{\bfb})$ (otherwise the inequality is trivial), which is equal to $\mnn{M}^{(N)}$ by the third property in~\cref{prop:canonical-payoff-baby-thresholds}. 
From the second property in~\cref{prop:canonical-payoff-baby-thresholds} we then conclude
$\Phi^{\bfb}(\tuple{b}) - s\,\weight(\Phi) \leq \beta(\tuple{b}) - s' = \psi^{\bfb'}(\tuple{b} \iota) - s'$, finishing the proof. 
 \end{proof}

\section{Linear equations} \label{app:hastad}

In this section we show that $\lin(1-\delta,1/2+\delta)$ satisfies the assumptions of~\cref{prop:red_from_label_cover,thm:homo_to_glc_sufficient_condition} for a suitable $\epsilon$ and all $D$ and $E$. The proof very much follows parts of \cite{DBLP:journals/jacm/Hastad01} with some small changes. 

It will be convenient to change the domain from $\{0,1\}$ to $\{-1,1\}$ and scale and shift the payoff functions, adjusting the completeness and soundness parameters, as follows.
We fix the template $(\bfa,\bfa,c,s)$, where $A = \{-1,1\}$, the signature consists of two ternary symbols $\phi_{-1}, \phi_{1}$ interpreted as $\phi^{\bfa}_i(a_1,a_2,a_3) = ia_1a_2a_3$ for $i=-1,1$, $c=1-2\delta$, $s=2\delta$. Note that indeed this template is just the template $\lin(1-\delta,1/2+\delta)$ re-scaled using the function $r \mapsto 2r-1$, with renamed elements $0 \mapsto 1$, $1 \mapsto -1$. We set the Gap Label Cover soundness parameter to $\epsilon = 16\delta^3$.
We fix $\mnn{M} = \polfeas \aaa$, so $\mnn{M}^{(N)}$ is the set of all $N$-ary \emph{Boolean} functions -- i.e., functions $f:A^N \to A$. 

The definition and analysis of $\Lambda$ and $\Phi$ from~\cref{prop:red_from_label_cover} require two ingredients that we discuss in the next two subsections.

\subsection{Folding}

A Boolean function $f \in \mnn{M}^{(N)}$ is \emph{folded} (also called self-dual) if it preserves the disequality relation, i.e.,  $f(-\tuple{a})=-f(\tuple{a})$ for every $\tuple{a} \in A^N$. We denote the function minion of folded Boolean functions by $\mnn F$. In order to ``fold'' a function we fix for each finite set $N$ a subset $A^N_{*}$ of $A^N$ that contains exactly one (arbitrarily chosen) element from $\{\tuple{a},-\tuple{a}\}$ for each $\tuple{a} \in A^N$. Let $\seltup{\tuple{a}}$ denote the selected element. We define a mapping $ \fold{ \ \cdot \ } :\mnn{M}^{(N)} \to \mnn{F}^{(N)}$ by $\fold{f}(\tuple{a}) = f(\tuple{a})$ if $\tuple{a} \in A^N_{*}$ and $\fold{f}(\tuple{a})=-f(-\tuple{a})$ otherwise; so only values of $f$ on tuples from $A^N_{*}$ are used to define $\fold{f}$.
Note that $\fold{f}$ is always folded and for folded functions we have $f=\fold{f}$. For instance, $\fold{\proj{N}{n}} = \proj{N}{n}$.

\subsection{Fourier Analysis of Boolean Functions.}

We will mostly follow the standard notation and refer the reader to~\cite{ODonnell14:book} for more details.

For a finite set $N$, let $\langle \cdot \, , \cdot\rangle$  denote the dot product in the vector space $\rr^{A^N}$ (which contains $\mnn{M}^{(N)}$) scaled as follows
$$
\langle f,g\rangle := \frac{1}{2^{|N|}} \sum_{\ba \in A^N} f(\ba)g(\ba).
$$ 
This vector space has an orthonormal basis consisting of Boolean functions $\rchi_I$, $I \subseteq N$, defined by
$$
\rchi_I(\ba):=\prod_{i\in I}\ba(i) \text{ for $I \neq \emptyset$, }
\rchi_{\emptyset}(\ba) := 1.
$$
Every function $f$ in $\rr^{A^N}$, in particular every $f \in \mnn M^{(N)}$, can thus be written in terms of the Fourier coefficients with respect to this basis:
$$
f=\sum_{I \subseteq N} \hat{f}_I \rchi_I, \quad
\mbox{ where } \hat{f}_I = \langle f, \rchi_I \rangle.
$$
This sum is referred to as the \emph{Fourier expansion} of $f$.

Here is a list of simple facts about the Fourier coefficients and basis vectors $\chi_I$. 
We use $I \symdif J$ to denote the symmetric difference of $I, J \subseteq N$ and, for $\pi: N \to N'$ and $I \subseteq N$, we define $\oddim{\pi}(I) = \{n' \in N' \mid |\pi^{-1}(n') \cap I| \mbox{ is odd}\}$.
\begin{enumerate}[label=(F\arabic*)]
    \item \label{it:fact1} $\sum_{I \subseteq D} \hat{f}^2_I =1$ for every $f \in \mnn{M}^{(N)}$ (from Parseval's identity since $f$ is $\{-1,1\}$-valued; see \cite[page 809]{DBLP:journals/jacm/Hastad01}).
    \item \label{it:fact2} $\hat{f}_{\emptyset} = 0$ for any folded $f \in \mnn{F}^{(N)}$  \cite[Lemma 2.32]{DBLP:journals/jacm/Hastad01}.
    \item \label{it:fact3} $\chi_I(\tuple{a}\tuple{b}) = \chi_I(\tuple{a}) \chi_I(\tuple{b})$ for every $I \subseteq N$ and $\tuple{a},\tuple{b} \in A^N$ \cite[Lemma 2.27]{DBLP:journals/jacm/Hastad01}.
    \item \label{it:fact4} $\chi_I(\tuple{a}) \chi_J(\tuple{a}) = \chi_{I \symdif J} (\tuple{a})$ for every $I,J \subseteq N$ and $\tuple{a} \in A^N$ \cite[Lemma 2.28]{DBLP:journals/jacm/Hastad01}.
    \item \label{it:fact5} $\ex{\tuple{a}} \ \chi_I(\tuple{a}) = 0$ for every $I \subseteq N$ unless $I = \emptyset$ in which case the expected value is 1. Here $\tuple{a}$ is selected uniformly from $A^N$ \cite[Lemma 2.29]{DBLP:journals/jacm/Hastad01}.
    \item \label{it:fact6}  $\chi_I(\tuple{a}\pi) = \chi_{\oddim{\pi}(I)}(\tuple{a})$ for every $\pi: N \to N'$, $I \subseteq N$, and $\tuple{a} \in A^{N'}$ \cite[Lemma 2.30]{DBLP:journals/jacm/Hastad01}.
\end{enumerate}

\subsection{Proof}

We first define the objects required in~\cref{prop:red_from_label_cover}. Recall that these are 
  \begin{itemize}
      \item mappings $\Lambda_D: \mnn{M}^{(D)} \to \Delta D$ and $\Lambda_E: \mnn{M}^{(E)} \to \Delta E$, and 
      \item for every $\pi: D \to E$ a normalized payoff formula $\Phi_{\pi}$ over the set of variables $A^{D} \cup A^{E}$, and 
      a linear nondecreasing function $\gamma_{\pi}: \rr \to \rr$ with $\gamma_{\pi}(s) \geq \epsilon$.
  \end{itemize}
For $f \in \mnn{M}^{(D)}$, the distribution $\Lambda_D(f)$ is sampled as follows.
\begin{itemize}
    \item Set $p = \fold{f}$.
    \item Select a random $I \subseteq D$ with probability $\hat{p}^2_I$. 
    \item Select a random $d \in I$ with the uniform probability.
\end{itemize}
Note that the $\hat{p}^2_I$ indeed define a probability distribution by \labelcref{it:fact1} and that the empty set is never picked in the second step by \labelcref{it:fact2} so the third step makes sense (and the probability that $\Lambda_D(f)=d$ is equal to $\sum_{d \in I \subseteq D} \hat{f}^2_I/|I|$). The mapping $\Lambda_E$ is defined analogously. 

For $\pi: D \to E$ the normalized payoff formula $\Phi = \Phi_{\pi}$, regarded as a probability distribution on constraints, is sampled as follows.
\begin{itemize}
    \item Choose $\tuple{a} \in A^E$ with the uniform probability.
    \item Choose $\tuple{b} \in A^D$ with the uniform probability.
    \item Choose $\tuple{\nu} \in A^D$ by setting, independently for each $d$, $\tuple{\nu}(d)=-1$ with probability $\delta$ and $\tuple{\nu}(d)=1$ with probability $1-\delta$ (this is the \emph{noise} vector).
    \item Return the constraint $\phi_i(\seltup{\tuple{a}},\seltup{\tuple{b}}, \seltup{((\tuple{a} \circ \pi)\tuple{b}\tuple{\nu})})$ (here juxtaposition denotes component-wise multiplication), where $i=1$ if the number of times $\tuple{c} \neq \seltup{\tuple{c}}$ for the three arguments $\tuple{c}$ of $\phi$ is even, and $i=-1$ otherwise; 
    in other words, $i$ is chosen so that 
    \begin{align} \label{eq:unfold}
    \phi_i^{\bfa}(h_E(\seltup{\tuple{a}}),&h_D(\seltup{\tuple{b}}),h_D(\seltup{((\tuple{a}\circ\pi)\tuple{b}\tuple{\nu})}))  \nonumber \\
     & =
     \phi_1^{\bfa}(\fold{h_E}(\tuple{a}),\fold{h_D}(\tuple{b}),\fold{h_D}((\tuple{a}\circ\pi)\tuple{b}\tuple{\nu}))
    \end{align} 
 for every assignment $(h_D,h_E): A^D \cup A^E \to A$.
\end{itemize}
Finally, let $\gamma=\gamma_{\pi}$ be the linear function tangent to the quadratic function $r \mapsto 4\delta r^2$ at $r = s=2\delta$, i.e., $\gamma(r) = 16\delta^2(r-\delta)$. Notice that $\gamma$ is increasing and $\gamma(s) = 16\delta^3 = \epsilon$.

We now need to verify the two conditions in~\cref{prop:red_from_label_cover}, namely the following. 
  \begin{enumerate}[label=\textnormal{\arabic*.}]
      \item $\Phi^{\bfa}(\proj{D}{d},\proj{E}{\pi(d)}) \geq 1-2\delta$ for every $d \in D$, and
      \item $\gamma(\Phi^{\bfa}(f_D,f_E)) \leq \ex{\substack{d \sim \Lambda_D(f_D)\\e\sim \Lambda_E(f_E)}} \pi(d,e)$ for every $f_D \in \mnn{M}^{(D)}$, $f_E \in \mnn{M}^{(E)}$.
  \end{enumerate}
The first one is easy:  $\Phi^{\bfa}(\proj{D}{d},\proj{E}{\pi(d)})$ is the expected value of $\phi^{\bfa}((\proj{D}{d},\proj{E}{\pi(d)}) \circ \tuple{x})$ when this constraint $\phi(\tuple{x})$ is selected according to the above distribution. By \labelcref{eq:unfold} and definitions, this number is equal to the expected value of
\arxiv{
\begin{align*}
\phi_1^{\bfa}(\proj{E}{\pi(d)}(\tuple{a}),\proj{D}{d}(\tuple{b}),\proj{D}{d}((\tuple{a}\circ\pi)\tuple{b}\tuple{\nu}))
&= \tuple{a}(\pi(d))\tuple{b}(d)(\tuple{a}\circ\pi)(d)\tuple{b}(d)\tuple{\nu}(d) \\
&= (\tuple{a}(\pi(d)))^2 (\tuple{b}(d))^2\tuple{\nu}(d)
= \tuple{\nu}(d),
\end{align*}
}\conference{
\begin{align*}
\phi_1^{\bfa}&(\proj{E}{\pi(d)}(\tuple{a}),\proj{D}{d}(\tuple{b}),\proj{D}{d}((\tuple{a}\circ\pi)\tuple{b}\tuple{\nu})) \\
& = \tuple{a}(\pi(d))\tuple{b}(d)(\tuple{a}\circ\pi)(d)\tuple{b}(d)\tuple{\nu}(d) \\
& = (\tuple{a}(\pi(d)))^2 (\tuple{b}(d))^2\tuple{\nu}(d)
= \tuple{\nu}(d),
\end{align*}
}which is $-1 \cdot \delta + 1 \cdot (1-\delta) = 1 - 2\delta$.

It remains to verify the second condition. Denote $p = \fold{f_D}$ and $q = \fold{f_E}$. Observe that $\Phi^{\bfa}(f_D,f_E)$, the expected value of $\phi^{\bfa}((f_D,f_E) \circ \tuple{x})$, is equal to the expected value of $\phi^{\bfa}((p,q) \circ \tuple{x})$ (since all variables in every constraint are from $A^D_* \cup A^E_*$ wherein $p,q$ coincide with $f_D,f_E$) which is, by \labelcref{eq:unfold},
$$
 \Phi^{\bfa}(f_D,f_E) = \Ex{\tuple{a},\tuple{b},\tuple{\nu}} \ q(\tuple{a}) p(\tuple{b}) p((\tuple{a} \circ \pi)\tuple{b}\tuple{\nu}).
$$
We replace $p$ and $q$ by their Fourier expansion and rewrite the latter expression using \labelcref{it:fact3}, linearity of expectation, and independence of the choices of $\tuple{a}$, $\tuple{b}$, $\tuple{\nu}$, as follows.
\begin{align*}
    \Ex{\tuple{a},\tuple{b},\tuple{\nu}} \ 
          &\left(\sum_{I \subseteq E} \hat{q}_I \rchi_I(\tuple{a})\right)
          \left(\sum_{J \subseteq D} \hat{p}_J \rchi_J(\tuple{b})\right)
          \left(\sum_{K \subseteq D} \hat{p}_K \rchi_K((\tuple{a} \circ \pi)\tuple{b}\tuple{\nu})\right) 
                \\
    &= \Ex{\tuple{a},\tuple{b},\tuple{\nu}}
       \sum_{I,J,K} 
          \hat{q}_I \hat{p}_J \hat{p}_K 
          \rchi_I(\tuple{a}) 
          \rchi_J(\tuple{b}) 
          \rchi_K(\tuple{a} \circ \pi) \rchi_K(\tuple{b}) \rchi_K(\tuple{\nu}) 
                \\
    &= \sum_{I,J,K}
       \hat{q}_I \hat{p}_J \hat{p}_K 
       \Ex{\tuple{a},\tuple{b},\tuple{\nu}}
          \rchi_I(\tuple{a}) 
          \rchi_J(\tuple{b}) 
          \rchi_K(\tuple{a} \circ \pi) \rchi_K(\tuple{b}) \rchi_K(\tuple{\nu}) 
               \\
   &= \sum_{I,J,K}               
       \hat{q}_I \hat{p}_J \hat{p}_K 
       \left(\Ex{\tuple{a}} \rchi_I(\tuple{a}) \rchi_K(\tuple{a} \circ \pi)\right)
       \left(\Ex{\tuple{b}}  \rchi_J(\tuple{b}) \rchi_K(\tuple{b}) \right)
       \left(\Ex{\tuple{\nu}}    \rchi_K(\tuple{\nu}) \right)
\end{align*}
The last parenthesis is, by definition of $\chi$, the independence of the choices of $\tuple{\nu}(k)$, and already observed $\ex{\tuple{\nu}} \ \tuple{\nu}(k)=1-2\delta$ 
$$
\Ex{\tuple{\nu}} \ \rchi_K(\tuple{\nu}) 
= \Ex{\tuple{\nu}} \ \prod_{k \in K} \tuple{\nu}(k)
= \prod_{k \in K} \Ex{\tuple{\nu}} \tuple{\nu}(k)
= (1-2\delta)^{|K|}.
$$
The second one is $\ex{} \chi_{J \symdif K}(\tuple{b})$ by \labelcref{it:fact4}, which is zero unless $J=K$ in which case it is one by \labelcref{it:fact5}. Finally, the first parenthesis is 
$\ex{} \rchi_I(\tuple{a}) \rchi_{\oddim{\pi}(K)}(\tuple{a})
= \ex{} \rchi_{I \symdif \oddim{\pi}(K)}(\tuple{a})$ by \labelcref{it:fact6} and \labelcref{it:fact4}, which is zero unless $I = \oddim{\pi}(K)$ in which case it is one by \labelcref{it:fact5}.
Altogether, we obtain
$$
\Phi^{\bfa}(f_D,f_E) = \sum_{K \subseteq D} {\hat{p}_K}^2 \hat{q}_{\oddim{\pi}(K)}  (1-2\delta)^{|K|}.
$$

The right hand side of item \labelcref{red_from_label_cover_ii} in \cref{prop:red_from_label_cover}, which is equal to the probability that $\pi(d)=e$ when $d$ and $e$ are selected according to $\Lambda_D(f_D)$ and $\Lambda_E(f_E)$ (respectively), can be crudely estimated as follows. For any $K \subseteq D$, the probability that $K$ and $\oddim{\pi}(K)$ are selected in the second step of sampling $\Lambda_D(f_D)$ and $\Lambda_E(f_E)$ is $\hat{p}^2_K \hat{q}^2_{\oddim{\pi}(K)}$. For each $e \in \oddim{\pi}(K)$ there exists $d \in K$  such that $\pi(d)=e$, so we succeed with probability at least $\hat{p}^2_K \hat{q}^2_{\oddim{\pi}(K)} |K|^{-1}$. Overall, the probability that $\pi(d)=e$ is thus at least the sum of these expressions over nonempty $K$ (as $\hat{p}_{\emptyset}=0$ by \labelcref{it:fact2}) and so is the expected value.

We can now finish the proof using \labelcref{it:fact1}, the Cauchy-Schwarz inequality, an auxiliary inequality $|K|^{-1/2} \geq 2\delta^{1/2}(1-2\delta)^{|K|}$ (inequality (16) in \cite{DBLP:journals/jacm/Hastad01}), and definition of $\gamma$ as follows.
\begin{align*} 
\Ex{\substack{d \sim \Lambda_D(f_D)\\e\sim \Lambda_E(f_E)}}  \ \pi(d,e) 
&\geq \sum_{K} {\hat{p}_K}^2 {\hat{q}_{\oddim{\pi}(K)}}^2|K|^{-1} \\
&= \left(\sum_{K} \left(\hat{p}_K \hat{q}_{\oddim{\pi}(K)}|K|^{-1/2}\right)^2\right)
   \left(\sum_K \hat{p}_K^2\right) \\
&\geq \left(\sum_K \hat{p}^2_K \hat{q}_{\oddim{\pi}(K)} |K|^{-1/2} \right)^2  \\
&\geq \left(\sum_K \hat{p}^2_K \hat{q}_{\oddim{\pi}(K)} 2 \delta^{1/2} (1-2\delta)^{|K|} \right)^2  \\
&= 4 \delta \left(\sum_K \hat{p}^2_K \hat{q}_{\oddim{\pi}(K)}  (1-2\delta)^{|K|} \right)^2 \\
&= 4 \delta (\Phi^{\bfa}(f_D,f_E))^2 \\
& \geq \gamma (\Phi^{\bfa}(f_D,f_E)).
\end{align*}

\section{Constant factor valued PCSPs} \label{app:kazda}

In this section we discuss the constant factor approximation version of the theory from \cite{Kazda21}, which was among the starting points of this work. The difference is that instead of fixing $c$ (completeness) and $s$ (soundness), we fix a constant factor $\kappa$, we let $c$ be a part of the input, and define the soundness as $s = \kappa c$. We also simplify the accepting inequality $\Phi^{\bfa}(h) \geq c w(\Phi)$ to $\Phi^{\bfa}(h) \geq c$ because $c$ is a part of the input anyway, and adjust the rejecting inequality accordingly.

The proofs are quite similar (and somewhat simpler) to the proofs in the previous section, therefore they are quite brief. 
Moreover, we abuse the terminology and name the concepts in the same way as in the main part of the paper. If the need arises in the future to simultaneously use both, then perhaps the adjective ``constant factor'' could be added, e.g., one could define a \emph{constant factor valued promise template}.

\begin{definition}[Valued PCSP]
A \textit{valued promise template} is a triple $\abt$ where
 \begin{itemize}
     \item $\bfa$, $\bfb$ are valued relational structures in the same signature $\Sigma$, and
     \item $\IO \in \qqpos$ is the \emph{constant factor}
 \end{itemize}
    such that $\max \Phi^{\bfb} \geq \IO \max \Phi^{\bfa}$ for every payoff $\Sigma$-formula $\Phi$.

Given a valued PCSP template $\abt$, the \emph{Promise Constraint Satisfaction
  Problem over $\abt$}, denoted by $\PCSP\abt$, is the following problem. 
\begin{enumerate}
    \item[\textsf{Input}]: A finite $\tau$-sorted set $X$, a payoff $\Sigma$-formula $\Phi$ over $X$, and $c \in \qq$ (the \emph{completeness}),
    \item[\textsf{Output}]: $\yes$ if $\exists h \ \Phi^{\bfa}(h) \geq c$; $\no$ if $\forall h \ \Phi^{\bfb}(h) < c\IO$.
\end{enumerate}
\end{definition}

We define polymorphisms in this context in exactly the same way as we previously defined $\IO$-polymorphism where $s = \kappa c$ (and $c$ is arbitrary). In other words,
$\Omega$ is a polymorphism if for each relation-matrix pair $(\phi,M)$, the point $(\OmegaI[\phi,M],\OmegaO[\phi,M]) \in \qq^2$  lies on or above the line with slope $\IO$ going through the origin $(0,0)$. The definition thus simplifies as follows.

\begin{definition}[Polymorphisms]
   Let $\ab$ be a pair of $\Sigma$-structures, $\mnn M = \polfeas \ab$, $c \in \qq$, and $\IO \in \qqpos$.  
   An $N$-ary weighting $\Omega$ of $\mnn M$ is
   a \emph{polymorphism} of $\abt$ if 
   $$
   \forall (\phi,M) \in \Mat{\bfa}{N} \ \ \OmegaO[\phi,M] \geq \IO \OmegaI[\phi,M].
   $$
   We denote by $\pol^{(N)} \abt$ the sets of all $N$-ary polymorphisms and $\pol \abt = (\pol^{(N)} \abt)_{N \in \FinSet}$. 
\end{definition}

An important technical tool is again the canonical payoff formulas. The version for constant factor is somewhat similar to the simpler canonical formula in \cref{prop:canonical-payoff-baby}. The main difference is that we allow a shift, which we can afford because the completeness parameter is now a part of the input to a PCSP.

\begin{proposition}[Canonical payoff formula] \label{prop:canonical-payoff-kazda}
    Let $\ab$ be a pair of $\Sigma$-structures, $\mnn M = \polfeas \ab$, $\IO \in \qqpos$, $N$ a finite set, $\alpha: N \to \qq$, and $\beta: \mnn M^{(N)} \to \qq$. The following are equivalent.
    \begin{enumerate}[label=\textnormal{(\roman*)}]
        \item \label{item:1:canonical-kazda} For each $\Omega \in \pol^{(N)} \ab$, 
        $\OmegaO[\beta] \geq \IO \OmegaI[\alpha]$. 
        \item \label{item:2:canonical-kazda} There exists a payoff formula $\Phi$ over the set of variables $A^{N}$ and $\delta \in \qq$ (the \emph{shift}) such that 
      \begin{align*}
          \forall n\in N \quad\quad \Phi^{\bfa}(\proj{N}{n}) &\geq  \alpha(n) +  \delta \\
          \forall f \in \mnn M^{(N)} \quad\quad \Phi^{\bfb}(f)  &\leq \beta(f) + \IO \delta \\
                     \feas(\Phi^{\bfb}) &= \mnn M^{(N)}. 
\end{align*}
    \end{enumerate}
\end{proposition}

\begin{proof}
   As in the previous canonical proofs we start by observing that  
       item \labelcref{item:2:canonical-kazda} is equivalent to the following condition.
   \begin{enumerate}[label=\textnormal{(\roman*)}] \setcounter{enumi}{2}
    \item\label{item:3:canonical-kazda} There exist $w_{\phi,M} \in \qqnonneg$ and $\delta \in \qq$ such that the payoff formula
    \begin{equation*} 
         \Phi = \sum_{(\phi,M) \in \Mat{\bfa}{N}} w_{\phi,M} \ \phi(\rows(M))
    \end{equation*}
    satisfies all the inequalities.  
   \end{enumerate}
   This condition is equivalent to the following system of linear inequalities with nonnegative rational unknowns $w_{\phi,M}, \delta^+, \delta^-$ and rational coefficients.
      \begin{align*}
          \forall n \in N \quad \quad
          \sum_{(\phi,M)}
          -\phi^{\bfa}(\col{n}{M}) w_{\phi,M} \ + \ \delta^+ - \delta^-
           &\leq - \alpha(n) 
          \\
          \forall f \in \mnn M^{(N)}\quad \quad
          \sum_{(\phi,M)}
           \phi^{\bfb}(f \rows(M))  w_{\phi,M} \ - \ \IO \delta^+ + \IO\delta^- 
           &\leq \beta(f) .
\end{align*}
This system $F \tuple{y} \leq \tuple{q}$ is by \cref{thm:Farkas} equivalent to
\begin{align*} 
\forall \tuple{x} \in (\qqnonneg)^{N \cup \mnn M^{(N)}} (F^T\tuple{x} \geq 0 \implies  \tuple{q}^T\tuple{x} \geq 0).
\end{align*}
Writing $\tuple{x}$ as 
 $(\IOI \OmegaI,\IOO\OmegaO)$, where $\IOI, \IOO \in \qqnonneg$, $\OmegaI \in \Delta N$, and $\OmegaO \in \Delta \mnn M^{(N)}$, we obtain an equivalent condition
\begin{align*} 
\forall \IOI, &\IOO \in \qqnonneg \ 
\forall \Omega \textnormal{ $N$-ary weighting of $\mnn M$} \\ 
   &\Big(\big(\forall (\phi,M) \in \Mat{\bfa}{N} \quad \ \IOO \OmegaO[\phi,M] \geq \IOI \OmegaI[\phi,M] \big)\\
  & \ \wedge \ (\IOI=\IO \IOO) \Big) \implies
    \IOO \OmegaO[\beta] \geq \IOI \OmegaI[\alpha], 
\end{align*}
which is equivalent to \labelcref{item:1:canonical-kazda}.
\end{proof}

Constant factor versions of \cref{prop:valued_combining_solutions} and \cref{prop:valued-char-of-templates} are as follows.

\begin{proposition}[Polymorphisms and payoffs] \label{prop:valued_combining_solutions_kazda}
Let $\ab$ be a pair of valued $\Sigma$-structures,  $N$ a finite set, $\IO \in \qqpos$, and $\Omega$ an $N$-ary polymorphism of $\abt$. 

Then for every finite set $X$, every payoff $\Sigma$-formula $\Phi$ over $X$, and every $M \in A^{X \times N}$ such that every column is in $\feas(\Phi^{\bfa})$ we have that
          $$
          \Ex{f \sim \OmegaO} \ \Phi^{\bfb}(f \rows(M))  \geq \IO \Ex{n \sim \OmegaI} \ \Phi^{\bfa}(\col{n}{M}).           
          $$
\end{proposition}

\begin{proof}
   We define a probability distribution $\mu$ on relation-matrix pairs as in the proof of \cref{prop:valued_combining_solutions} and the same calculations give us
   \begin{align*}
   \Ex{f \in \OmegaO} \Phi^{\bfb}(f\rows(M)) &= \Ex{(\phi,M') \sim \mu} \OmegaO[\phi,M'] \\ &\geq \IO \Ex{(\phi,M') \sim \mu} \OmegaI[\phi,M'] \\ &= \IO \Ex{n \sim \OmegaI} \ \Phi^{\bfa}(\col{n}{M}).\qedhere
   \end{align*} 
\end{proof}

\begin{proposition}[Characterization of templates]
    Let $\ab$ be a pair of $\Sigma$-structures and  $\IO \in \qq$. The following are equivalent.
    \begin{enumerate}[label=\textnormal{(\roman*)}]
        \item \label{item:1:ValuedCharTemplKazda} $\abt$ is a valued promise template.
        \item \label{item:2:ValuedCharTemplKazda} For each payoff formula $\Phi$ over the set of variables $A$, 
        $\max \Phi^{\bfb} \geq \IO \Phi^{\bfa}(\id_A)$.
        \item \label{item:3:ValuedCharTemplKazda} There exists a unary polymorphism of $\abt$.
    \end{enumerate}
\end{proposition}

\begin{proof}
  The implication from \labelcref{item:3:ValuedCharTemplKazda} to
  \labelcref{item:1:ValuedCharTemplKazda} follows from~\cref{prop:valued_combining_solutions_kazda} for a singleton set $N$ and
  the implication from \labelcref{item:1:ValuedCharTemplKazda} to \labelcref{item:2:ValuedCharTemplKazda} is trivial.
  The implication from \labelcref{item:2:ValuedCharTemplKazda} to \labelcref{item:3:ValuedCharTemplKazda} is obtained by applying \cref{prop:canonical-payoff-kazda} with $N = \{n\}$, $\alpha(n)=0$, and $\beta(f) < 0$ for each $f$.
\end{proof}

We now get to valued minions and homomorphisms. Even though it is clearer in this case what closure properties we should require on the sets of polymorphisms~\cite{Kazda21}, 
we do not need these for our results and take the most liberal definition anyway.
\begin{definition}[Valued minion]
Let $\mnn M$ be a minion. A \emph{valued minion over $\mnn M$} is a collection $\vmnn M = (\vmnn M^{(N)})_{N \in \FinSet}$ where each $\vmnn M^{(N)}$ is a set of $N$-ary weightings of $\mnn M$. 
\end{definition}

\begin{definition}[Valued minion homomorphisms] \label{def:valued_minion_homo-kazda}
    Let $\vmnn M$, $\vmnn M'$ be valued minions over  minions $\mnn M$ and $\mnn M'$, respectively. A \emph{valued minion homomorphism} $\vmnn M \to \vmnn M'$ is a probability distribution $\Xi$ on the set of minion homomorphisms $\mnn M \to \mnn M'$ such that for every finite set $N$ and every $(\OmegaI,\OmegaO) \in \vmnn M^{(N)}$, the weighting    $\Xi(\Omega) := (\OmegaI, \Xi(\OmegaO))$ is in $\vmnn M'^{(N)}$.
\end{definition}

\begin{theorem}[Reductions via valued minion homomorphism] \label{thm:main-kazda}
    Let $\abt$, and $\abtp$ be valued promise templates. If there is a valued minion homomorphism from $\pol\abt$ to $\pol\abtp$, then $\PCSP\abtp \leq \PCSP\abt$.     
\end{theorem}

The corresponding VMC problem for constant factor approximation differs from our previous setting in that the completeness is a part of the input (so it also does not make much sense to shift the payoffs) and the condition on $\alpha$ and $\beta$ simplifies.

\begin{definition}[Valued Minor Condition Problem]
Given a  minion $\mnn M$, a 
valued minion $\vmnn M$ over $\mnn M$, and
  an integer $k$, the Valued Minor Condition Problem for $\mnn M$, $\vmnn M$,
  and $k$, denoted by $\vmc(\mnn M,\vmnn M,k)$ is the following problem

\begin{enumerate}
    \item[\textsf{Input}]  \begin{enumerate}[label=\textnormal{\arabic*.}]
     \item disjoint sets $U$ and $V$ (the sets of \emph{variables}),
     \item a set $D_x$ with $|D_x| \leq k$ for every $x \in U \cup V$ (the \emph{domain} of $x$),
     \item a set of formal expressions of the form $\pi(u) = v$, where $u \in U$, $v \in V$, and $\pi: D_u \to D_v$ (the \emph{minor conditions}), 
     \item for each $u \in U$, a pair of functions $\alpha_u:D_u \to \qq$, $\beta_u:\mnn M^{(D_u)} \to \qq$ (the \emph{input and output payoff functions}) which satisfy the following condition. 
    \begin{equation*} \tag{$\ast$}\label{star-condition-kazda}  
         \forall \Omega \in \vmnn M^{(D_u)} \quad 
         \OmegaO[\beta_u] \geq \IO \OmegaI[\alpha_u].
    \end{equation*}
     \item $c \in \qq$ (the \emph{completeness})
    \end{enumerate}
    \item[\textsf{Output}]    \begin{enumerate}[label=\textnormal{(\roman*)}]
        \item[\yes]if there exists a function $h$ from $U \cup V$ with $h(x) \in D_x$ (for each $x \in U \cup V$) such that, for each minor condition $\pi(u) = v$, we have $\pi(h(u)) = h(v)$, and $\sum_{u \in U} \alpha_u(h(u)) \geq c$.
        
        \item[\no]if there does not exist a function $h$ from $U \cup V$ with $h(x) \in \mnn M^{(D_x)}$  such that, for each minor condition $\pi(u) = v$, we have $\mnn M^{(\pi)}(h(u)) = h(v)$, and $\sum_{u \in U} \beta_u(h(u)) \geq c\IO$.
    \end{enumerate}
\end{enumerate}
\end{definition}

Finally, we give the three reductions needed to prove \cref{thm:main-kazda}.

\begin{proposition}[Between VMCs]
  Let $\vmnn M$, $\vmnn M'$ be valued minions over minions $\mnn M$ and $\mnn M'$, respectively, such that there exists a valued minion homomorphism $\vmnn M \to \vmnn M'$. Then $\vmc (\mnn M', \vmnn M', k) \leq \vmc (\mnn M, \vmnn M, k)$ for any positive integer $k$. 
\end{proposition}

\begin{proof}
    The resulting instance is unchanged except for changing $\beta$ in the same way as in the proof of \cref{prop:BetweenVMCs}. The proof is also almost the same and we omit the details. 
\end{proof}

\begin{proposition}[From PCSP to VMC]\label{prop:PCSPleqVMC-kazda}
    Let $\abt$ be a valued promise template, $\mnn M = \polfeas \ab$, and $\vmnn M = \pol \abt$. If $k$ is a sufficiently large integer, then
    $\PCSP \abt \leq \vmc( \mnn M, \vmnn M, k)$.
\end{proposition}

\begin{proof}
From a payoff formula $\Phi = \sum_{i \in I} w_i \phi_i(\tuple{x}_i)$ over $X$ 
we create an instance of $\vmc(\mnn M, \vmnn M, k)$ analogously to the proof of \cref{prop:PCSPleqVMC} but without shifting, that is, as follows. 
  \begin{enumerate}[label=\textnormal{\arabic*.}]
      \item $U = I$, $V = X$.
      \item $D_i = \feas(\phi_i^{\bfa})$, $D_x = A_{\sort(x)}$ for each $i \in I$, $x \in X$.
      \item For each $i \in I$  and $z \in \ar(\phi_i)$, we introduce the constraint $\pi_{i,z}(i) = \tuple{x}_i(z)$, where $\pi_{i,z}$ is the domain-codomain  restriction of $\proj{\ar(\phi_i)}{z}$ to $\feas(\phi_i^{\bfa})$ and $A_{\sort(z)}$.
\item $\alpha_i(\tuple{a}) = w_i\phi_i^{\bfa}(\tuple{a})$,  
$\beta_i(f) = w_i (\phi_i^{\bfb}(f \ \rows(\GM{\feas(\phi_i^{\bfa})}))$ 
for each $i\in I$, $\tuple{a} \in \feas(\phi_i^{\bfa})$, and $f \in \mnn M^{(\feas(\phi_i^{\bfa}))}$. 
\item The completeness parameter is unchanged.
  \end{enumerate}    

Condition \labelcref{star-condition-kazda} is verified in the same way as \labelcref{stronger-star-condition} in the proof of \cref{prop:PCSPleqVMC} (with fixed $\IO$ and without shifting by $c$ and $s$), and the completeness and soundness are almost the same as well.
\end{proof}

\begin{proposition}[From VMC to PCSP] \label{prop:VMCleqPCSP-kazda}
    Let $\abt$ be a valued promise template, $\mnn M = \polfeas \ab$, and $\vmnn M = \pol \abt$.  For any positive integer $k$, 
    $\vmc( \mnn M, \vmnn M, k) \leq \PCSP \abt$.
\end{proposition}

\begin{proof}
  By virtue of \labelcref{star-condition-kazda}, we can find a collection of payoff formulas $(\Phi_u)_{u \in U}$ and rationals $\delta_u \geq 0$ as in item \labelcref{item:2:canonical-kazda} of \cref{prop:canonical-payoff-kazda} (for $\Phi_u$ we use $\alpha=\alpha_u$ and $\beta=\beta_u$). We create $\Psi$ in the same way as in the proof of \cref{prop:VMCleqPCSP} and set the PVCSP completeness parameter to $c' = c + \sum_{u \in U} \delta_u$. 

  If a function $h$ witnesses that
  the $\vmc$ instance is a $\yes$ instance, then the assignment $h'$ corresponding to $(\proj{D_x}{h(x)})_{x \in X}$ satisfies the minor conditions  and 
  \begin{align*}
  \Psi^{\bfa} (h') 
  &= \sum_{u \in U} \Phi_u^{\bfa}(\proj{D_u}{h(u)})
  \geq  \sum_{u \in U} (\alpha_u(h(u)) + \delta_u) \\
  &\geq c + \sum_{u \in U} \delta_u = c'.
  \end{align*}
  On the other hand, if $\Psi^{\bfb}(h') \geq \IO c'$, then the corresponding $(h_x)_{x \in X}$ consists of polymorphism that satisfy the minor conditions and
  \begin{align*}
  \sum_{u\in U} \beta_u(h_u) 
  &\geq \sum_{u \in U} (\Phi^{\bfb}_u(h_u) - \IO \delta_u)
  \\ &= \Psi^{\bfb}(h') - \IO \sum_{u \in U} \delta_u 
  \geq \IO c' - \IO \sum_{u \in U} \delta_u \geq \IO c.
\qedhere  \end{align*}
\end{proof}